\documentclass[12pt]{article}
\usepackage{amssymb}
\usepackage{amsmath}
\usepackage{amsfonts}
\usepackage{amsthm}
\usepackage{graphicx} 
\usepackage{epstopdf}
\usepackage{color}
\usepackage[onehalfspacing]{setspace}
\usepackage{hyperref}
\usepackage{multirow}
\usepackage{float}
\usepackage{bm}
\usepackage{bbm}
\usepackage{mathabx}
\usepackage{booktabs}
\usepackage{enumitem}
\usepackage{url}
\usepackage{algorithm}
\usepackage{algorithmic}

\usepackage[flushleft]{threeparttable}

\usepackage[font=small]{caption}

\def\E{\mathbb{E}}

\numberwithin{equation}{section}

\newtheorem{theorem}{Theorem}[section]
\newtheorem{lemma}{Lemma}[section]
\newtheorem{corollary}{Corollary}[section]
\newtheorem{remark}{Remark}[section]

\newtheorem{proposition}{Proposition}[section]

\theoremstyle{definition}
\newtheorem{assumption}{Assumption}

\theoremstyle{definition}
\newtheorem{example}{Example}[section]

\theoremstyle{definition}

\newcommand{\indep}{\raisebox{0.05em}{\rotatebox[origin=c]{90}{$\models$}}}

\usepackage{xr}
\usepackage{geometry}
\usepackage{color, hyperref}
\usepackage{natbib}

\title{
Double Robust Bayesian Inference on Average Treatment Effects\thanks{We thank the anonymous reviewers, as well as 
Xiaohong Chen,  Yanqin Fan, Essie Maasoumi, Yichong Zhang, and numerous seminar and conference participants for helpful comments and illuminating discussions. Breunig gratefully acknowledges the support  of the Deutsche Forschungsgemeinschaft (DFG, German Research Foundation) under Germany’s Excellence Strategy – EXC-2047/1 – 390685813. Yu gratefully acknowledges the support of JSPS KAKENHI Grant Number 21K01419.
}}
		
\author{Christoph Breunig\thanks{Department of Economics, University of Bonn. Email: \url{cbreunig@uni-bonn.de}}\and Ruixuan Liu\thanks{CUHK Business School, Chinese University of Hong Kong. Email: \url{ruixuanliu@cuhk.edu.hk}}\and Zhengfei Yu\thanks{Faculty of Humanities and Social Sciences, University of Tsukuba. Email: \url{yu.zhengfei.gn@u.tsukuba.ac.jp}}}

\begin{document}
  \maketitle
	\begin{abstract}
	\vskip -.1cm
	{\small 
	We propose a double robust Bayesian inference procedure on the average treatment effect (ATE) under unconfoundedness. 
For our new Bayesian approach, we first adjust the prior distributions of the conditional mean functions, and then correct the posterior distribution of the resulting ATE.
	Both adjustments make use of pilot estimators motivated by the semiparametric influence function for ATE estimation.
	We prove asymptotic equivalence of our Bayesian procedure and efficient frequentist  ATE estimators by establishing a new semiparametric Bernstein-von Mises theorem under double robustness; i.e., the lack of smoothness of conditional mean functions can be compensated by high regularity of the propensity score and vice versa.
		Consequently, the resulting Bayesian credible sets form confidence intervals with asymptotically exact coverage probability.
		In simulations, our method provides precise point estimates of the ATE through the posterior mean and credible intervals that closely align with the  nominal coverage probability. Furthermore, our approach achieves a shorter interval length in comparison to existing methods. We illustrate our method in an application to the National Supported Work Demonstration following \cite{lalonde1986} and \cite{dehejia1999causal}.
}
	\end{abstract}		
	\noindent{\footnotesize \noindent \textsc{Keywords:}  Average treatment effects, unconfoundedness, double robustness, nonparametric Bayesian inference, Bernstein–von Mises theorem, Gaussian processes. }	

\section{Introduction}
This paper proposes a double robust Bayesian approach for estimating the average treatment effect (ATE) under unconfoundedness, given a set of pretreatment covariates. Our new Bayesian procedure involves both prior and posterior adjustments. First, following \cite{ray2020causal}, we adjust the prior distributions of the conditional mean function using an estimator of the propensity score. Second, we use this propensity score estimator together with a pilot estimator of the conditional mean to correct the posterior distribution of the ATE. 
The adjustments in both steps are closely related to the functional form of the semiparametric influence function for ATE estimation under unconfoundedness. They do not only shift the center but also change the shape of the posterior distribution.
For our robust Bayesian procedure, we derive a new Bernstein–von Mises (BvM) theorem, 
which means that this posterior distribution, when centered at any efficient estimator, is asymptotically normal with the efficient variance in the semiparametric sense. The key innovation of our paper is that this result holds under double robust smoothness assumptions within the Bayesian framework.

Despite the recent success of Bayesian methods, the literature on ATE estimation is predominantly frequentist-based.
 For the missing data problem specifically, it was shown that conventional Bayesian approaches (i.e., using uncorrected priors) can produce inconsistent estimates, unless some unnecessarily strong smoothness conditions on the underlying functions were imposed; see the results and discussion in \cite{robins1997} or \cite{ritov2014}. 
 Once the prior distribution was adjusted using some pre-estimated propensity score, \cite{ray2020causal} recently established a novel semiparametric BvM theorem under weaker smoothness requirement for the propensity score function.\footnote{Strictly speaking, the main objective in \cite{ray2020causal} concerns the mean response in a missing data model, which is equivalent to observing one arm (either the treatment or control) of the causal setup. } 
However, a minimum differentiability of order $p/2$ is still required for the conditional mean function in the outcome equation, where $p$ denotes the dimensionality of covariates.
 In this paper, we are interested in Bayesian inference under double robustness  that allows for a trade-off between the required levels of smoothness in the propensity score and the conditional mean functions.

Under double robust smoothness conditions, 
we show that Bayesian methods, which use propensity score adjusted priors as in \cite{ray2020causal}, satisfy the BvM Theorem only up to a ``bias term'' depending on the unknown true conditional mean and propensity score functions.
In this paper, our robust Bayesian approach accounts for this bias term in the BvM Theorem by considering an explicit posterior correction.
Both the prior adjustment and the posterior correction are based on functional forms that are closely related to the efficient influence function for the ATE, see \cite{hahn1998role}. 
We show that the corrected posterior satisfies the BvM Theorem under double robust smoothness assumptions. 
Our novel procedure combines the advantages of Bayesian methodology with the robustness features that are the strengths of frequentist procedures. Our credible intervals are Bayesianly justifiable in the sense of \cite{rubin1984bayesian}, as the uncertainty quantification is made conditional on the observed data and can be also interpreted as frequentist confidence intervals with asymptotically exact coverage probability.
Our procedure is inspired by insights from the double machine learning (DML) literature, as well as the bias-corrected matching approach from \cite{abadie2011bias}, as our robustification of an initial procedure removes some non-negligible bias and remains asymptotically valid under weaker regularity conditions.
 While the main part of our theoretical analysis focuses on the ATE of binary outcomes, also considered by \cite{ray2020causal}, we outline extensions of our methodology to continuous and multinomial cases, as well as to other causal parameters.

In both simulations and an empirical illustration using the National Supported Work Demonstration data, we provide evidence that our procedure performs well compared to existing Bayesian and frequentist approaches. In our Monte Carlo simulations, we find that our method results in improved empirical coverage probabilities, while maintaining very competitive lengths for confidence intervals. This finite sample advantage is also observed over Bayesian methods that rely solely on prior corrections. In particular, we note that our approach leads to more accurate uncertainty quantification and is less sensitive to estimated propensity scores being close to boundary values.
	 
	 	The BvM theorem for parametric Bayesian models is well-established; see, for instance, \cite{van1998asymptotic}. Its semiparametric version is still being studied very actively when nonparametric priors are used \citep{castillo2012gaussian,castillo2015bvm,ray2020causal}. To the best of our knowledge, our new semiparametric BvM  theorem is the first one that possesses the double robustness property. 
Our paper is also connected to another active research area concerning Bayesian inference for parameters in econometric models, which is robust to partial or weak identification \citep{chen2018MC,giacomini2020robust,andrews2022gmm}.
	 	The framework and the approach we take is different. Nonetheless, they share the same scope of tailoring the Bayesian inference procedure to new challenges in contemporary econometrics.

\section{Setup and Implementation}\label{sec:model}
This section provides the main setup of the average treatment effect (ATE). We motivate the new Bayesian methodology and detail the practical implementation.

	\subsection{Setup}
We consider a family of probability distributions $\{P_\eta:\eta\in\mathcal H\}$ for some parameter space $\mathcal H$, where the (possibly infinite dimensional) parameter $\eta$ characterizes the probability model. Let $\eta_0$ be the true value of the parameter and denote $P_0=P_{\eta_0}$, which corresponds to the frequentist distribution generating the observed data.

For individual $i$, consider a treatment indicator $D_i\in\{0,1\}$.
The observed outcome $Y_i$ is determined by 
$Y_i =D_i Y_i(1)  +(1 - D_i)Y_i(0)$ where $(Y_i(1), Y_i(0))$ are the potential outcomes of individual $i$ associated with  $D_i=1$ or $0$.  We now focus on the binary outcome case where both $Y_i(1)$ and $Y_i(0)$ take values in $\{0,1\}$. An extension to multinomial or continuous outcomes is provided in Section \ref{sec:extension}.
The covariates for individual $i$ are denoted by $X_i$, a vector of dimension $p$, with the distribution $F_0$ and the density $f_0$.\footnote{If $X_i$ does not have a density we can simply consider the conditional density of $(Y_i, D_i)$ given $X_i=x$ instead of the joint density of $(Y_i, D_i, X_i)$.}
Let $\pi_0(x)=P_0(D_i=1|X_i=x)$ denote the propensity score and $m_0(d,x)= P_0(Y_i=1|D_i=d, X_i=x)$ the conditional mean. 
Suppose that the researcher observes independent and identically distributed (i.i.d.) observations of $Z_i=(Y_i,D_i,X_i^\top)^\top$ for $i=1,\dots,n$. The joint density of $Z_i$ is given by $p_{\pi_0,m_0,f_0}$ where
	\begin{equation}\label{x_den}
	p_{\pi,m,f}(z)=\pi(x)^d (1-\pi(x))^{1-d}m(d,x)^{y} (1-m(d,x))^{(1-y)}f(x).
	\end{equation}
The parameter of interest is the ATE given by $\tau_0=\mathbb{E}_0[Y_i(1)-Y_i(0)]$,  where $\mathbb{E}_0[\cdot]$ denotes the expectation under $P_0$. For its identification, we impose the following standard assumption of unconfoundedness and overlap \citep{rosenbaum1984,imbens2004,imbens2015causal}.
 		\begin{assumption}\label{Ass:unconfounded}
		(i)	 $(Y_i(0),Y_i(1)) ~ \indep ~ D_i \mid X_i$ and (ii) there exists $\bar\pi>0$ such that $\bar\pi <\pi_0(x)< 1-\bar\pi$ for all $x$ in the support of $F_0$. 
		\end{assumption}
We introduce additional notations from the Bayesian perspective, following the similar setup from \cite{ray2020causal}. For the purpose of assigning prior distributions to $(\pi, m)$ in the Bayesian procedure, it is convenient to transform them by a link function. We make use of the Logistic function $\Psi(t)=1/(1+e^{-t})$ here. Specifically, we consider the reparametrization of $(\pi, m, f)$ given by $\eta=(\eta^{\pi},\eta^m,\eta^f)$. 
We index the probability model as $P_{\eta}$, in line with the notation introduced at the first paragraph of this section, where
\begin{eqnarray}\label{repar}
	\eta^{\pi} = \Psi^{-1}(\pi), ~~\eta^{m} = \Psi^{-1}(m),~~\eta^f =\log f. 
\end{eqnarray}
	Below, we write $m_\eta=\Psi(\eta^m)$,  $\pi_\eta=\Psi(\eta^\pi)$, and $f_\eta=\exp(\eta^f)$ to make the dependence on $\eta$ explicit. Given any prior on the triplet
$(\eta ^{\pi},\eta ^{m},\eta ^{f})$, Bayesian inference on the ATE is
achieved by deriving the posterior distribution of
	 \begin{equation}\label{ate}
\tau_\eta=\mathbb E_\eta\left[m_\eta(1,X)-m_\eta(0,X)\right],
\end{equation}
where $\mathbb E_\eta[\cdot]$ denotes the expectation under $P_\eta$. Our aim is to examine large-sample behavior of the posterior of $\tau_{\eta}$ under the true probability distribution $P_0$. In the same vein, the true parameter of interest becomes $\tau_0=\tau_{\eta_0}$.

The construction of our double robust Bayesian procedure in Section \ref{sec:method_outline} has fundamental connection to the efficient influence function. For any parameter $\eta$, the efficient influence function (\cite{hahn1998role,hirano2003efficient}) is
\begin{align}\label{eif_ate}
	\widetilde{\tau}_{\eta}(z) &= m_\eta(1,x)-m_\eta(0,x)+ \gamma_\eta(d,x)(y-m_\eta(d,x))-\tau_\eta
\end{align} 
for the Riesz representer $\gamma_\eta$, which is given by
\begin{equation}\label{riesz:def}
	\gamma_\eta(d,x)=\frac{d}{\pi_\eta(x)}-\frac{1-d}{1-\pi_\eta(x)}.
\end{equation}
We write $\widetilde{\tau}_{0}=\widetilde{\tau}_{\eta_0}$ and $\gamma_{0}=\gamma_{\eta_0}$. Both the prior adjustment and posterior correction of our approach require a pilot estimator for $\gamma_{0}$. Under Assumption \ref{Ass:unconfounded}, the true Riesz representer $\gamma_0$ is well defined.
\subsection{Double Robust Bayesian Point Estimators and Credible Sets}\label{sec:method_outline}				
We build upon the ATE expression in \eqref{ate} to develop our doubly robust inference procedure. Our approach 
is based on nonparametric prior processes for $\eta^m$ and $\eta^f$. For the latter, we consider the Dirichlet process, which is a default prior on spaces of probability measures. This choice is also convenient for posterior computation via the Bayesian bootstrap; see Remark \ref{rem:BB}. For the former, we make use of Gaussian process priors, along with an adjustment that involves a preliminary estimator of $\gamma_0$.
 Gaussian process priors are also closely related to spline smoothing, as discussed in \cite{wahba1990spline}. Their posterior contraction properties (see \cite{ghosal2017fundamentals}), together with excellent finite sample behavior (see \cite{rassmusen2006gaussian}), make Gaussian process priors popular in the related literature. 
Since $\tau_{\eta}$ does not depend on $\eta^{\pi}$, the specification of a prior on the propensity score is not required.

We consider pilot estimators $\widehat\pi$ of the propensity score $\pi_0$ and  $\widehat m$ of the conditional mean function $m_0$, which both are based on an auxiliary sample.
 We consider a plug-in estimator for the Riesz representer $\gamma_0$ given by
\begin{align}\label{Riesz_est}
\widehat{\gamma}(d,x)=\frac{d}{\widehat\pi(x)}-\frac{1-d}{1-\widehat\pi(x)}.
\end{align}
Below, let $\Gamma_n$ denote the sample average of the absolute value of $\widehat{\gamma}$, which we use for scale normalization in our prior adjustment (see Section \ref{sec:implement:GP} for details).
 The use of an auxiliary data for pilot estimators simplifies the technical analysis related to the propensity score adjusted priors; see \cite{ray2020causal}. Also, it provides an effective way to control some negligible higher-order terms, see our Lemma \ref{lemma:negligible} in the Supplemental Material and the related discussion about the sample splitting in the DML type methods on Page C6 of \cite{chernozhukov2018double}. In practice, we use the full data twice and do not split the sample, as we have not observed any over-fitting or loss of coverage thereby.
  Algorithm \ref{algorithm} describes our double robust Bayesian inference procedure.

\begin{algorithm}[H]
\caption{Double Robust Bayesian Procedure}
\label{algorithm}
\begin{algorithmic}
    \STATE \textbf{Input:} Data $Z_i=(Y_i,D_i,X_i^\top)^\top$ for $i=1,\dots,n$, number of posterior draws $S$,  initial estimators $\widehat \gamma$ and $\widehat m$,
   and  $\lambda \sim N(0,\sigma_n^2)$ where $\sigma_n=\left(\log n\right)/(\sqrt n\, \Gamma_n)$. 
    \STATE \textbf{Prior Specification:} 
    \STATE (a) Select a Gaussian process prior $W^m$. 
     \STATE
     (b) Set an adjusted prior for $m_\eta(d,X_i) = \Psi\left(\eta^m(d,X_i)\right)$, where $ \eta^m(d,X_i)=W^m(d,X_i) + \lambda\,\widehat \gamma(d,X_i)$.     \\
\textbf{Posterior Computation:} 
    \FOR{$s=1,\ldots, S$}
        \STATE  (a) 
        Generate the $s$-th draw of the posterior of $(m_\eta(d, X_i))_{i=1}^n$ using the adjusted prior and the data; denote it as $(m^s_\eta(d, X_i))_{i=1}^n$.
        \STATE (b) 
        Draw Bayesian bootstrap weights $M^s_{ni}=e^s_i/\sum_{j=1}^n e^s_j$ where $e_i^s \stackrel{iid}{\sim} \textup{Exp}(1)$, $i=1,\dots,n$.
      
        \STATE (c) Calculate the corrected posterior draw for the ATE: 
        \begin{equation}\label{recentered_bay_est}
	    \check{\tau}_\eta^{s}=\tau_\eta^{s}-\widehat{b}^s_{\eta},
        \end{equation}
        \begin{equation}\label{debiased_bay_est}
	    \tau_\eta^{s}= \sum_{i=1}^n M^s_{ni}\big(m^s_\eta(1,X_i)-m^s_\eta(0,X_i)\big)\quad \text{and}\quad \widehat{b}^s_{\eta}=\frac{1}{n}\sum_{i=1}^n \boldsymbol{\tau}[m_\eta^s-\widehat m](Z_i),
        \end{equation}
        where $\boldsymbol{\tau}[m](z):=m(1,x)-m(0,x)+\widehat{\gamma}(d,x)(y-m(d,x))$.
    \ENDFOR
    
 \STATE \textbf{Output: $\{\check{\tau}_\eta^{s}:s=1,\ldots,S\}$} 
\end{algorithmic}
\end{algorithm}

Given the draws from the corrected posterior calculated in Algorithm \ref{algorithm}, we obtain the point estimate and credible set as follows. The Bayesian point estimator is $\overline{\tau}_{\eta}=\frac{1}{S}\sum_{s=1}^S \check{\tau}_\eta^{s}$. The $100\cdot(1-\alpha)\%$ credible set for the ATE parameter $\tau_0$ is given by
\begin{equation*}
	\mathcal{C}_n(\alpha)=\big\{\tau: q_n(\alpha/2)\leq \tau \leq q_n(1-\alpha/2)\big\},
\end{equation*}
where $q_n(a)$ denotes the $a$-th quantile of $\{\check{\tau}_\eta^{s}:s=1,\ldots,S\}$.

For the implementation of our pilot estimator $\widehat\gamma$ given in \eqref{Riesz_est},  we recommend using propensity scores estimated by the Logistic Lasso.
 For the implementation of the pilot estimator $\widehat m$, we adopt the posterior mean of $m_{\eta}$ generated from a Gaussian process prior without adjustment, as in \cite{ghosal2006binary}. Section \ref{sec:implement:GP} provides more implementation details. To approximate the posterior distribution, we make  use of the Laplace approximation, but one can also resort to the Markov Chain Monte Carlo (MCMC) algorithms.
The parameter $\sigma_n$ controls the relative weight placed on the prior adjustment relative to the standard unadjusted prior on $\eta^m$ (e.g., a Gaussian prior with a squared exponential covariance function). Regarding the tuning parameter $\sigma_n$, we emphasize that our finite sample results are not sensitive to its choice, as we show in Supplemental Appendix \ref{appendix:simu}.

						\begin{remark}	[Bayesian bootstrap]\label{rem:BB}	Under unconfoundedness and the reparametrization in \eqref{repar}, the ATE can be written as $\tau_{\eta}=\int [ \Psi\left(\eta^m(1,x)\right)- \Psi\left(\eta^m(0,x)\right)]\,\mathrm{d}F_\eta(x)$. 
			With independent priors on $\eta^m$ and $F_\eta$, their posteriors also become independent. It is thus sufficient to consider the posterior for $\eta^m$ and $F_\eta$ separately. We place a Dirichlet process prior for $F_\eta$ with the base measure to be zero. Consequently,  the posterior law of $F_\eta$ coincides with the Bayesian bootstrap \citep{rubin1981bayesian}; also see \cite{chamberlain2003bayesian}. One key advantage of the Bayesian bootstrap is that it allows us to incorporate a broad class of data generating processes, whose posterior can be easily sampled. Replacing $F_\eta$ by the standard empirical cumulative distribution function does not provide sufficient randomization of $F_\eta$, as it yields an underestimation of the asymptotic variance;
			see \cite[p. 3008]{ray2020causal}. In principle, one could consider other types of bootstrap weights; however, these generally do not correspond to the posterior of any given prior distribution.
			\end{remark}

	\section{Main Theoretical Results}\label{sec:asympt}
In this section, we derive the  Bernstein-von Mises (BvM) theorem which establishes the asymptotic equivalance between our Bayesian procedure and the frequentist-type semiparametric efficient one for the ATE. We consider an asymptotically efficient estimator $\widehat{\tau}$ with the following linear representation:
		\begin{equation}\label{def:est:chi}
		\widehat{\tau}=\tau_0+\frac{1}{n}\sum_{i=1}^n \widetilde{\tau}_0(Z_i)+o_{P_0}(n^{-1/2}),
		\end{equation}
		where $\widetilde{\tau}_0=\widetilde{\tau}_{\eta_0}$ is the	efficient influence function in accordance with \eqref{eif_ate}. Below, we denote $Z^{(n)}=(Z_1,\ldots, Z_n)$. 
	By virtue of the BvM Theorem, two conditional distributions $\sqrt{n}(\tau_{\eta}-\widehat{\tau})|Z^{(n)}$ and $\sqrt{n}(\widehat{\tau}-\tau_{\eta})|\eta=\eta_0$ are asymptotically equivalent under the underlying sampling distribution. 
	Another important consequence of the BvM theorem is about the asymptotic normality and efficiency of the Bayesian point estimator. That is, $\sqrt n(\overline\tau_{\eta}-\tau_0)$ is asymptotically normal with mean zero and variance
$\textsc v_0=\E_0\left[ \widetilde{\tau}_0^2(Z_i)\right]$. Thus,  $\overline\tau_{\eta}$ achieves the semiparametric efficiency bound of \cite{hahn1998role}.

	\subsection{Least Favorable Direction}
	Our prior correction through the Riesz representer $\gamma_0$ is motivated by the least favorable direction of Bayesian submodels. We first provide such least favorable calculations, which are closely linked to the semiparametric efficiency. Consider the one-dimensional submodel $t\mapsto \eta_t$ defined by the path
	\begin{eqnarray}\label{submodel}
	\pi_t(\cdot ) =  \Psi(\eta^{\pi}+t\mathfrak{p})(\cdot),~~
	m_t(\cdot) =  \Psi(\eta^m+t\mathfrak{m})(\cdot ),~~
	f_t(\cdot )= \frac{f(\cdot)e^{t\mathfrak{f}(\cdot)}}{\int e^{t\mathfrak{f}(x)}f(x)\,\mathrm{d}x},
	\end{eqnarray}
	for a given direction $(\mathfrak{p}, \mathfrak{m},\mathfrak{f})$ with $\int \mathfrak{f}(x)f(x)\,\mathrm{d}x=0$. The difficulty of estimating the parameter $\tau_{\eta_{t}}$ for the submodels depends on the direction  $(\mathfrak{p}, \mathfrak{m},\mathfrak{f})$. Among them, let $ \xi_{\eta}= (\xi_{\eta}^{\pi},\xi_{\eta}^{m},\xi_{\eta}^{f})$ be the \emph{least favorable direction} that is associated with the most difficult submodel.  It yields the largest asymptotic optimal variance for estimating $\tau_{\eta_{t}}$ among all submodels.
	Let $p_{\eta_t}$ denote the joint density of $Z$ depending on $\eta_t:= (\pi_t, m_t, f_t)$. Taking derivative of the logarithmic density $\log p_{\eta_t} (z)$ with respect to $t$ and evaluating at $t=0$ gives the score operator:
\begin{equation}\label{score_ate}
B_{\eta}(\mathfrak{p},\mathfrak{m},\mathfrak{f})(z)= B_{\eta}^{\pi}\mathfrak{p}(z) + B_{\eta}^{m}\mathfrak{m}(z) + B_{\eta}^{f}\mathfrak{f}(z),
\end{equation}
where $B_{\eta}^{\pi}\mathfrak{p}(z) =   (d-\pi_\eta(x))\mathfrak{p}(x)$, $B_{\eta}^{m}\mathfrak{m}(z)=  (y-m_\eta(d,x))\mathfrak{m}(d,x)$ and $B_{\eta}^{f}\mathfrak{f}(z)= \mathfrak{f}(x)$. The least favorable direction is defined as the solution $\xi_\eta$ which solves the equation
$B_{\eta}\xi_{\eta}=\widetilde{\tau}_{\eta}$, see \citet[p.370]{ghosal2017fundamentals}. We immediately obtain the following.

\begin{lemma}\label{lemma:lfd}
Consider the submodel \eqref{submodel}. 
Let Assumption \ref{Ass:unconfounded} hold for $P_\eta$ with any $\eta$ under consideration, then
the least favorable direction for estimating the ATE parameter in \eqref{ate} is:
\begin{equation}\label{lfd}
 \xi_{\eta}(d,x)= \left(0,\gamma_\eta(d,x), m_\eta(1,x)-m_\eta(0,x)-\tau_\eta\right),
\end{equation}
where the Riesz representer $\gamma_\eta$ is given in \eqref{riesz:def}.
 \end{lemma}
Lemma \ref{lemma:lfd} motivates the adjustment of the prior distribution as considered in our Bayesian procedure in Section \ref{sec:method_outline}. Our prior correction, which takes the form of the (estimated) least favorable direction, provides an exact invariance under a shift of nonparametric components in this direction. It provides additional robustness against posterior inaccuracy in the ``most difficult direction'', i.e., the one inducing the largest bias in the average treatment effects.
We also note that Lemma \ref{lemma:lfd} extends the result in Section 2.1 in \cite{ray2020causal} for the missing data problem, which is equivalent as observing only one arm (either the treatment or control arm), to the context of ATE estimation that involves both arms.

\subsection{Assumptions for Inference}
We now provide additional notations and assumptions. The posterior distribution plays an important role in the following analysis and is given by
\begin{equation*}
	\Pi\left((\pi,m)\in A, F\in B|Z^{(n)}\right)=\int_{B}\frac{\int_{A}\prod_{i=1}^{n}p_{\pi,m}(Y_i, D_i |X_i) \,\mathrm{d}\Pi(\pi,m)}{\int \prod_{i=1}^{n}p_{\pi,m}(Y_i, D_i | X_i) \,\mathrm{d}\Pi(\pi,m)}\mathrm{d}\Pi(F|X^{(n)})
\end{equation*}
where $p_{\pi,m}$ denotes the conditional density of $(Y_i, D_i) $ given $X_i$, given by \eqref{x_den} divided by the marginal density of $X_i$.
		We write $\mathcal{L}_{\Pi}(\sqrt{n}(\tau_\eta-\widehat{\tau})|Z^{(n)})$ for the marginal posterior distribution of $\sqrt{n}(\tau_\eta-\widehat{\tau})$.
		 We focus on the case that $\eta^{\pi }$ has a prior that is independent of the prior for $(\eta^m,F)$. Because the likelihood function (\ref{x_den}) factorizes into $(\eta^m,\eta^{\pi},F)$ separately, the posterior of $\eta^{\pi }$ is also independent of the posterior for $(\eta^m,F)$. Due to the fact that $\tau_{\eta}$ does not depend on $\eta^{\pi}$, it is unnecessary to further discuss a prior or posterior distribution on $\eta^{\pi}$. 
		 
We first introduce high-level assumptions and discuss primitive conditions for those in the next section. Below, we consider some measurable sets $\mathcal H^m_n$ of functions $\eta^m$ such that $\Pi(\eta^m\in\mathcal{H}^m_n|Z^{(n)})\to_{P_0} 1$. W                               e also denote $\mathcal{H}_n=\{\eta:\eta^m\in\mathcal{H}_n^m\}$ when we index the conditional mean function $m_{\eta}$ by its subscript $\eta$. 
We introduce the notation $\|\phi\|_{2, F_0}:= \sqrt{\int \phi^2(x)\,\mathrm{d}F_0(x)}$ for all $\phi\in L^2(F_0):=\{\phi:\|\phi\|_{2, F_0}<\infty\}$, as well as the supremum norm $\|\cdot\|_\infty$. For two sequences $\{a_n\}$ and $\{b_n\}$ of positive numbers, we write $a_n \lesssim b_n$ if $\limsup_{n\to\infty} (a_n / b_n)<\infty$, and $a_n \sim b_n$ if $a_n \lesssim b_n$ and $b_n \lesssim a_n$.

	\begin{assumption}\label{Assump:Rate}[Rates of Convergence]
The estimators $\widehat \pi$ and $\widehat m$, which are based on an auxiliary sample independent of $Z^{(n)}$, satisfy 	$\Vert \widehat{\pi}-\pi_0\Vert_{2, F_0}=O_{P_0}(r_n) $ and for $d\in\{0,1\}$:
	\begin{equation*}
\Vert \widehat{m}(d,\cdot)-m_0(d,\cdot)\Vert_{ 2,F_0}=O_{P_0}(\varepsilon_n)~\text{ and }~\sup_{\eta\in\mathcal{H}_n}\Vert m_\eta(d,\cdot)-m_0(d,\cdot)\Vert_{2,F_0}\lesssim \varepsilon_n,
	\end{equation*}
	where $\max\{\varepsilon_n, r_n\}\to 0$ and $\sqrt{n}\,\varepsilon_nr_n\to 0$. Further, $\Vert \widehat\gamma\Vert_{\infty}=O_{P_0}(1)$. 
	\end{assumption}
 We adopt the standard empirical process notations as follows. For a function $h$ of a random vector $Z_i$ that follows distribution $P_0$, we let $P_0[h]=\int h(z)\,\mathrm{d}P(z),\mathbb{P}_n[h]=n^{-1}\sum_{i=1}^{n}h(Z_i)$, and $\mathbb{G}_n[h]=\sqrt n\left(\mathbb{P}_n-P_0\right)[h]$. Below, we make use of the notations $\bar{m}_{\eta}(\cdot)=m_{\eta}(1,\cdot)-m_{\eta}(0,\cdot)$ and $\bar{m}_{0}(\cdot)=m_{0}(1,\cdot)-m_{0}(0,\cdot)$.
 
	\begin{assumption}\label{Assump:Donsker}[Complexity] 
	For $\mathcal{G}_n=\{\bar{m}_{\eta}(\cdot):\eta\in \mathcal{H}_n\}$ it holds
	$\sup _{\bar{m}_{\eta} \in \mathcal{G}_n}\left|(\mathbb{P}_n-P_0) \bar{m}_{\eta}\right| =o_{P_0}(1)$
and 
	\begin{align}\label{NewSE}
		\sup_{\eta\in\mathcal{H}_n}\left|\mathbb{G}_n\left[\left(\widehat\gamma-\gamma_0\right)(m_\eta-m_0)\right]\right|&=o_{P_0}(1).
	\end{align}
	\end{assumption}	
Recall the propensity score-dependent prior on $m$ given by
$m_\eta(\cdot) = \Psi\left(\eta^m(\cdot)\right)$ where $\eta^m(\cdot )=W^m(\cdot) + \lambda\,\widehat \gamma(\cdot)$.
The restriction on $\lambda$ is made through its hyperparameter $\sigma_n>0$. 
	\begin{assumption}\label{Assump:Prior}[Prior Stability]
		For $d\in\{0,1\}$, $W^m(d,\cdot)$ is a continuous stochastic process independent of the normal random variable $\lambda\sim N(0,\sigma_n^2)$, where $\sigma_n\lesssim 1$, $n\sigma^2_{n}\to\infty$ and that satisfies: (i)
		$\Pi\left(\lambda:|\lambda|\leq u_n\sigma_n^2\sqrt{n}\mid Z^{(n)}\right)\to_{P_0}1$, 
	for some deterministic sequence $u_n\to 0$	and (ii) 
		$\Pi\left((w,\lambda):w+(\lambda+tn^{-1/2})\widehat{\gamma}\in\mathcal{H}_n^m\mid Z^{(n)} \right)\to_{P_0}1$
		for any $t\in\mathbb R$.
	\end{assumption}
	
	\textit{Discussion of Assumptions:} 
Assumption~\ref{Assump:Rate} imposes
sufficiently fast convergence rates for the pilot estimators for the conditional
mean function $m_{0}$ and the propensity score $\pi _{0}$. When considering frequentist pilot estimators, these
rate conditions can be justified by adopting the recent proposals of
\cite{ChernozhukovNeweySingh2020a, ChernozhukovNeweySingh2020b}. One can also use Bayesian point estimators such as the posterior mean of
the Gaussian process for $\widehat{m}$ and $\widehat{\pi}$. The posterior
convergence rate for the conditional mean $m_{\eta}$ can be derived in
the same spirit of \cite{ray2020causal}. The rate conditions in Assumption 2 also resemble conditions (i) and (ii) of Theorem~1 of \cite{farrell2015} in the context of frequentist
estimation. Remark~\ref{rem:DR} illustrates that under classical smoothness assumptions, this assumption is less restrictive than the method of \cite{ray2020causal} or
other approaches for semiparametric estimation of ATEs as found in
\cite{chen2008semiparametric} or \cite{farrell2021deep}. Assumption~\ref{Assump:Prior} incorporates Conditions (3.9) and (3.10) from Theorem~2 in \cite{ray2020causal}, and it is imposed to check the invariance property
of the adjusted prior distribution. These restrictions are mild and extend
beyond the Gaussian processes considered in Section~\ref{sec:gauss:prior} for concreteness.
 
Assumption~\ref{Assump:Donsker} restricts the functional class
$\mathcal{G}_{n}$ to form a $P_{0}$-Glivenko--Cantelli class; see Section~2.4 of \cite{van1996empirical}. 
This imposes a new stochastic equicontinuity condition, as (\ref{NewSE}) restricts a product structure involving $\widehat\gamma$ and $m_{\eta}$, which further relaxes the corresponding condition from \cite{ray2020causal}, namely, $\sup_{\eta \in \mathcal{H}^{m}_{n}}\mathbb{G}_{n} [m_{\eta}-m_{0}] = o_{P_{0}}(1)$.
In the next section,
we demonstrate that our formulation allows for double robustness under
H\"older classes (see Remark~\ref{rem:DR}). Hence, the complexity of the
functional class $(m_{\eta}-m_{0})$ can be compensated by sufficient regularity
of the corresponding Riesz representer and vice versa. A condition
similar to our Assumption~\ref{Assump:Donsker} is also used in the frequentist
literature; see Section~2 of  Benkeser, Carone, van der Laan, and
Gilbert (\citeyear{benkeser2017doubly}). Nonetheless, the
technical argument differs substantially from the frequentist's study,
because we mainly need the condition (\ref{NewSE}) to control changes in
the likelihood under perturbations along the estimated and true least favorable
directions. This is unique to Bayesian analysis with nonparametric priors.
	
\subsection{A Double Robust Bernstein-von Mises Theorem}
	We now establish a new Bernstein–von Mises theorem, which establishes the asymptotic normality of
the posterior distribution, modulo a ``bias term''. In a next step, we show that posterior correction, as proposed in our procedure, eliminates this ``bias term''.
This asymptotic equivalence result is established using the bounded Lipschitz distance. 
For two probability measures $P,Q$ defined on a metric space $\mathcal{Z}$, we define the bounded Lipschitz distance as
 \begin{equation}
 	d_{BL}(P,Q)=\sup_{f\in BL(1)}\left| \int_{\mathcal{Z}}f(\mathrm{d}P-\mathrm{d}Q)\right|,
 \end{equation}
where 
\begin{equation*}
	BL(1)=\left\{f:\mathcal{Z}\mapsto\mathbb{R}, \sup_{z\in\mathcal{Z}}|f(z)|+\sup_{z\neq z'}\frac{|f(z)-f(z')|}{\|z-z'\|_{\ell_2}}\leq 1 \right\}.
\end{equation*}
Here,  $\|\cdot\|_{\ell_2}$ denotes the vector $\ell_2$ norm. 

Below is our main statement about the asymptotic behavior of the posterior distribution of $\tau_{\eta}$. As in the modern Bayesian paradigm, the exact posterior is rarely of closed-form, and one needs to rely on certain Monte Carlo simulations, such as the implementation procedure in Section \ref{sec:method_outline}, to approximate this posterior distribution, as well as the resulting point estimator and credible set.
\begin{theorem}\label{thm:BvM}
	Let Assumptions \ref{Ass:unconfounded}--\ref{Assump:Prior} hold. Then we have
	\begin{equation*}
		d_{BL}\left(\mathcal{L}_{\Pi}(\sqrt{n}(\tau_\eta-\widehat{\tau}-b_{0,\eta})|Z^{(n)}), N(0,\textsc v_0) \right)\to_{P_0} 0,
	\end{equation*}
	where $b_{0,\eta}:=	\mathbb{P}_n[\gamma_0(m_0-m_{\eta})-(\bar{m}_0-\bar{m}_{\eta})] $.
\end{theorem}
We emphasize that the above BvM theorem is not feasible for applications, because it depends on the ``bias term'' $b_{0,\eta}$, which depends on the unknown conditional mean $m_0$. Nonetheless, it provides an important theoretical benchmark. One can follow the existing literature on semiparametric BvM theorems to impose the so-called ``no-bias" condition, but this generally leads to strong smoothness restrictions and may not be satisfied when the dimensionality of covariates is large relative to the smoothness properties of the underlying functions; see the discussion on page 395 of \cite{van1998asymptotic}. 

This ``bias term'' in our context consists of two key components, with the first involving unknown true functions and the second depending on the posterior of $m_{\eta}$. We consider pilot estimators for the unknown functional parameters in $b_{0,\eta}$. The correction term $\widehat{b}_{\eta}$, as introduced in \eqref{debiased_bay_est}, results in a feasible Bayesian procedure that satisfies the BvM theorem under double robustness, as demonstrated below.
\begin{theorem}\label{thm:Debias}
	Let Assumptions \ref{Ass:unconfounded}--\ref{Assump:Prior} hold. Then we have
	\begin{equation*}
		d_{BL}\left(\mathcal{L}_{\Pi}(\sqrt{n}(\tau_\eta-\widehat{\tau}-\widehat{b}_\eta)|Z^{(n)}), N(0,\textsc v_0) \right)\to_{P_0} 0.
	\end{equation*}
\end{theorem}

We now show how Theorem \ref{thm:Debias} can provide frequentist justification of Bayesian methods to construct the point estimator and the confidence sets. Recall that $\overline{\tau}_{\eta}$ represents the posterior mean. Introduce a Bayesian credible set $\mathcal{C}_n(\alpha)$ for $\tau_\eta$, which satisfies $\Pi(\tau_\eta\in \mathcal{C}_n(\alpha)|Z^{(n)})=1-\alpha$ for a given nominal level $\alpha\in(0,1)$. The next result shows that $\mathcal{C}_n(\alpha)$ also forms a confidence interval in the frequentist sense for the ATE parameter whose coverage probability under $P_0$ converges to $1-\alpha$. 

\begin{corollary}\label{cor:CI}
	Let Assumptions \ref{Ass:unconfounded}--\ref{Assump:Prior} hold. Then under $P_0$, we have
	\begin{equation}
		\sqrt{n}\left(\overline{\tau}_{\eta}-\tau_0\right)\Rightarrow N(0,\textsc v_0).
	\end{equation}
	Also, for any $\alpha\in(0,1)$ we have
		$P_0\big(\tau_0\in  \mathcal{C}_n(\alpha)\big) \to 1-\alpha$.
\end{corollary}

 To the best of our knowledge, this is the first BvM theorem that entails the double robustness. 
We discuss the distinction from Theorem 2 in \cite{ray2020causal}. Their work laid the theoretical foundation for  Bayesian inference based on propensity score adjusted priors. Specifically,  under this prior adjustment, they establish a BvM result under weak regularity conditions on the propensity score function.
Our analysis differs from \cite{ray2020causal} in two crucial ways. First, we improve on their Lemma 3 by showing that it is possible to verify the prior stability condition for propensity score-adjusted priors under the product structure in Assumption \ref{Assump:Donsker}, modulo the ``bias term" $b_{0,\eta}$. This separation is essential to identify the source of the restrictive condition, such as the Donsker property on $m_{\eta}$, which is mainly used to eliminate $b_{0,\eta}$.
Second, our proposal introduces an explicit debiasing step, borrowing key insights from recent developments in the DML literature.
 
	\begin{remark}[Connection with frequentist robust estimation] 
In our BvM theorem, we do not restrict the centering estimator $\widehat{\tau}$, as long as it admits the linear representation	as in \eqref{def:est:chi}.
A popular frequentist estimator for the ATE that achieves double robustness is
	\begin{equation}\label{freq_DR}
	\widehat{\tau}=	n^{-1}\sum_{i=1}^n \big(\widehat m(1,X_i)-\widehat m(0,X_i)\big) + n^{-1}\sum_{i=1}^n\widehat\gamma(D_i,X_i)\big(Y_i-\widehat m(D_i,X_i)\big)
	\end{equation}
	based on frequentist-type pilot estimators $\widehat{m}$ of the conditional mean function $m_0$ and $\widehat{\gamma}$ of the Riesz representer $\gamma_0$; see \cite{robins1995} and more recently	\cite{ChernozhukovNeweySingh2020a,ChernozhukovNeweySingh2020b}.
	The double robust or double machine learning estimator \eqref{freq_DR} recenters the plug-in type functional by an explicit correction factor that depends on the Riesz representer.\footnote{Another popular method in the statistics literature is the targeted learning approach \citep{vanderLaan2011tl,benkeser2017doubly}.	} 
	Our main result establishes the asymptotic equivalence of our estimator and \eqref{freq_DR}. This not only offers frequentist validity to our Bayesian procedure but also provides a Bayesian interpretation for doubly robust frequentist methods.

\end{remark}
\begin{remark}[Parametric Bayesian Methods]
A couple of recent papers propose doubly robust Bayesian recipes for ATE inference, under parametric model restrictions. \cite{saarela2016double} considered a Bayesian procedure based on an analog of the double robust frequentist estimator given in Equation \eqref{freq_DR}, replacing the empirical measure with the Bayesian bootstrap measure. However, there was no formal BvM theorem presented therein. 
 Another recent paper by \cite{yiu2020unequal} explored Bayesian exponentially tilted empirical likelihood with a set of moment constraints that are of a double-robust type. They proved a BvM theorem for the posterior constructed from the resulting exponentially tilted empirical likelihood under parametric specifications.
  \cite{luo2023semiparametric} provided Bayesian results for ATE estimation in a partial linear model, which implies homogeneous treatment effects. They also assign parametric priors to the propensity score. Their BvM Theorem allows for misspecification only in a parametric nonlinear component of the outcome equation. It is not clear how to extend their analysis to incorporate flexible nonparametric modeling strategies.
\end{remark}

	\section{Illustration using Squared Exponential Process Priors}\label{sec:gauss:prior}
We illustrate the general methodology by placing a particular Gaussian process prior on $\eta^m(d,\cdot)$ in relation to the conditional mean functions for $d\in\{0,1\}$. The Gaussian process regression has been extensively used among the machine learning community, and started to gain popularity among economists \citep{kasy2018tax}. We provide primitive conditions used in our main results in the previous section. In addition, we provide details on the implementation using Gaussian process priors and discuss the data-driven choices of tuning parameters.
		\subsection{Asymptotic Results under Primitive Conditions}
 Let $(W(t):t\in\mathbb{R}^p)$ be a generic centered and homogeneous Gaussian random field with covariance function of the following form $\mathbb{E}[W(s)W(t)]=\phi(s-t)$,
for a given continuous function $\phi:\mathbb{R}^p\mapsto \mathbb{R}$. We consider $W(t)$ as a Borel measurable map in the space of continuous functions on $[0,1]^p$, equipped with the supremum norm $\Vert \cdot\Vert_{\infty}$. The Gaussian process is completely determined by the covariance function.
For example, the covariance function of the squared exponential process is given by $\mathbb{E}[W(s)W(t)]=\exp(-\Vert s-t\Vert_{\ell_2}^2)$, as its name suggests.  In this section, we focus on the squared exponential process prior, which is one of the most commonly used priors in applications; see \cite{rassmusen2006gaussian} and \cite{murphy2023pml}.
We also consider a rescaled Gaussian process $\big(W(a_nt):\,t\in [0,1]^p\big)$ for some $a_n>0$. Intuitively speaking, $a_n^{-1}$ can be thought as a bandwidth parameter. For a large $a_n$ (or equivalently a small bandwidth), the prior sample path $t\mapsto W(a_nt)$ is obtained by shrinking the long sample path $t\mapsto W(t)$. Thus, it employs more randomness and becomes suitable as a prior model for less regular functions, see \cite{van2008gaussian,van2009adaptive}. 

Below, $\mathcal{C}^{s_m}([0,1]^p)$ denotes a H\"older space with the smoothness index $s_m$. Specifically, we illustrate our theory with the case where $m_0(d,\cdot)\in \mathcal{C}^{s_m}([0,1]^{p})$ for $d\in\{0,1\}$.  Given such a H\"older-type smoothness condition, we choose
\begin{equation}\label{RescaleRate}
a_n\sim n^{1/(2s_m+p)}(\log n)^{-(1+p)/(2s_m+p)}.
\end{equation}		
 Under (\ref{RescaleRate}), a rescaled Gaussian process $\big(W(a_nt):\,t\in [0,1]^p\big)$ induces the posterior contraction rate for the conditional mean function $m_\eta(d,\cdot)$ to be $\varepsilon_n=n^{-s_m/(2s_m+p)}(\log n)^{s_m(1+p)/(2s_m+p)}$; see Section 11.5 of \cite{ghosal2017fundamentals}.
 The particular choice of $a_n$ mimics the corresponding kernel bandwidth based on kernel smoothing methods. Other choices of $a_n$ will generally make the convergence rate slower. Nonetheless, as long as the propensity score is estimated with a sufficiently fast rate, our BvM theorem still holds. The next proposition illustrates our general theory when we adopt the rescaled squared exponential process prior for the conditional mean function. We use the superscript $m$ for the prior process $W^m$ to signify this relationship.
\begin{proposition}\label{prop:exponential}
Let Assumption \ref{Ass:unconfounded} hold.
The estimator $\widehat\gamma$ satisfies $\|\widehat\gamma\|_\infty=O_{P_0}(1)$ and $\|\widehat{\gamma}-\gamma_0\Vert_\infty= O_{P_0}\big((n/\log n)^{-s_\pi/(2s_\pi+p)}\big)$ for some $s_\pi>0$. Suppose $m_0(d,\cdot)\in \mathcal{C}^{s_m}([0,1]^{p})$ for $d\in\{0,1\}$ and some $s_m>0$ with $\sqrt{s_\pi \, s_m}>p/2$. Also, $\|\widehat{m}(d,\cdot)-m_0(d,\cdot)\Vert_{2, F_0}= O_{P_0}\big((n/\log n)^{-s_m/(2s_m+p)}\big)$. 
Consider the propensity score-dependent prior on $m$ given by $m(d,x) = \Psi\left(W^m(d,x) + \lambda\,\widehat \gamma(d,x)\right)$, where $W^m(d,\cdot)$ is the rescaled squared exponential process for $d\in\{0,1\}$, with its rescaling parameter $a_n$ of the order in \eqref{RescaleRate} and
$\left(n/\log n\right)^{-s_m/(2s_m+p)}\lesssim u_n\sigma_n$ for some deterministic sequence $u_n\to 0$, and  $\sigma_n\lesssim 1$.
	Then, the corrected posterior distribution for the ATE satisfies Theorem \ref{thm:BvM}.
\end{proposition}

\begin{remark}[Double Robust H\"older Smoothness]\label{rem:DR}
Proposition \ref{prop:exponential} requires $\sqrt{s_\pi \, s_m}>p/2$, which represents a trade-off between the smoothness requirement for $m_0$ and $\pi_0$. This encapsulate the \textit{double robustness}; i.e., a lack of smoothness of the conditional mean function $m_0$ can be mitigated by exploiting the regularity of the propensity score and vice versa. Referring to the H\"older class $\mathcal{C}^{s_m}([0,1]^{p})$, its complexity measured by the bracketing entropy of size $\varepsilon$ is of order $\varepsilon^{-2\upsilon}$ for $\upsilon=p/(2s_m)$. One can show that the key stochastic equicontinuity assumption in \cite{ray2020causal}, that is, their condition (3.5), is violated by exploring the Sudkov lower bound \citep{han2021set} when $\upsilon>1$ or equivalently when $s_m<p/2$. In contrast, our framework accommodates this non-Donsker regime as long as $\sqrt{s_\pi \, s_m}>p/2$, which enables us to exploit the product structure and a fast convergence rate for estimating the propensity score. Our methodology is not restricted to the case where propensity score belongs to a H\"older class per se. For instance, under a parametric restriction (such as in logistic regression) or an additive model with unknown link function, the possible range of the posterior contraction rate $\varepsilon_n$ for the conditional mean function can be substantially enlarged.
In the case $ s_m>p/2$, the bias term becomes asymptotically negligible, i.e., $b_{0,\eta}=o_{P_0}(n^{-1/2})$. This allows for smoothness robustness only with respect to the propensity score and is also known as single robustness. In this case, no posterior correction is required, see \cite{ray2020causal}.
\end{remark}

 \subsection{Implementation Details}\label{sec:implement:GP}
We provide details on the Gaussian process prior placed on $\eta^m(d,x)$ and its posterior computation. Algorithm \ref{algorithm} sets the adjusted prior as $\eta^m(d,x)=W^m(d,x) + \lambda\,\widehat \gamma(d,x)$: In our implementation, we choose the first component $W^m(d,x)$ to be a zero-mean Gaussian process  with the commonly used squared exponential covariance function \citep[p.83]{rassmusen2006gaussian}. That is, 
$K\left((d,x),(d',x')\right):= \nu^2 \exp\left(-a_{0n}^{2}(d-d^\prime)^2/2-\sum_{l=1}^{p}a_{ln}^{2}(x_{l}-x^\prime_{l})^2/2\right)$
where the hyperparameter $\nu^2$ is the kernel variance and $a_{0n},\ldots,a_{pn}$ are rescaling parameters that reflect the relevance of treatment and each covariate in predicting $\eta^m$. They are selected by maximizing the marginal likelihood. 
Conditional on  the data used to obtain the propensity score estimator $\widehat\pi$, the prior for $\eta^m$ has zero mean
and the covariance kernel $K^c$, which includes an additional term based on the estimated Riesz representer $\widehat\gamma$. It is given by
 $K^c\left((d,x),(d^\prime,x^\prime)\right) = K\left((d,x),(d^\prime,x^\prime)\right)  + \sigma_n^2\widehat \gamma(d,x)\,\widehat \gamma(d^\prime,x^\prime),$
cf. related constructions from \cite{ray2019debiased} and \cite{ray2020causal}. 
The parameter $\sigma_n$, representing the standard deviation of $\lambda$, controls the weight of the prior adjustment relative to the standard Gaussian process. The choice $\sigma_n=(\log n)/(\sqrt{n}\,\Gamma_n)$, where $\Gamma_n= n^{-1}\sum_{i=1}^n\vert\widehat\gamma(D_i,X_i)\vert$, as specified in Algorithm \ref{algorithm},  satisfies the conditions $\sigma_n\lesssim 1$ and $n\sigma^2_{n}\to\infty$ in Assumption \ref{Assump:Prior}, with probability approaching one. It is similar to the choice suggested by \citet[page 6]{ray2019debiased}, where $\sigma_n$ is proportional to $1/(\sqrt{n}\,\Gamma_n)$. The factor $\Gamma_n$ normalizes the second term (adjustment term) of $K^c$  to have the same scale as the unadjusted covariance $K$.
Supplemental Appendix \ref{appendix:simu} shows that the finite sample performance of the double robust Bayesian approaches remains stable across different choices of $\sigma_n$.

Utilizing Gaussian process priors with zero mean and covariance function $K^c$, and incorporating the available data, we generate posterior draws of the vector $ \left[\eta^m(d,X_1),\cdots,\eta^m(d,X_n)\right]^{\top}$ for $d\in \{0,1\}$. This can be achieved through the Laplace approximation method detailed in Supplemental Appendix \ref{appendix:LPA}. 

For the implementation of the pilot estimator $\widehat\gamma$ given in \eqref{Riesz_est}, we recommend Logistic Lasso for the propensity score, with the penalty parameter chosen by cross-validation  \citep{friedman2010regularization}.
As a  pilot estimator $\widehat{m}$ in Algorithm \ref{algorithm} for posterior correction, we use the uncorrected posterior mean $\sum_{s=1}^S m_\eta^{s}/S$, where $m_\eta^{s}$ is calculated following Step (a) of posterior computation in Algorithm \ref{algorithm},  but with a Gaussian process prior without adjustment, i.e., $\Psi\left(W^m(d,\cdot)\right)$. When the rescaling parameter $a_n$ is as stated in Proposition \ref{prop:exponential}, the convergence rate of $\widehat{m}$ is $O_{P_0}\big((n/\log n)^{-s_m/(2s_m+p)}\big)$. This can be shown by combining Theorems 11.22, 11.55 and 8.8 from \cite{ghosal2017fundamentals}.

\section{Numerical Results}
In this section, we apply our method to one version of the Lalonde--Dehejia--Wahba data that contains a treated sample of 185 men from the National Supported Work (NSW) experiment and  a control sample of 2490 men from the Panel Study of Income Dynamics (PSID). The data has been used by \cite{lalonde1986}, \cite{dehejia1999causal}, \cite{abadie2011bias}, and \cite{armstrong2021finite}, among others. We refer readers to \cite{lalonde1986}, and \cite{dehejia1999causal} for reviews of the data.\footnote{The data is available on Dehejia's website:
 \href{http://users.nber.org/~rdehejia/nswdata2.html}{http://users.nber.org//$\sim$ rdehejia/nswdata2.html}.}
 
 \subsection{Simulations}\label{sec:simu}
In this section, we consider a simulation study where the observations are randomly drawn from a large sample generated by the Wasserstein Generative Adversarial Networks (WGAN) method from the the Lalonde--Dehejia--Wahba data, see \cite{athey2021using}.
We view their simulated data as the population and repeatedly draw our simulation samples (each consisting of 185 treated observations and 2490 control observations) for each of the $1000$ Monte Carlo replications.
We slightly depart from previous studies by focusing on a binary outcome $Y$: the employment indicator for the year 1978, which is defined as an indicator for positive earnings. The treatment $D$ is the participation in the NSW program. We are interested in the average treatment effect of the NSW program on the employment status. For the  set of covariates, we  follow \cite{abadie2011bias} and include nine variables: age,  education,  black,  Hispanic,  married, earnings in 1974,  earnings in 1975,  unemployed in 1974, and unemployed in 1975.
We implement our double robust Bayesian method (DR Bayes) following Algorithm \ref{algorithm}, using  $S=5000$ posterior draws and  the pilot estimator $\widehat \gamma$ and $\widehat m$, as detailed at the end of Section \ref{sec:implement:GP}.
We compare DR Bayes to two other Bayesian procedures:
 First, we consider the prior adjusted  Bayesian method  (PA Bayes) proposed by \cite{ray2020causal}, which 
constructs the point estimate and credible interval based on $\tau_{\eta}^s$ in (\ref{debiased_bay_est}).  Second, we examine an unadjusted Bayesian method (Bayes) which is also based on $\tau_{\eta}^s$ but is generated using
Gaussian process priors without the adjustment.

We also compare our method to frequentist estimators. Match/Match BC corresponds to the nearest neighbor matching estimator and its bias-corrected version by \cite{abadie2011bias}, which adjusts for differences in covariate values through regression.
DR TMLE corresponds to the doubly robust targeted maximum likelihood estimator by \cite{benkeser2017doubly}.
DML refers to the double/debiased machine learning estimator from \cite{chernozhukov2017double}, where the nuisance functions $\pi_0$ and $m_0$ are estimated using random forests (which outperformed DML combined with other nuisance function estimators, such as Lasso, in our simulation setup).
Since the job-training data contains a sizable proportion of units with propensity score estimates very close to $0$ and $1$, we follow \cite{crump2009dealing} and discard observations with the estimated propensity score outside the range $[t, 1-t]$, with the trimming threshold $t\in\{0.10, 0.05, 0.01\}$.\footnote{\cite{crump2009dealing} suggested a simple rule of thumb with a threshold of $t=0.10$, while \cite{athey2021using} used $t=0.05$. Applying the optimal trimming rule proposed by \cite{crump2009dealing} to our simulated samples yields an average optimal trimming threshold $0.073$.}

   \begin{table}[H]
\centering
\caption{Simulation results using WGAN-generated data: trimming is based on the estimated propensity score within $[t,1-t]$, $\bar n =$ the average sample size after trimming, CP = coverage probability of $95\%$ credible/confidence interval, CIL = average length of the $95\%$ credible/confidence interval.
\qquad}\label{tab:simu_1}
\vskip.15cm
{\footnotesize 
 \begin{tabular}{lccccccccccccc}\toprule
\multicolumn{1}{c}{Methods}&\multicolumn{1}{c}{}& \multicolumn{1}{c}{Bias}&\multicolumn{1}{c}{ CP}&\multicolumn{1}{c}{CIL}&\multicolumn{1}{c}{}&\multicolumn{1}{c}{Bias}& \multicolumn{1}{c}{CP}&\multicolumn{1}{c}{CIL}&\multicolumn{1}{c}{}&\multicolumn{1}{c}{Bias}& \multicolumn{1}{c}{CP}&\multicolumn{1}{c}{CIL}\\
\cline{1-1}\cline{3-5}\cline{7-9}\cline{11-13}
\multicolumn{1}{c}{}&\multicolumn{1}{c}{}&\multicolumn{3}{c}{$t = 0.10  (\bar n = 240)$ }&\multicolumn{1}{c}{}& \multicolumn{3}{c}{$t =0.05 (\bar n =363)$}&\multicolumn{1}{c}{}&\multicolumn{3}{c}{$t = 0.01  (\bar n =664)$} \\
\cline{3-13}
Bayes&& -0.040 &  0.683 & 0.147&&-0.010 & 0.841 &0.149&& -0.006 & 0.911 &    0.120  \\
PA Bayes &&  -0.008 &    0.981  &    0.260 && 0.033 &    0.949  &    0.254 &&0.047 &    0.897  &    0.308 \\
DR Bayes  &&  -0.024 &    0.983 &    0.223 && 0.014 &    0.970 &    0.221   &&0.023 &    0.952 &    0.258  \\
\cline{1-5}\cline{7-9}\cline{11-13}
Match && 0.027 & 0.933 & 0.334 && 0.048 &0.908 & 0.323&& 0.033 & 0.965 & 0.323 \\
Match BC && 0.040 &0.880 & 0.347 && 0.065 & 0.816 & 0.334&& 0.083 & 0.804 & 0.339 \\
DR TMLE&& 0.015 & 0.832 & 0.300 &&  0.039 & 0.746 & 0.282&&0.039 & 0.668 & 0.242 \\
DML && 0.045 & 0.927 & 0.524&& 0.052 & 0.870 & 0.393&&0.054 & 0.918 & 0.522  \\
\bottomrule
\end{tabular}}
\end{table}

Table \ref{tab:simu_1} presents the finite sample performance of the Bayesian and frequentist methods mentioned above. We use the full data twice in computing the prior/posterior adjustments and the posterior distribution of the conditional mean function. 
Supplemental Appendix \ref{appendix:simu} reports the performance of DR Bayes using sample splitting, which results in similar coverage but a larger credible interval length due to the halved sample size.

Concerning the Bayesian methods for estimating the ATE, Table \ref{tab:simu_1} reveals that unadjusted Bayes yields highly inaccurate coverage except for the case with trimming constant $t=0.01$. If the prior is corrected using the propensity score adjustment, the results improve significantly. Nevertheless, our DR Bayes method demonstrates two further improvements:
   First, DR Bayes leads to smaller average confidence lengths in each case while simultaneously improving the coverage probability. This can be attributed to a reduction in bias and/or more accurate uncertainty quantification via our posterior correction. Second, when the trimming threshold is small (i.e., $t=0.01$), propensity score estimators can be less accurate, leading to reduced coverage probabilities of  PA Bayes. Our double robust Bayesian method, on the other hand, is still able to provide accurate coverage probabilities. In other words, DR Bayes exhibits more stable performance than PA Bayes with respect to the trimming threshold.\footnote{In additional simulations without trimming ($t=0$), we find that all double robust methods, including DR Bayes, substantially under-cover and/or inflate the length of their confidence intervals. This is consistent with \cite{crump2009dealing}, who point out that propensity score estimates close to the boundaries tend to induce substantial bias and large variances in estimating the ATE. We also note that unadjusted Bayes severely undercovers in this case.}

Our DR Bayes also exhibits encouraging performances when compared to frequentist methods. It provides a more accurate coverage than  bias-corrected matching, DR TMLE and DML. Compared with the matching estimator that exhibits a similarly good coverage performance, DR Bayes yields considerably shorter credible intervals.

\subsection{An Empirical Illustration}
We apply the Bayesian and frequentist methods considered above to the Lalonde--Dehejia--Wahba data.  Similar to the simulation exercise,  we consider a varying choice of the threshold $t\in \{0.10,0.05,0.01\}$.\footnote{Applying the optimal trimming rule proposed by \cite{crump2009dealing} yields an optimal threshold of $0.064$.} The ATE point estimates and confidence intervals are presented in Table \ref{tab:PSID_1}. 
As a benchmark,  the experimental data  that uses both treated and control groups in NSW ($n=445$) yields an ATE estimate (treated-control mean difference) of $0.111$ with a $95\%$ confidence interval $[0.026, 0.196]$. 

\begin{table}[H]
\centering
\caption{Estimates of ATE for the Lalonde--Dehejia--Wahba data: trimming is based on the
estimated propensity score within $[t,1-t]$, $\bar n =$ sample size after trimming.  ATE = point estimate,  $95\%$ CI = $95\%$ credible/confidence interval,  CIL = $95\%$ credible/confidence interval length.
\qquad}\label{tab:PSID_1}
\vskip.15cm
{\footnotesize 
\addtolength{\tabcolsep}{-1.5pt}    
 \begin{tabular}{lccc|ccc|ccc}\toprule
\multicolumn{1}{c}{Methods }&\multicolumn{3}{c}{$t = 0.10 (\bar n=245) $}&\multicolumn{3}{c}{$t = 0.05 (\bar n=398) $}&\multicolumn{3}{c}{ $t =0.01 (\bar n=740)$}\\
\cline{2-10}
\multicolumn{1}{c}{}& \multicolumn{1}{c}{ATE}&\multicolumn{1}{c}{$95\%$ CI}&\multicolumn{1}{c}{CIL}&\multicolumn{1}{c}{ATE}& \multicolumn{1}{c}{$95\%$ CI}&\multicolumn{1}{c}{CIL}&\multicolumn{1}{c}{ATE}& \multicolumn{1}{c}{$95\%$ CI}&\multicolumn{1}{c}{CIL}\\
\cline{1-10}
Bayes&  0.213 &   [0.120, 0.301] &  0.181 & 0.214 &   [0.132, 0.292] &    0.161 &0.198 &   [0.140, 0.251] &    0.112 \\
PA Bayes&  0.158 &   [0.019, 0.288]  &    0.270 &    0.170 &   [0.045, 0.281]  &    0.236 &    0.090 &  [-0.078, 0.233]  &    0.311\\
DR Bayes &   0.178 &   [0.061, 0.293]  &    0.231&    0.184 &  [0.064, 0.294]  &    0.230&    0.121 &  [-0.031, 0.250]  &    0.281\\
\cline{1-10}
Match & 0.188 & [0.022, 0.355] & 0.333 & 0.140 & [-0.029, 0.309] & 0.338 & 0.079 & [-0.111, 0.269] & 0.380\\ 
Match BC& 0.157 & [-0.006, 0.321] & 0.327 & 0.145 & [-0.021, 0.310] & 0.331 & 0.180 & [-0.004, 0.365] & 0.369\\ 
DR TMLE& -0.023 & [-0.171, 0.125] & 0.296 & 0.073 & [-0.074, 0.220] & 0.294 & 0.071 & [-0.146, 0.289] & 0.435\\
DML& 0.172 & [0.018, 0.327] & 0.308 &0.150 & [-0.010, 0.310] & 0.320 & 0.258 & [-0.183, 0.699] & 0.882\\ 
\bottomrule
\end{tabular}}
\end{table}
 As we see from Table \ref{tab:PSID_1},  the unadjusted Bayesian method yields larger estimates.  The adjusted Bayesian methods (PA and DR Bayes), on the other hand, produce estimates comparable to the experimental estimate. PA Bayes finds that the job training program enhanced the employment by $9.0\%$ to $17.0\%$ across different trimming thresholds, and DR Bayes estimates the effect from $12.1\%$ to $18.4\%$. Among frequentist estimators, the matching estimator and its bias-corrected version produce similar estimates as PA and DR Bayes, but with wider confidence intervals. DR TMLE produces negative estimates for $t=0.10$ when all other estimates are positive. For $t=0.10$ and $0.05$, DML yields similar point estimates as PA and DR Bayes, but with less estimation precision. In the case $t=0.01$ where the overlapping condition is closer to violation, however, its point estimate and confidence interval length become considerably larger than other methods.

\section{Extensions}\label{sec:extension}
This section extends the binary variable $Y$ to encompass general cases, including continuous, counting, and multinomial outcomes. First, we examine the class of single-parameter exponential families, where the conditional density function is solely determined by the nonparmatric conditional mean function. This covers continuous outcomes and counting variables. Second, we consider the ``vector" case of exponential families for multinomial outcomes. For both classes, we derive the novel correction to the Bayesian procedure and delegate more technical discussions to Supplemental Appendices \ref{appendix:proof_exponent} and \ref{appendix:exponent}. Additionally, we outline extensions to other causal parameters of interest.

\subsection{A Single-parameter Exponential Family}\label{sec:single}
In this part, we assume that the distribution of $Y_i$ conditional on $D_i$ and $X_i$ belongs to the ``single-parameter" exponential family, where the unknown parameter is the nonparametric conditional mean function $m(d,x)=\mathbb{E}[Y_i|D_i=d,X_i=x]$. The conditional density function is given by
\begin{equation}\label{condpdf}
f_{Y|D,X}(y\mid d,x) = c(y)\exp\left[q(m(d,x))ay-A(m(d,x))\right],
\end{equation}
where $A(m)= \log\int c(y)\exp\left[q(m) ay\right]\mathrm{d}y$,  and the function $q(\cdot)$ links the mean to the ``natural parameter'' of the exponential family. We also restrict the sufficient statistic to be linear in $y$. 

The family (\ref{condpdf}) not only encompasses the Bernoulli distribution (with  $q(m)=\log(m/(1-m))$, $A(m)=-\log(1-m)$, and $c(y)=a=1$), as considered in the previous sections, but also allows for counting and continuous outcomes. For instance, when $a=1$,  the Poisson distribution  corresponds to the choices $c(y)= 1/(y!)$, $q(m)=\log m$, and $A(m)=m$, while the  exponential distribution is represented by $c(y)=1$, $q(m)=-1/m$, and $A(m)=\log m$. Furthermore, the normal distribution with  $\text{Var}(Y|D,X)=\sigma^2$ for some $\sigma>0$, is captured by $c(y)=\exp(-y^2/(2\sigma^2))/\sqrt{2\pi\sigma^2}$, $q(m)=m/\sigma$, $A(m)=m^2 / (2\sigma^2)$, and $a=1/\sigma$.
We emphasize that model \eqref{condpdf} does not impose functional form assumptions on the conditional mean function $m$.
The joint density of $(Y_i,D_i,X_i)$ can be written as
	\begin{equation}\label{x_den_exp}
	p_{\pi,m,f}(y,d,x)=\pi(x)^d (1-\pi(x))^{1-d}c(y)\exp\left[q(m(d,x))ay-A(m(d,x))\right]f(x).
	\end{equation}
We consider the same reparametrization of $(\pi, m, f)$ as in \eqref{repar} except that now the second component of $\eta$ uses the general link function $q$ satisfying $\eta^{m} = q(m)$. We now state the least favorable direction for the exponential family case, which serves as motivation for the prior adjustment.
	
\begin{lemma}\label{lemma:lfd_exp}
Let Assumption \ref{Ass:unconfounded} hold for $P_{\eta}$ for any $\eta$ under consideration. Then, for the joint distribution \eqref{x_den_exp} and the submodel  $t\mapsto \eta_t$ defined by the path $m_t(d,x) =  q^{-1}(\eta^m+t\mathfrak{m})(d,x)$ with $(\pi_t, f_t)$ as defined in \eqref{submodel},
the least favorable direction for estimating the ATE parameter in \eqref{ate} is:
\begin{equation}\label{lfd_exp}
 \xi_{\eta}(d,x)= \left(0,\frac{1}{a}\gamma_\eta(d,x), m_\eta(1,x)-m_\eta(0,x)-\tau_\eta\right),
\end{equation}
where the Riesz representer $\gamma_\eta$ is given in \eqref{riesz:def}.
\end{lemma}
For the outcome family with $a=1$, which includes Bernoulli, Poisson and exponential distributions, the least favorable direction for ATE estimation coincides with the one as given in Lemma \ref{lemma:lfd}.
 To implement the double robust Bayesian procedure for general outcomes, one can still follow Algorithm \ref{algorithm}, with the logistic function $\Psi$ replaced by the inverse link function $q^{-1}$. For the normal (homoscedastic) outcome where prior adjustment $\lambda\widehat\gamma(d,x)$ in Algorithm \ref{algorithm} becomes $\lambda\widehat\gamma(d,x)/a$, the hyperparameter $a$ can be determined together with other parameters of the Gaussian process by optimizing the marginal likelihood as in \cite{ray2019debiased}.
Proposition \ref{prop:OnExponential} in the Supplemental Material provides primitive conditions for the BvM Theorem to hold under double robust smoothness conditions.

\subsection{Multinomial Outcomes}\label{sec:multi}
We now assume that the dependent variable $Y_i$ takes values in a finite set, specifically  $Y_i\in\{0,1, \dots,J\}$. The ATE can then be written as $\tau_\eta=\sum_{j=0}^J j\, \mathbb E_\eta\left[m_{\eta,j}(1,X) - m_{\eta, j}(0,X) \right]$, where the choice probabilities are
$m_{\eta, j}(d,x) = \Psi_j\left(\eta^{m_1},\cdots,\eta^{m_J}\right)(d,x)$
with the multinomial logit specification:
\begin{align*}
	\Psi_0\left(\eta^{m_1},\cdots,\eta^{m_J}\right)=\frac{1}{1+\sum_{l=1}^J \exp(\eta^{m_l})}
	\quad \text{and}\quad
	\Psi_j\left(\eta^{m_1},\cdots,\eta^{m_J}\right)=\frac{\exp(\eta^{m_j})}{1+\sum_{l=1}^J \exp(\eta^{m_l})},
\end{align*}
for $j=1,\ldots, J.$ The multinomial logit specification implies $m_{\eta, 0}(d,x) =1-\sum_{j=1}^{J}m_{\eta, j}(d,x)$. We now provide the least favorable direction for multinomial outcomes in the presence of multinomial outcomes and discuss its consequences for prior adjustment below.
\begin{lemma}\label{lemma:lfd:multinomial}
Consider the submodel  $t\mapsto \eta_t$ defined by the path $	m_{t, j}(d,x) =  \Psi(\eta^{m_j}+t\mathfrak{m}_j)(d,x)$, $1\leq j\leq J$,  with $(\pi_t, f_t)$ as defined in \eqref{submodel}. 
Let Assumption \ref{Ass:unconfounded} hold for $P_{\eta}$ for any $\eta$ under consideration, then the least favorable direction for estimating the ATE parameter  is:
\begin{equation*}
\xi_{\eta}(d,x)= \left(0,\gamma_\eta(d,x), 2\gamma_\eta(d,x), \ldots, J\gamma_\eta(d,x), m_\eta(1,x)-m_\eta(0,x)-\tau_\eta\right),
\end{equation*}
where the Riesz representer $\gamma_\eta$ is given in \eqref{riesz:def}.
 \end{lemma} 
 We emphasize that the least favorable direction calculation is not a trivial extension of \cite{hahn1998role} or \cite{ray2020causal}. This is because there are $J$ nonparametric components involved in the conditional probability function of the multinomial outcomes given covariates, and we need to consider the perturbation of those $J$ components together. Nonetheless, we show that the efficient influence function is of the same generic form as derived in \cite{hahn1998role}. In the proof of Lemma \ref{lemma:lfd:multinomial}, we compute the derivative of the parameter mapping along the path considered herein. We derive inner products involving the least favorable direction for each nonparametric component consisting of the conditional choice probabilities. The extension to the multinomial case had not been considered in the literature to our knowledge, and it offers a result of independent interest.
  
 Lemma \ref{lemma:lfd:multinomial} motivates the following modification of our double robust Bayesian estimator based on the propensity score-dependent prior on $m_{\eta,j}$ for $1\leq j\leq J$: 
 \begin{align*}
m_{\eta, j}(d,x) = \Psi_j\left(\eta^{m_1},\cdots,\eta^{m_J}\right)(d,x)\qquad \text{and}\qquad \eta^{m_j}(d,x)=W^{m_j}(d, x) + \lambda\,j \widehat \gamma(d,x),
\end{align*}
where $W^{m_j}(d,\cdot)$ is a continuous stochastic process independent  $\lambda\sim N(0,\sigma_n^2)$ for $\sigma_n>0$. We may then follow the implementation as described in Section \ref{sec:method_outline}	using $m_\eta(d,x)=\sum_{j=0}^J j\, m_{\eta,j}(d,x)$.

\subsection{Other Causal Parameters}\label{sec:other_param}
We now extend our procedure to general linear functionals of the conditional mean function. We do so only for binary outcomes, as the modification to other types of outcomes follows as above.
Recall that the observable data consists of $i.i.d.$ observations of $Z=(Y,D,X^\top)^\top$. 
The causal parameter of interest is	$\tau_0=\mathbb{E}_0[\psi(Z,m_0)]$,
where the function $\psi$ is linear with respect to the conditional mean function $m_0$. We introduce the Riesz representer $\gamma_0(d,x)$ satisfying
	$\mathbb{E}_0[\psi(Z,m)]=\mathbb{E}_0[\gamma_{0}(D,X)m(D,X)]$.
Let $\widehat{m}$ and $\widehat{\gamma}$ be  pilot estimators for the conditional mean and Riesz representer, respectively,  computed over an auxiliary sample. Our double robust Bayesian procedure can be extended by considering the corrected posterior distribution for $\tau_\eta$ as follows:
$	\check{\tau}_\eta^{s}= \sum_{i=1}^n M_{ni}^s \psi(Z_i,m^s_\eta)-n^{-1}\sum_{i=1}^n \boldsymbol{\tau}[m_\eta^s-\widehat m](Z_i)$, $s=1,\ldots,S$,
where here
$\boldsymbol{\tau}[m](z):=\psi(z,m)+\widehat{\gamma}(d,x)(y-m(d,x))$. The derivations of the least favorable directions in the following two examples are provided in Supplemental Appendix \ref{appendix:lfd}.
\begin{example}[Average Policy Effects]
The policy effect from changing the distribution of $X$ is
$	\tau_\eta^{P}=\int m_{\eta}(x)\,\mathrm{d}(G_1(x)-G_0(x))$, where the known distribution functions $G_1$ and $G_0$ have their supports contained in the support of the marginal covariate distribution $F_{\eta}$. 
Following the general setup, $	\psi(z,m_{\eta})=\psi(m_{\eta}):=\int m_{\eta}(x)\,\mathrm{d}(G_1(x)-G_0(x))$ with its Riesz representer 
$\gamma_{\eta}^{P}(x)=(g_1(x)-g_0(x))/f_{\eta}(x)$, where $g_1$ and $g_0$ stand for the density function of $G_1$ and $G_0$, respectively. \end{example}
\begin{example}[Average Derivative]
	For a continuous scalar (treatment) variable $D$, the average derivative is given by
		$\tau_\eta^{AD}=\mathbb{E}_\eta\left[\partial_dm_\eta(D,X)\right]$,	where $\partial_dm$ denotes the partial derivatives of $m$ with respect to the continuous treatment $D$. Thus, we have $\psi(Z,m_{\eta})=\partial_dm_{\eta}(D,X)$	with its Riesz representer given by $\gamma_\eta^{AD}(D,X)=\partial_d \pi_\eta(D,X)/\pi_\eta(D,X)$,	where here $\pi_\eta$ denotes the conditional density function of $D$ given $X$.
\end{example}

\appendix
 
\section{Proofs of Main Results}\label{appendix:main:proofs}
In the Appendix, $C>0$ denotes a generic constant, whose value might change line by line. We introduce additional subscripts when there are multiple constant terms in the same display. 
	In the following, we denote the log-likelihood based on $Z^{(n)}=(Z_i)_{i=1}^n$ as
	\begin{align*}
	\ell_n(\eta)=\sum_{i=1}^n\log p_\eta(Z_i)=\ell_n^\pi(\eta^\pi)+\ell_n^m(\eta^m)+\ell_n^f(\eta^f),
	\end{align*}
	where each term is the logarithm of the factors involving only $\pi$ or $m$ or $f$. 
	Recall the definition of the measurable sets $\mathcal H^m_n$ of functions $\eta^m$ such that $\Pi(\eta^m\in\mathcal{H}^m_n\mid Z^{(n)})\to_{P_0} 1$. We introduce the conditional prior $\Pi_n(\cdot):=\Pi(\cdot \cap \mathcal{H}^m_n)/\Pi(\mathcal{H}^m_n)$.
	The following posterior Laplace transform of $\sqrt{n}(	\tau_\eta-\widehat{\tau}-b_{0, \eta})$ given by
	\begin{equation}\label{laplace:transform:def}
	I_n(t)= \mathbb{E}^{\Pi_n}\left[e^{t\sqrt{n}(	\tau_\eta-\widehat\tau-b_{0, \eta})}\mid Z^{(n)} \right], ~~~\forall t\in\mathbb{R}
	\end{equation}
	plays a crucial role in establishing the BvM theorem \citep{castillo2012gaussian,castillo2015bvm,ray2020causal}. 
	To abuse the notation slightly, we define a perturbation of $\eta=(\eta^\pi,\eta^m)$ along the least favorable direction, restricted to the components corresponding to $\pi$ and $m$:
	\begin{equation}\label{def:pert}
		\eta_t(\eta):= \left(\eta^{\pi},\eta^m-\frac{t}{\sqrt{n}}\xi_0^m\right).
	\end{equation}
	We explicitly write the perturbation of $\eta^m$ by
$		\eta_t^m:=\eta_t(\eta^m)=\eta^m-t\xi_0^m/\sqrt{n}$.	Recall that $\xi_0^m$ coincides with the Riesz representer $\gamma_0$ by Lemma \ref{lemma:lfd}. In addition, we introduce the following notation:
\begin{equation}\label{rho}
	\rho^m(y,d,x):=y-m(d,x).
\end{equation}
Also, recall the notation $\bar m_{\eta}(\cdot) = m_{\eta}(1)-m_{\eta}(0)$, which is used in the following. 
In the proofs below, we make use of Lemmas \ref{lemma:Taylor}--\ref{lemma:DClass} which can be found in the Supplementary Appendix \ref{appendix:auxiliary}.

	\begin{proof}[Proof of Theorem \ref{thm:BvM}]
	Since the estimated least favorable direction $\widehat{\gamma}$ is based on observations that are independent of $Z^{(n)}$, we may apply Lemma 2 of \cite{ray2020causal}. It suffices to handle the ordinary posterior distribution with $\widehat{\gamma}$ set equal to a deterministic function $\gamma_n$.	By Lemma 1 of \cite{castillo2015bvm}, it is sufficient to show that the Laplace transform  $I_n(t)$ given in \eqref{laplace:transform:def} satisfies
	\begin{align}\label{laplace:transform:conv}
		I_n(t)\to_{P_0}\exp\left(t^2\textsc v_0/2\right),
	\end{align} 
	for every $t$ in a neighborhood of $0$, where the limit at the right hand side of (\ref{laplace:transform:conv}) is the Laplace transform of a $N(0,\textsc v_0)$ distribution.
	Note that we can write $\tau_\eta=\int\bar m_\eta\, \mathrm{d}F_\eta$. Further, let $\widehat\tau=\int \bar{m}_0\,\mathrm{d}F_0+\mathbb P_n[\widetilde{\tau}_0]$, which satisfies \eqref{def:est:chi}. 
	
	The Laplace transform $I_n(t)$ can thus be written as
	{\small 	\begin{align*}
			\int\int_{\mathcal{H}^m_n}\! \! \! \! \frac{\exp\big(t\sqrt{n}(\int \bar{m}_\eta \mathrm{d}F_\eta-\bar{m}_0\mathrm{d}F_0-b_{0, \eta})-t\mathbb{G}_n[\widetilde{\tau}_0]+\ell_n^m(\eta^m)-\ell_n^m(\eta^m_t)\big)\exp\big(\ell_n^m(\eta^m_t)\big)}{\int_{\mathcal{H}^m_n}\exp\big(\ell_n^m(\eta^{m\prime)}\big) \mathrm{d}\Pi(\eta^{m\prime})}\mathrm{d}\Pi(\eta^m)\mathrm{d}\Pi(F_\eta|Z^{(n)}).
	\end{align*}}%
The expansion in Lemma \ref{lemma:likelihood} gives the following identity for all $t$ in a sufficiently small neighborhood around zero and uniformly for $\eta^m\in\mathcal H^m_n$: 
	\begin{equation*}
		\ell_n^m(\eta^m)-\ell_n^m(\eta^m_t)=t\mathbb{G}_n[\gamma_0\rho^{m_0}]+t\mathbb{G}_n[\gamma_0(m_0-m_{\eta}) ]+t\sqrt{n}\int (\bar{m}_0-\bar{m}_{\eta})\mathrm{d}F_0 +\frac{t^2}{2}P_0(B_0^m\xi^m_0)^2+o_{P_0}(1),
	\end{equation*}
	where we make use of the notation $\rho^m(y,d,x)=y-m(d,x)$ and the score operator $B_0^m=B_{\eta_0}^m$ defined through (\ref{score_ate}).
	
	Next, we plug this into the exponential part in the definition of $I_n(t)$, which then gives
	\begin{footnotesize}
		\begin{align*}
			\int \int_{\mathcal{H}^m_n}\! \! \! \! &\frac{\exp\left(t\sqrt{n}\left(\int (\bar{m}_\eta \mathrm{d}F_\eta-\bar{m}_0\mathrm{d}F_0) +\int (\bar{m}_0-\bar{m}_{\eta})\mathrm{d}F_0-b_{0, \eta}\right)+t\mathbb{G}_n[\gamma_0(m_0-m_{\eta}) ]+\ell_n^m(\eta^m_t)\right)}{\int_{\mathcal{H}_n} \exp\left(\ell^m_n(\eta^{m\prime})\right)\mathrm{d}\Pi(\eta^{m\prime})}\mathrm{d}\Pi(\eta^m)\mathrm{d}\Pi(F_\eta|Z^{(n)})\\
			&\qquad\times \exp\left(-t\mathbb{G}_n[\widetilde{\tau}_0]+t\mathbb{G}_n[\gamma_0\rho^{m_0}]+\frac{t^2}{2}P_0(B_0^m\xi^m_0)^2+o_{P_0}(1) \right) \\
			&=\int \int_{\mathcal{H}^m_n} \frac{\exp\left(t\sqrt{n}\left(\int \bar{m}_\eta \mathrm{d}(F_\eta-F_0)-b_{0, \eta}\right)+t\mathbb{G}_n[\gamma_0(m_0-m_{\eta}) ]\right)\exp\left(\ell_n^m(\eta^m_t)\right)}{\int_{\mathcal{H}_n} \exp\left(\ell^m_n(\eta^{m\prime})\right)\mathrm{d}\Pi(\eta^{m\prime})}\mathrm{d}\Pi(\eta^m)\mathrm{d}\Pi(F_\eta|Z^{(n)})\\
			&\qquad\times \exp\left(-t\mathbb{G}_n[\widetilde{\tau}_0]+t\mathbb{G}_n[\gamma_0\rho^{m_0}]+\frac{t^2}{2}P_0(B_0^m\xi^m_0)^2+o_{P_0}(1) \right). 
		\end{align*}
	\end{footnotesize}	
By Fubini's theorem, the numerator of the previous expression coincides
with
	\begin{equation*}
		\int_{\mathcal{H}^m_n}\! \! \exp\Big(t\mathbb{G}_n[\gamma_0(m_0-m_{\eta})]-t\sqrt{n}b_{0, \eta}+\ell_n^m(\eta^m_t)\Big)\int \exp\left(t\sqrt{n}\int \bar{m}_{\eta}\mathrm{d}(F_\eta-F_0) \right)\mathrm{d}\Pi(F_\eta|Z^{(n)})\mathrm{d}\Pi(\eta^m).
	\end{equation*}
	By the assumed $P_0$-Glivenko-Cantelli property for $\mathcal{G}_n=\{\bar{m}_{\eta}:\eta\in\mathcal{H}_n\}$ in Assumption \ref{Assump:Donsker}, that is, $\sup _{\bar{m}_{\eta} \in \mathcal{G}_n}\left|(\mathbb{P}_n-P_0) \bar{m}_{\eta}\right| =o_{P_0}(1)$,  and the boundedness of $\bar{m}_{\eta}$, we apply Lemma \ref{lemma:DP} which establishes the convergence of the Laplace transform for the Dirichlet posterior process. Specifically, it implies the convergence in probability of $\int e^{t\sqrt{n}\int \bar{m}_{\eta}\mathrm{d}(F_\eta-F_0) }\mathrm{d}\Pi(F_\eta|Z^{(n)})$ to $e^{t\sqrt{n}\int \bar{m}_{\eta}\mathrm{d}(\mathbb{F}_n-F_0)+\frac{t^2}{2}\Vert \bar{m}_0-F_0\bar{m}_0\Vert^2_{2, F_0}}$ uniformly over $\{\bar{m}_{\eta}:\eta\in\mathcal{H}_n\}$, using the notation $F_0\bar{m}_0:= \int \bar{m}_0(x)\mathrm{d}F_0(x)$ and $\mathbb{F}_n\bar{m}_0:= 1/n\sum_{i=1}^n\bar{m}_0(X_i)$. Further, we may apply the convergence of $m_{\eta}$ imposed in Assumption \ref{Assump:Rate}, so that the above display becomes
	\begin{footnotesize}
		\begin{align*}
			e^{o_{P_0(1)}}\int_{\mathcal{H}^m_n}\! \! &\exp\Big(t\mathbb{G}_n[\gamma_0(m_0-m_{\eta})]-t\sqrt{n}b_{0, \eta}+\ell_n^m(\eta^m_t)\Big) \exp\left(t\sqrt{n}\int \bar{m}_{\eta}\mathrm{d}(\mathbb{F}_n-F_0)+\frac{t^2}{2}\Vert \bar{m}_0-F_0\bar{m}_0\Vert^2_{2, F_0} \right)\mathrm{d}\Pi(\eta^m)\\
			&=e^{o_{P_0(1)}}\exp\left(t\sqrt{n}\int \bar{m}_{0}\mathrm{d}(\mathbb{F}_n-F_0)+\frac{t^2}{2}\Vert \bar{m}_0-F_0\bar{m}_0\Vert^2_{2, F_0} \right)\\
			&\qquad\times\int_{\mathcal{H}^m_n}\exp\Big(t\mathbb{G}_n[\gamma_0(m_0-m_{\eta})-(\bar{m}_0-\bar{m}_{\eta})]-t\sqrt{n}b_{0, \eta}+\ell_n^m(\eta^m_t)\Big)\mathrm{d}\Pi(\eta^m).
		\end{align*}
			\end{footnotesize}%
We now analyze the empirical process term in the integral and examine its relationship with the bias term $b_{0, \eta}$. To do so, we calculate
\begin{small}
		\begin{align*}
			&\mathbb{G}_n[\gamma_0(m_0-m_{\eta})-(\bar{m}_0-\bar{m}_{\eta})] \\
			&=\mathbb{G}_n \left[\frac{(d-\pi_0(x))(m_0(1,x)-m_\eta(1,x))}{\pi_0(x)}-\frac{(\pi_0(x)-d)(m_0(0,x)-m_\eta(0,x))}{1-\pi_0(x)}\right]\\
			&=\sqrt n \mathbb{P}_n \left[\frac{(d-\pi_0(x))(m_0(1,x)-m_\eta(1,x))}{\pi_0(x)}-\frac{(\pi_0(x)-d)(m_0(0,x)-m_\eta(0,x))}{1-\pi_0(x)}\right]=\sqrt{n}b_{0, \eta}, \label{CenteredTerm}
		\end{align*}
	\end{small}%
	where the last line follows from the definition of the bias term, that is,  $b_{0,\eta}=	\mathbb{P}_n[\gamma_0(m_0-m_{\eta})-(\bar{m}_0-\bar{m}_{\eta})] $.

	Further, observe that
	$\mathbb{G}_n[\gamma_0\rho^{m_0}]-\mathbb{G}_n[\tilde{\tau}_0]=-\mathbb{G}_n[\bar m_0]$ and 
	$\mathbb{G}_n[\bar m_0]=\sqrt{n}\int \bar{m}_0d(\mathbb F_n-F_0)$ by the definition of the efficient influence function given in \eqref{eif_ate}. 
	As we insert these in the previous expression for $I_n(t)$,
	we obtain for all $t$ in a sufficiently small neighborhood around zero and uniformly for $\eta\in\mathcal H_n$:
	\begin{small}
		\begin{align*}
			I_n(t)&=\exp\Bigg( \underbrace{-t\mathbb{G}_n[\bar m_0]+ t\sqrt{n}\int \bar{m}_0\mathrm{d}(\mathbb F_n-F_0)}_{=0}+\frac{t^2}{2}\Big(\underbrace{P_0(B_0^m\xi^m_0)^2+\overbrace{\Vert \bar{m}_0-F_0\bar{m}_0\Vert^2_{2, F_0}}^{=P_0(B_0^f\xi^f_0)^2}}_{=P_0 (B_0\xi_0)^2}\Big)+o_{P_0}(1)\Bigg)\\
			&\qquad\times\frac{\int_{\mathcal{H}^m_n}\exp\big(\ell_n^m(\eta^m_t)\big)\mathrm{d}\Pi(\eta^m)}{\int_{\mathcal{H}^m_n}\exp\big(\ell_n^m(\eta^{m\prime})\big) \mathrm{d}\Pi(\eta^{m\prime})}=\exp\left(\frac{t^2}{2}P_0 (B_0\xi_0)^2\right)+o_{P_0}(1),
		\end{align*} 
	\end{small}%
	where the last equality follows from the prior invariance condition established in Lemma \ref{lemma:PriorInv}. This implies \eqref{laplace:transform:conv} using that $P_0 (B_0\xi_0)^2=P_0\widetilde{\tau}_0^2=\textsc v_0$ by the Lemma \ref{lemma:lfd}.
\end{proof}

\begin{proof}[Proof of Theorem \ref{thm:Debias}]
		It is sufficient to show that
		$\sup_{\eta\in\mathcal{H}_n}\left|b_{0,\eta}-\widehat{b}_{\eta}\right|=o_{P_0}(n^{-1/2})$,
	where $b_{0,\eta}=	\mathbb{P}_n[\gamma_0(m_0-m_{\eta})+\bar{m}_{\eta} -\bar{m}_0]$ and $\widehat{b}_\eta=	\mathbb{P}_n[\widehat \gamma(\widehat m - m_{\eta}) + \bar{m}_{\eta}-\widehat{\bar{m}}]$.
	We make use of the decomposition
	\begin{equation}\label{dec:proof:feasible:drb}
		b_{0,\eta}-\widehat{b}_{\eta}=\mathbb{P}_n[\gamma_0(m_0-m_{\eta})-\widehat{\gamma}\rho^{m_{\eta}}]-\mathbb{P}_n[\bar{m}_0-\widehat{\bar{m}} - \widehat{\gamma}\rho^{\widehat m}].
	\end{equation}
	Consider the first summand on the right hand side of the previous equation. We have uniformly for $\eta\in\mathcal{H}_n$:
	\begin{align*}	
		\mathbb{P}_n[\gamma_0(m_0-m_{\eta})-\widehat{\gamma}\rho^{m_{\eta}} ]
		=&-\mathbb{P}_n[\widehat{\gamma} \rho^{m_0} ]+\mathbb{P}_n[(\gamma_0-\widehat{\gamma})(m_0-m_{\eta})]\\
		=&-\mathbb{P}_n[\widehat{\gamma}\rho^{m_0} ]+o_{P_0}(n^{-1/2}),
	\end{align*}
	where the last equation follows from the following derivation:
\begin{align*}
		\sqrt n\sup_{\eta\in\mathcal{H}_n}&\left|\mathbb{P}_n[(\gamma_0-\widehat{\gamma})(m_0-m_{\eta})]\right|
		\leq 		\sup_{\eta\in\mathcal{H}_n}\left|\mathbb{G}_n[(\gamma_0-\widehat{\gamma})(m_0-m_{\eta})]\right|\\
		&\qquad\qquad\qquad\qquad\qquad\qquad+\sqrt n\sup_{\eta\in\mathcal{H}_n}\left|P_0[(\gamma_0-\widehat{\gamma})(m_0-m_{\eta})]\right|\\
		&\qquad\qquad\leq o_{P_0}(1)+	O_{P_0}(1)\times \sqrt{n}\Vert \pi_0 - \widehat{\pi}\Vert_{2, F_0}\sup_{\eta\in\mathcal{H}_n}\Vert m_\eta-m_0\Vert_{2, F_0}=o_{P_0}(1),
	\end{align*} 
using the Cauchy-Schwarz inequality, Assumption \ref{Assump:Rate}, and Assumption \ref{Assump:Donsker}.
	Consider the second summand on the right hand side of \eqref{dec:proof:feasible:drb}. 
From Lemma \ref{lemma:FreqDR} we infer
	\begin{align*}
		\mathbb{P}_n[\widehat{\bar{m}} + \widehat{\gamma}\rho^{\widehat m} - \bar{m}_0]=\mathbb{P}_n[\gamma_0\rho^{m_0}]+o_{P_0}(n^{-1/2}).
	\end{align*}
	Consequently, decomposition \eqref{dec:proof:feasible:drb} together with the asymptotic expansion of each summand yields
	\begin{align*}
				\sup_{\eta\in\mathcal{H}_n}\left|b_{0,\eta}-\widehat{b}_{\eta}\right|\leq \left|\mathbb{P}_n[(\gamma_0-\widehat \gamma)\rho^{m_0}]\right|+o_{P_0}(n^{-1/2})
=o_{P_0}(n^{-1/2}),
	\end{align*}
	where the last equation is due to the equation \eqref{lemma:negligible:bound1}.
\end{proof}

\begin{proof}[Proof of Corollary \ref{cor:CI}]
The weak convergence of the Bayesian point estimator directly follows from our asymptotic characterization of the posterior and the
argmax theorem; see the proof of Theorem 10.8 in \cite{van1998asymptotic}. The corrected Bayesian credible set $\mathcal{C}_n(\alpha)$ satisfies $\Pi(\check{\tau}_\eta\in \mathcal{C}_n(\alpha)\mid Z^{(n)})=1-\alpha$ for any $\alpha\in(0,1)$. In particular,  we have
\begin{align*}
\Pi\left(\sqrt{n/\textsc v_0}(\tau_\eta-\widehat{\tau}-\widehat{b}_{\eta})\in \sqrt{n/\textsc v_0}(\mathcal{C}_n(\alpha)-\widehat{\tau})\mid Z^{(n)}\right)=1-\alpha.
\end{align*}
Now the definition of the estimator $\widehat{\tau}$ given in \eqref{def:est:chi} yields $\sqrt n\widehat{\tau}=\sqrt n\big(\tau_0+\mathbb P_n\widetilde{\tau}_0\big)+o_{P_0}(1)$. For any set $A$, we write $\mathbb{N}(A):=\int_{A}e^{-u^2/2}/\sqrt{2\pi}\,\mathrm{d}u$.
Theorem \ref{thm:BvM} implies
\begin{align*}
\mathbb{N}\left(\sqrt{n/\textsc v_0}(\mathcal{C}_n(\alpha)-\tau_0-\mathbb{P}_n\widetilde{\tau}_0)\right)\to_{P_0} 1-\alpha.
\end{align*}
We may thus write $\mathcal{C}_n(\alpha)=\sqrt{\textsc v_0/n} \, \mathcal A_n(\alpha)+\tau_0+\mathbb{P}_n\widetilde{\tau}_0+o_{P_0}(1)$ for some set $\mathcal A_n(\alpha)$ satisfying $\mathbb{N}(\mathcal A_n(\alpha))\to_{P_0}  1-\alpha$. 
	Therefore, the frequentist coverage of the Bayesian credible set is 
	\begin{align*}
	P_0\left(\tau_0\in \mathcal{C}_n(\alpha)\right)=P_0\left(\tau_0\in  \sqrt{\textsc v_0/n} \, \mathcal A_n(\alpha)+\tau_0+\mathbb{P}_n\widetilde{\tau}_0\right)
	=P_0\left(-\frac{\mathbb{G}_n\widetilde{\tau}_0}{\sqrt \textsc v_0}\in \mathcal A_n(\alpha)\right)
	\to  1-\alpha,
	\end{align*}
noting that $\mathbb{G}_n\widetilde{\tau}_0$ is asymptotically normal with mean zero and variance $\textsc v_0$ under $P_0$.
\end{proof}

\begin{proof}[Proof of Proposition \ref{prop:exponential}]
Note that $\widehat\gamma$ is based on an auxiliary sample and hence we can treat $\widehat\gamma$ below as a deterministic function denoted by $\gamma_n$ satisfying the rate restrictions $\|\gamma_n\|_\infty=O(1)$ and $\|\gamma_n-\gamma_0\Vert_\infty= O\big((n/\log n)^{-s_\pi/(2s_\pi+p)}\big)$. Regarding the conditional mean functions, we consider the set $\mathcal{H}_{n,d}^m:= \left\{w_d+\lambda\gamma_n:(w_d,\lambda)\in \mathcal{W}_{n,d}\right\}$, where for $d\in\{1,0\}$ and some constant $C>0$:
{\small \begin{align}\label{WnSet}
		\mathcal{W}_{n,d}:= \left\{(w_d,\lambda):w_d\in \mathcal{B}^m_n,|\lambda|\leq C\sigma_n\sqrt{n}\varepsilon_n \right\}\cap \left\{(w_d,\lambda):\Vert \Psi(w_d(\cdot)+\lambda\gamma_n)-m_0(d,\cdot)\Vert_{2, F_0}\leq \varepsilon_n \right\},
\end{align}}%
where $\mathcal{B}_n^m$ in the first restriction for the Gaussian process $W(d,\cdot)$ is a regularity class of functions defined in the equation (\ref{SieveSet}) in the online Supplementary Appendix \ref{appendix:auxiliary}. We write $\mathcal{H}^m_n=\mathcal{H}^m_{n,1}\times \mathcal{H}^m_{n,0}$.

We first verify Assumption \ref{Assump:Rate} with $\varepsilon_n=n^{-s_m/(2s_m+p)}(\log n)^{s_m(p+1)/(2s_m+p)}$. The posterior contraction rate is shown in our Lemma \ref{lemma:RatesGP}. Referring to the product rate condition, that is, $\sqrt{n}\varepsilon_nr_n=o(1)$ for $r_n\sim  (n/\log n)^{-s_\pi/(2s_\pi+p)}$. This is satisfied if $2s_m/(2s_m+p)+2s_\pi/(2s_\pi+p)>1$, which can be rewritten as $\sqrt{s_\pi \, s_m}>p/2$.

We now verify Assumption \ref{Assump:Donsker}. It is sufficient to deal with the resulting empirical process $\mathbb{G}_n$.
Note that the Cauchy-Schwartz inequality implies
\begin{align*}
|P_0(m_\eta-m_0)|&=|\E_0[D(m_\eta(1,X)-m_0(1,X))]+\E_0[(1-D)(m_\eta(0,X)-m_0(0,X))]|\\
&\leq \sqrt{\E_0[(m_\eta(1,X)-m_0(1,X))^2]}+\sqrt{\E_0[(m_\eta(0,X)-m_0(0,X))^2]}\\
&= \Vert m_\eta(1,\cdot)-m_0(1,\cdot)\Vert_{2, F_0}+\Vert m_\eta(0,\cdot)-m_0(0,\cdot)\Vert_{2, F_0}.
\end{align*}
Consequently, from Lemma \ref{lemma:product} we infer
\begin{align*}
	\E_0\sup_{\eta\in\mathcal{H}^m_n}&\left|\mathbb{G}_n[(\gamma_n-\gamma_0)(m_\eta-m_0)] \right|\leq 4\Vert \gamma_n-\gamma_0\Vert_{\infty}\E_0\sup_{\eta\in\mathcal{H}^m_n}\left|\mathbb{G}_n[m_\eta-m_0] \right|\\
	&\qquad+\Vert \gamma_n-\gamma_0\Vert_{2, F_0}\sup_{\eta\in\mathcal{H}_n}\Big(\Vert m_\eta(1,\cdot)-m_0(1,\cdot)\Vert_{2, F_0}+\Vert m_\eta(0,\cdot)-m_0(0,\cdot)\Vert_{2, F_0}\Big)\\
	&\lesssim  (n/\log n)^{-s_\pi/(2s_\pi+p)}\E_0\sup_{\eta\in\mathcal{H}^m_n}\left|\mathbb{G}_n[m_\eta-m_0] \right|+(n/\log n)^{-s_\pi/(2s_\pi+p)}\varepsilon_n\\
	&=  (n/\log n)^{-s_\pi/(2s_\pi+p)}\E_0\sup_{\eta\in\mathcal{H}^m_n}\left|\mathbb{G}_n[m_\eta-m_0] \right|+o(1).
\end{align*}
Note that if $s_m>p/2$,  from   Lemma \ref{lemma:DClass} we infer $\E_0\sup_{\eta\in\mathcal{H}^m_n}\mathbb{G}_n\left[m_\eta-m_0\right]=o(1)$.
Thus it remains to consider the case $s_m\leq p/2$. By the entropy bound presented in the proof of Lemma \ref{lemma:RatesGP}, we have
$\log N(\varepsilon_n,\mathcal{H}^m_n,L^2(F_0))\lesssim \varepsilon_n^{-2\upsilon}$,
with $\upsilon = p/(2s_m)$ modulo some $\log n$ term on the right hand of the bound. Because $\Psi(\cdot)$ is monotone and Lipschitz, a set of $\varepsilon$-covers in $L^2(F_0)$ for $\eta^m\in\mathcal{H}_n^m$ translates into a set of $\varepsilon$-covers for $m_{\eta}$. In this case, the empirical process bound of \citep[p.2644]{han2021set} yields
\begin{equation*}
	\E_0\sup_{\eta\in\mathcal{H}^m_n}\left|  \mathbb{G}_n[m_\eta-m_0] \right|\lesssim L_n n^{(\upsilon-1)/(2\upsilon)}=O(L_nn^{1/2-s_m/p}),
\end{equation*}
where $L_n$ represents a term that diverges at certain polynomial order of $\log n$. Consequently, we obtain
\begin{align*}
(n/\log n)^{-s_\pi/(2s_\pi+p)}\E_0\sup_{\eta\in\mathcal{H}^m_n}\left|\mathbb{G}_n[m_\eta-m_0] \right|=o(1),
\end{align*}
which is satisfied under the smoothness restriction $-s_\pi/(2s_\pi+p) +1/2 -s_m/p<0$
or equivalently $4s_\pi s_m +2ps_m>p^2$.  This condition  automatically holds given $\sqrt{s_\pi \, s_m}>p/2$.

Finally, it remains to verify Assumption \ref{Assump:Prior}. By the univariate Gaussian tail bound, the prior mass of the set $\Lambda_n:= \{\lambda:|\lambda|>u_n\sigma^2_n\sqrt{n} \}$ satisfies $\Pi(\lambda\in\Lambda_n)\leq  2\exp(-u_n^2\sigma^2_n n/2)$. Also, the Kullback-Leibler neighborhood around $\eta_0^m$ has prior probability at least $e^{-n\varepsilon_n^2}$. We may thus apply  Lemma 4 of \cite{ray2020causal}, which yields $\Pi(\lambda\in\Lambda_n\mid Z^{(n)})\to_{P_0}0$, as imposed in Assumption \ref{Assump:Prior}(i).
 
 Regarding Assumption \ref{Assump:Prior}(ii), we need to show the posterior probability of the shifted version of $\mathcal{H}_n^m$ is tending to one. Considering $\mathcal{H}_n^m$ itself, the first set in the intersection of \eqref{WnSet} that defines $\mathcal{W}_{n,d}$ is seen to have posterior probability tending to one by the result in (II) of Lemma \ref{lemma:RatesGP}, combined with the univariate Gaussian tail probability bound 
 \begin{equation*}
 	\Pi(|\lambda|\geq C\sigma_n\sqrt{n}\varepsilon_n)\leq 2\exp(-Cn\varepsilon_n^2/2).
 \end{equation*}
 The second set in the intersection of \eqref{WnSet} has posterior probability tending to one by Lemma 17 of \cite{ray2020causal}. Hence, $\mathcal{H}_n^m$ has posterior probability going to one. Next, we consider $\mathcal{H}^m_n+t\gamma_n/\sqrt{n}$, for any $t\in\mathbb{R}$. To slightly abuse the notation, we write $\eta^m_d=w_d+\lambda\gamma_n$ for $d\in\{0,1\}$ in the sequel. By the Lipschitz continuity of the Logistic link function, we have 
$	\Vert \Psi(\eta_d^m)-\Psi(\eta_d^m+t\gamma_n/\sqrt{n})\Vert_{2, F_0}\leq |t|\Vert\gamma_n\Vert_{\infty}/\sqrt{n}$
 for $d\in\{0,1\}$. Therefore, we get $\mathcal{H}_{n,d}^m+t\gamma_n/\sqrt{n}\supset \Xi_{n,d,t}$ with probability $P_0$ approaching one, where 
\begin{align*}
\Xi_{n,d,t} :=\left\{\eta_d^m:\Vert\eta_d^m\Vert_{\mathbb{H}}\leq C\sqrt{n}\tilde{\varepsilon}_n, \Vert \Psi(\eta_d^m)-m_0(d,\cdot)\Vert_{2, F_0}\leq \tilde{\varepsilon}_n\right\}
\end{align*}
with $\tilde{\varepsilon}_n:= C\varepsilon_n-|t|\Vert\gamma_n\Vert_\infty/\sqrt{n} $
and $\Vert\cdot\Vert_{\mathbb{H}}$ denotes the norm of the Reproducing Kernel Hilbert Space associated with the squared exponential process; see online supplementary Appendix \ref{appendix:auxiliary} for a formal definition. Because $\sqrt{n}\varepsilon_n\to \infty$ and $\Vert \gamma_n\Vert_{\infty}=O(1)$, the posterior probability of $\Xi_{n,d,t}$ tends to one following similar arguments concerning the set $\mathcal{H}^m_n$, after replacing $\varepsilon_n$ with a multiple of itself for $d\in\{0,1\}$. Hence, the posterior probability of $\mathcal{H}_n^m+t\gamma_n/\sqrt{n}$ is seen to tend to one, which completes the proof.
\end{proof}

\section{Key Lemmas}\label{appendix:key:lemmas}
We now present key lemmas used in the derivation of our BvM Theorem. 
We introduce $\eta_u:= (\eta^\pi,\eta_u^{m})$ where
\begin{equation}\label{path}
	\eta_u^m=\eta^m-tu\xi_0^m/\sqrt{n}, ~~\text{for}~~u\in[0,1].
\end{equation}
This defines a path from  $\eta_{u=0}=(\eta^\pi,\eta^{m})$ to $\eta_{u=1}=(\eta^\pi,\eta_t^m)$. We also write $g(u):= \log p_{\eta^m_u}$, for $u\in[0,1]$, so that $\log p_{\eta^m}-\log p_{\eta_t^m}=g(0)-g(1)$, cf. the proof of Theorem 1 in \cite{ray2020causal}.
\begin{lemma}\label{lemma:likelihood}
	Let Assumptions \ref{Ass:unconfounded} and \ref{Assump:Rate} hold. Then, we have uniformly for $\eta\in\mathcal H_n$:
	\begin{equation*}
		\ell_n^m(\eta^m)-\ell_n^m(\eta^m_t)=t\mathbb{G}_n[\gamma_0\rho^{m_0}]+t\mathbb{G}_n[\gamma_0(m_0-m_{\eta}) ]+t\sqrt{n}\int (\bar{m}_0-\bar{m}_{\eta})\mathrm{d}F_0 +\frac{t^2}{2}P_0(B_0^m\xi^m_0)^2+o_{P_0}(1).
	\end{equation*}
\end{lemma}
\begin{proof}
We start with the following decomposition:
	\begin{align*}
		\ell_n^m(\eta^m)-\ell_n^m(\eta^m_t)=t\mathbb{G}_n[\gamma_0\rho^{m_0}]+\underbrace{\sqrt{n}\mathbb{G}_n[\log p_{\eta^m}-\log p_{\eta^m_t}-\frac{t}{\sqrt{n}}\gamma_0\rho^{m_0}]}_{\text{Stochastic Equicontinuity}}		+\underbrace{n P_0[\log p_{\eta^m}-\log p_{\eta^m_t}]}_{\text{Taylor Expansion}}.
	\end{align*}
	From the calculation in the proof of Lemma \ref{lemma:Taylor}, we have $g'(0)=-\frac{t}{\sqrt{n}}\gamma_0\rho^{m_0}+\frac{t}{\sqrt{n}}\gamma_0(m_\eta-m_0)$. Then, we infer for the stochastic equicontinuity term that
	\begin{align*}
	\sqrt{n}\mathbb{G}_n[\log p_{\eta^m}-\log p_{\eta^m_t}-\frac{t}{\sqrt{n}}\gamma_0\rho^{m_0}]
	+t\mathbb{G}_n[\gamma_0(m_\eta-m_0)]=o_{P_0}(1),
	\end{align*}
uniformly in $\eta^m\in\mathcal{H}^m_n$. We can thus write uniformly in $\eta^m\in\mathcal{H}^m_n$:
	\begin{align*}
	\ell_n^m(\eta^m)-\ell_n^m(\eta^m_t)&=t\mathbb{G}_n[\gamma_0\rho^{m_0}]+t\mathbb{G}_n[\gamma_0(m_0-m_\eta)]+n P_0[\log p_{\eta^m}-\log p_{\eta^m_t}]+o_{P_0}(1).
	\end{align*}
The rest of the proof involves a standard Taylor expansion for the third term on the right hand side of the above equation. By the equation (\ref{TaylorExp}) in the proof of Lemma \ref{lemma:Taylor}, we get
\begin{align*}
	-nP_0g'(0)=t\sqrt{n}P_0[\gamma_0 \rho^{m_0}]+t\sqrt{n}P_0[\gamma_0(m_0-m_\eta)]=t\sqrt{n}\int (\bar{m}_0-\bar{m}_{\eta})\mathrm{d}F_0,
\end{align*}
by the fact that $P_0[\gamma_0 \rho^{m_0}]=0$ and the definition of the Riesz representer $\gamma_0$ in \eqref{riesz:def}. Regarding the second-order term in the Taylor expansion in the equation (\ref{2ndTaylor}) of the proof of Lemma \ref{lemma:Taylor}, we get
\begin{equation*}
	g^{(2)}(0)=-\frac{t^2}{n}\gamma_0^2 m_0(1-m_0)-\frac{t^2}{n}\gamma_0^2(m_\eta(1-m_\eta)-m_0(1-m_0)).
\end{equation*}
Considering the score operator $B_0^m=B_{\eta_0}^m$ defined in \eqref{score_ate}, we have 
\begin{align*}
P_0(B_0^m\xi^m_0)^2&=\E_0\left[\gamma_0^2(D,X)(Y-m_0(D,X))^2\right]\\
&=\E_0\left[\frac{D}{\pi_0^2(X)}(Y(1)-m_0(1,X))^2\right]+\E_0\left[\frac{1-D}{(1-\pi_0(X))^2}(Y(0)-m_0(0,X))^2\right].
\end{align*}
Consequently, by the unconfoundedness imposed in Assumption \ref{Ass:unconfounded}(i) and the binary nature of $Y$, we have $\E_0[Y(d)^2|D=d,X=x]=\E_0[Y(d)|D=d,X=x]=m_0(d,x)$. We thus obtain
\begin{align*}
P_0(B_0^m\xi^m_0)^2&=\E_0\left[\frac{D}{\pi_0^2(X)}m_0(1,X)(1-m_0(1,X))\right]+\E_0\left[\frac{1-D}{(1-\pi_0(X))^2}m_0(0,X)(1-m_0(0,X))\right]\\
&=P_0[\gamma_0^2m_0(1-m_0)].
\end{align*}
Then, by employing Assumption \ref{Ass:unconfounded}(ii), that is, $\bar\pi <\pi_0(x)< 1-\bar\pi$ for all $x$, it yields uniformly for $\eta\in\mathcal H_n$:
\begin{align*}
-&n P_0g^{(2)}(0)-t^2P_0(B_0^m\xi^m_0)^2=t^2P_0[\gamma_0^2(m_\eta(1-m_\eta)-m_0(1-m_0))]\\
&=t^2P_0[\gamma_0^2(m_\eta-m_0)(1-m_0)]+t^2P_0[\gamma_0^2m_\eta (m_0-m_\eta)]\\
&\leq 2t^2\,\E_0\left[\frac{D}{\pi_0^2(X)}\big|m_\eta(1,X)-m_0(1,X)\big|\right]+2t^2\,\E_0\left[\frac{1-D}{(1-\pi_0(X))^2}\big|m_\eta(0,X)-m_0(0,X)\big|\right]\\
	&\leq  \frac{2t^2}{\bar\pi^2}\Big(\Vert m_\eta(1,\cdot)-m_0(1,\cdot)\Vert_{2, F_0}+\Vert m_\eta(0,\cdot)-m_0(0,\cdot)\Vert_{2, F_0}\Big)=o_{P_0}(1),
\end{align*}
where the last equation is due to the posterior contraction rate of the conditional mean function $m(d,\cdot)$ imposed in  Assumption \ref{Assump:Rate}.
 Consequently, we obtain, uniformly for $\eta\in\mathcal H_n$, 
\begin{align*}
	n P_0[\log p_{\eta^m}-\log p_{\eta^m_t}] &= -n (P_0g^\prime(0)+ P_0g^{(2)}(0))+o_{P_0}(1)\\
	&=t^2P_0(B_0^m\xi^m_0)^2+t\sqrt{n}\int (\bar{m}_0-\bar{m}_{\eta})\,\mathrm{d} F_0+o_{P_0}(1),
\end{align*}
which leads to the desired result.
\end{proof}

The next lemma verifies the prior stability condition under our double robust smoothness conditions.
\begin{lemma}\label{lemma:PriorInv}
	Let Assumptions \ref{Ass:unconfounded}--\ref{Assump:Prior} hold. Then we have
		\begin{equation}\label{lemma:PriorInv:term1}
		\frac{\int_{\mathcal{H}^m_n}\exp\big(\ell_n^m(\eta^m_t)\big)\mathrm{d}\Pi(\eta^m)}{\int_{\mathcal{H}^m_n}\exp\big(\ell_n^m(\eta^{m\prime})\big) \mathrm{d}\Pi(\eta^{m\prime})}\to_{P_0}1,
	\end{equation} 			
	for a sequence of measurable sets $\mathcal{H}^m_n$ such that $\Pi(\eta^m\in\mathcal{H}^m_n|Z^{(n)})\to_{P_0}1$.
\end{lemma}
\begin{proof}
Since $\widehat\gamma$ is based on an auxiliary sample, it is sufficient to consider deterministic functions $\gamma_n$ with the same rates of convergence as $\widehat\gamma$. Denote the corresponding propensity score by $\pi_n$. By Assumption \ref{Assump:Prior}, we have $\lambda\sim N(0,\sigma_n^2)$ and 
	\begin{equation}\label{lemma:PriorInv:term2}
		\frac{\int_{\mathcal{H}^m_n}\exp\big(\ell_n^m(\eta^m_t)\big)\mathrm{d}\Pi(\eta^m)}{\int_{\mathcal{H}^m_n}\exp\big(\ell_n^m(\eta^{m\prime})\big) \mathrm{d}\Pi(\eta^{m\prime})}
		=\frac{\int_{\Theta_n}e^{\ell_n^m(w+\lambda\gamma_n-t \gamma_0/\sqrt{n})}\phi_{\sigma_n}(\lambda)\,\mathrm{d}\lambda \mathrm{d}\Pi(w) }{\int_{\Theta_n}e^{\ell_n^m(w+\lambda \gamma_n)}\phi_{\sigma_n}(\lambda)\, \mathrm{d}\lambda \mathrm{d}\Pi(w) }+o_{P_0}(1),
	\end{equation}
	where $\phi_{\sigma_n}$ denotes the probability density function of a $N (0,\sigma_n^2)$ random variable and the set $\Theta_n$ is defined by
		$\Theta_{n}=\left\{(w, \lambda): w+\lambda \gamma_n \in \mathcal{H}_{n}^m,|\lambda| \leq 2 u_{n} \sigma_{n}^{2} \sqrt{n}\right\}$
	where $u_n\to0$ imposed in Assumption \ref{Assump:Prior} and $u_n n\sigma_n^2\to \infty$.

	Considering the log likelihood ratio of two normal densities together with the constraint $|\lambda| \leq 2 u_{n} \sigma_{n}^{2} \sqrt{n}$, it is shown on page 3015 of \cite{ray2020causal} that
	\begin{equation*}
		\left|\log \frac{\phi_{\sigma_n}(\lambda)}{\phi_{\sigma_n}(\lambda-t/\sqrt{n})}  \right|\leq \frac{|t\lambda|}{\sqrt{n}\sigma_n^2}+\frac{t^2}{2n\sigma_n^2}\to 0.
	\end{equation*}
We show at the end of the proof that
$	\left|\ell_n^m(w+\lambda \gamma_n-t \gamma_0/\sqrt{n})-\ell_n^m(w+\lambda \gamma_n-t \gamma_n/\sqrt{n})   \right|=o_{P_0}(1)$,
uniformly for $(w,\lambda)\in \Theta_n$. 
Consequently, the numerator of this leading term in \eqref{lemma:PriorInv:term2} becomes
	\begin{equation*}
	\int_{\Theta_n}e^{\ell_n^m(w+\lambda \gamma_n-t \gamma_0/\sqrt{n})}\phi_{\sigma_n}(\lambda)\,\mathrm{d}\lambda \mathrm{d}\Pi(w)=	e^{o_{P_0}(1)}\int_{\Theta_n}e^{\ell_n^m(w+\gamma_n(\lambda -t /\sqrt{n}))}\phi_{\sigma_n}(\lambda-t/\sqrt{n})\,\mathrm{d}\lambda \mathrm{d}\Pi(w).
	\end{equation*}
	By the change of variables $\lambda-t/\sqrt{n}\mapsto \lambda^\prime$ on the numerator and using the notation $\Theta_{n,t}=\left\{(w, \lambda): (w,\lambda+ t /\sqrt n) \in \Theta_n\right\}$, the prior invariance property becomes
	\begin{equation*}
		e^{o_{P_0}(1)}\frac{\int_{\Theta_{n,t}}e^{\ell_n^m(w+\lambda^\prime \gamma_n)}\phi_{\sigma_n}(\lambda^\prime)\,\mathrm{d}\lambda^\prime \mathrm{d}\Pi(w) }{\int_{\Theta_n}e^{\ell_n^m(w+\lambda \gamma_n)}\phi_{\sigma_n}(\lambda)\,\mathrm{d}\lambda \mathrm{d}\Pi(w) }=e^{o_{P_0}(1)}\frac{\Pi(\Theta_{n,t}|Z^{(n)})}{\Pi(\Theta_{n}|Z^{(n)})}.
	\end{equation*}
	The desired result would follow from $\Pi(\Theta_{n}|Z^{(n)})=1-o_{P_0}(1)$ and $\Pi(\Theta_{n,t}|Z^{(n)})=1-o_{P_0}(1)$. The first convergence directly follows from Assumption \ref{Assump:Prior}. The set $\Theta_{n,t}$ is the intersection of these two conditions in Assumption \ref{Assump:Prior}, except that the restriction on $\lambda$ in $\Theta_{n,t}$ is $\left|\lambda+t/\sqrt{n} \right|\leq 2u_n\sqrt{n}\sigma_{n}^2$ instead of $\left|\lambda\right|\leq u_n\sqrt{n}\sigma_{n}^2$. By construction, we have $t/\sqrt{n}=o(u_n\sqrt{n}\sigma_{n}^2)$, so that $\Pi(\Theta_{n,t}|X^{(n)})=1-o_{P_0}(1)$.
	
We complete the proof by establishing the following result:
	\begin{equation}\label{bound:likelihood}
		\sup_{\eta^m\in \mathcal{H}_n^m}\left|\ell_n^m(\eta^m-t \gamma_n/\sqrt{n})-\ell_n^m(\eta^m-t \gamma_0/\sqrt{n})   \right|=o_{P_0}(1).
	\end{equation}
	We denote $\eta^m_{n,t}= \eta^m-t \gamma_n/\sqrt{n}$ and $\eta_t^m=\eta^m-t \gamma_0/\sqrt{n}$. Consider the following decomposition of the log-likelihood:
	\begin{align*}
		\ell_n^m(\eta^m_{n,t})-\ell_n^m(\eta^m_{t})&=\ell_n^m(\eta^m_{n,t})-\ell_n^m(\eta^m)+\ell_n^m(\eta^m)-\ell_n^m(\eta^m_{t})\\
		&=n\mathbb{P}_n[\log p_{\eta^m_{n,t}}-\log p_{\eta^m}]+n\mathbb{P}_n[\log p_{\eta^m}-\log p_{\eta^m_{t}}].
	\end{align*}
	Next, we apply third-order Taylor expansions in Lemma \ref{lemma:Taylor} separately to the two terms in the brackets of the above display making use of the notation $\rho^m(y,d,x)=y-m(d,x)$:
	\begin{align*}
		n\mathbb{P}_n[\log p_{\eta^m_{n,t}}-\log p_{\eta^m}]&=-t\sqrt{n}\mathbb{P}_n\left[\gamma_n\rho^{m_{\eta}}\right]-\frac{t^2}{2}\mathbb{P}_n\left[\gamma_n^2m_{\eta}\left(1-m_{\eta}\right)\right]-\frac{t^3}{6\sqrt n}\mathbb{P}_n\left[\gamma_n^3\Psi^{(2)}(\eta_{u^*}^m)\right],\\
		n\mathbb{P}_n[\log p_{\eta^m}-\log p_{\eta^m_{t}}]&=t\sqrt{n}\mathbb{P}_n\left[\gamma_0\rho^{m_{\eta}}\right]+\frac{t^2}{2}\mathbb{P}_n\left[\gamma_0^2m_{\eta}\left(1-m_{\eta}\right)\right]+\frac{t^3}{6\sqrt n}\mathbb{P}_n\left[\gamma_0^3\Psi^{(2)}(\eta_{u^{**}}^m)\right],
	\end{align*}
for some intermediate points $u^{*}, u^{**}\in(0,1)$, cf. the equation (\ref{path}).
	Combining the previous calculation yields
	\begin{align*}
		\ell_n^m(\eta_{n,t})-\ell_n^m(\eta_{t})		&=t\sqrt n\mathbb{P}_n[(\gamma_0-\gamma_n)\rho^{m_{\eta}}]-\frac{t^2}{2}\mathbb{P}_n[dm_{\eta}(1-m_{\eta})(\gamma_n^2-\gamma_0^2)]\\
		&\quad +\frac{t^3}{6\sqrt n}\mathbb{P}_n\left[(\gamma_0^3-\gamma_n^3)\left(\Psi^{(2)}(\eta_{u^{**}}^m)-\Psi^{(2)}(\eta_{u^{*}}^m)\right)\right]=:T_1+T_2+T_3.
	\end{align*}
In order to control $T_1$, we evaluate 
	\begin{align*}
		T_1=t\mathbb{G}_n[(\gamma_0-\gamma_n)\rho^{m_0}]+t\mathbb{G}_n[(\gamma_0-\gamma_n)(m_0-m_{\eta})]+t\sqrt{n}P_0[(\gamma_0-\gamma_n)\rho^{m_\eta}].
	\end{align*}
	Note that the first term is centered, so it becomes $t\sqrt{n}\mathbb{P}_n[(\gamma_0-\gamma_n)\rho^{m_0}]$. We apply Lemma \ref{lemma:negligible} to conclude that it is of smaller order. The middle term is negligible by our Assumption \ref{Assump:Donsker}. Referring to the last term, the Cauchy–Schwarz inequality yields
	\begin{align*}
		&\sup_{\eta\in\mathcal{H}_n}\left|\sqrt{n}P_0[(\gamma_n-\gamma_0)(m_\eta-m_0)]\right|\\
		&\lesssim \sqrt{2\, n} \, \Vert\pi_n-\pi_0\Vert_{2, F_0}\sup_{\eta\in\mathcal{H}_n}\Big(\Vert m_\eta(1,\cdot)-m_0(1,\cdot)\Vert_{2, F_0}+\Vert m_\eta(0,\cdot)-m_0(0,\cdot)\Vert_{2, F_0}\Big)=o_{P_0}(1),
	\end{align*}
	where the last equality is due to Assumption \ref{Assump:Rate}. We thus obtain $T_1=o_{P_0}(1)$ uniformly in $\eta\in\mathcal{H}^m_n$. Consider $T_2$. We note that $\|m_{\eta}(1-m_{\eta})\|_\infty\leq 1$ uniformly in  $\eta\in\mathcal{H}^m_n$. Hence, we obtain
	\begin{equation*}
		P_0|T_2|\leq \frac{t^2}{2}P_0|\gamma_n^2-\gamma_0^2|=\frac{t^2}{2}P_0[(\gamma_n-\gamma_0)(\gamma_n+\gamma_0)]\lesssim \frac{t^2}{2}\|\pi_n-\pi_0\|_{2, F_0}\to0
	\end{equation*}
	as $\pi_n\to\pi_0$ in $L^2(F_0)$-norm by Assumption \ref{Assump:Rate}. Thus, $T_2=o_{P_0}(1)$ uniformly in $\eta\in\mathcal{H}_n$. 
	Finally, we control $T_3$ by evaluating		
	$|T_3|\lesssim 
		t^3 n^{-1/2}\mathbb{P}_n(\Vert\gamma_n\Vert_{\infty}^3+\Vert\gamma_0\Vert_{\infty}^3)=o_{P_0}(1)$
	uniformly in  $\eta\in\mathcal{H}^m_n$, which shows \eqref{bound:likelihood}. 
\end{proof}

\bibliographystyle{abbrvnat}

\bibliography{Bayes_bib}

\newpage
$\,$
\setcounter{page}{1}
\vskip 1cm

\begin{center}
{\LARGE Supplement to ``Double Robust Bayesian Inference on Average Treatment Effects"\par\vspace{0.6\baselineskip}}
\end{center}
\begin{center}
{ \large Christoph Breunig}
  \qquad \quad{ \large Ruixuan Liu}
    \qquad \quad{ \large Zhengfei Yu}

\end{center}
\begin{center}
{ \large\today}
\end{center}

This Supplemental Material contains materials to support our main paper. Appendix \ref{appendix:auxiliary} collects some auxiliary results. Appendix \ref{appendix:proof_exponent} collects the proofs for lemmas in Section \ref{sec:extension} of the main paper.
Appendix \ref{appendix:lfd} provides least favorable directions for other causal parameters of interest besides the ATE. Appendix \ref{appendix:exponent} states and proves the BvM theorem for outcome variables belonging to one-parameter exponential family described in Section \ref{sec:extension} of the main paper. 
Appendix \ref{appendix:LPA} describes how to draw the posterior of the conditional mean function using the Laplace approximation. Appendix \ref{appendix:simu} presents additional simulation evidence. 

In this supplement, $C>0$ denotes a generic constant, whose value might change line by line. We introduce additional subscripts when there are multiple constant terms in the same display. For two sequences $a_n,b_n$, we write $a_n\lesssim b_n$, if $a_n\leq C b_n$.

\appendix

\addtocounter{section}{2} 

\section{Auxiliary Results} \label{appendix:auxiliary}
The part in the likelihood associated with the component $\eta^m=\Psi^{-1}(m_\eta)$ is given by
\begin{equation*}
	p_{\eta^m}(z)=m_\eta(d,x)^y(1-m_\eta(d,x))^{1-y},
\end{equation*}
with the corresponding log-likelihood version $\ell_n^m(\eta^m)=\sum_{i=1}^n\log	p_{\eta^m}(Z_i)$. In other words, $p_{\eta^m}(\cdot)$ is the density with respect to the dominating measure
\begin{equation}\label{DominateMeasure}
	\mathrm{d}\nu(x,d,y)=(\pi_0(x))^d(1-\pi_0(x))^{1-d}\mathrm{d}\vartheta(d,y)\mathrm{d}F_0(x),
\end{equation}
where $\vartheta$ stands for the counting measure on $\left\{\{0,0\},\{0,1\},\{1,0\},\{1,1\}\right\}$. 
For two generic probability densities $p$ and $q$, the \textit{Kullback-Leibler divergence} is defined as $K(p,q)=\int p\log(q/p)d\nu $, and the  \textit{Kullback-Leibler variation} as $V(p,q)=\int p|\log(q/p)|^2d\nu $; see Appendix B in \cite{ghosal2017fundamentals}.
Recall the notation $\rho^m(y,d,x)=y-m(d,x)$ used below.

\begin{lemma}\label{lemma:Taylor}
Let  Assumption \ref{Ass:unconfounded}  be satisfied and $m_\eta=\Psi(\eta^m)$, then we have uniformly for $\eta^m\in\mathcal H_n^m$:
\begin{equation*}
	\log p_{\eta^m}-\log p_{\eta_t^m}=\frac{t}{\sqrt{n}}\gamma_0\rho^{m_\eta}+\frac{t^2}{2n}\gamma_0^2m_\eta(1-m_\eta)+R_n,
\end{equation*}
for some function $R_n$ with $\|R_n\|_\infty\lesssim n^{-3/2}$.
\end{lemma}
\begin{proof}
The logistic distribution function $\Psi$ satisfies $\Psi^\prime=\Psi(1-\Psi)$ and $\Psi^{(2)}=\Psi(1-\Psi)(1-2\Psi)$.
	Recall the perturbation of $\eta^m$ along the least favorable direction in \eqref{def:pert} given by
$		\eta_t^m=\eta^m-t\xi_0^m/\sqrt{n}$ for $t\in\mathbb R$.
Thus, $\log p_{\eta^m}-\log p_{\eta_t^m}=g(0)-g(1)$, where $g(u)= \log p_{\eta^m_u}$ for $u\in[0,1]$, as introduced at the beginning of Section \ref{appendix:key:lemmas}. 
We examine the following Taylor expansion uniformly for $\eta^m\in\mathcal H_n^m$:
	\begin{equation}\label{taylor:eq}
		g(0)-g(1)=-g^\prime(0)-g^{(2)}(0)/2-\theta,
	\end{equation}
	for some function $\theta$ with $|\theta|\leq \Vert g^{(3)}\Vert_{\infty}$.

	 We express the part of the log-likelihood involving $\eta^m$ explicitly as follows.
\begin{align}
\log p_{\eta^m}&(z)=dy\log\frac{e^{\eta^m(1,x)}}{1+e^{\eta^m(1,x)}}+d(1-y)\log\frac{1}{1+e^{\eta^m(1,x)}}  \notag\\
		&\qquad +(1-d)y\log\frac{e^{\eta^m(0,x)}}{1+e^{\eta^m(0,x)}}+(1-d)(1-y)\log\frac{1}{1+e^{\eta^m(0,x)}}  \notag\\
				=& d\left(y\eta^{m}(1,x)-\log(1+e^{\eta^{m}(1,x)})\right)+(1-d)\left(y\eta^{m}(0,x)-\log(1+e^{\eta^{m}(0,x)})\right). \label{LogLikelihood}
	\end{align}
Given equation \eqref{taylor:eq}, it remains to calculate the first three derivatives of the function $g$. Its first derivative is given by
	\begin{align*}
			g^\prime (u)&=-\frac{t}{\sqrt{n}}\gamma_0\rho^{\Psi(\eta_u^m)},
	\end{align*}
where $\gamma_0\rho^{\Psi(\eta_u^m)}(y,d,x)=y-\Psi(\eta_u^m(d,x))$.
	The second and third derivative of $g$ can be computed along the same lines:
	\begin{align*}
		g^{(2)}(u)=-\frac{t^2}{n}\gamma_0^2\Psi^\prime(\eta_u^m),~~
		g^{(3)}(u)=-\frac{t^3}{n^{3/2}}\gamma_0^3\Psi^{(2)}(\eta_u^m).
	\end{align*}
	In the above expression involving the Riesz representor $\gamma_0$, we have
	\begin{equation*}
		\gamma_0^2(d,x)=\frac{d}{\pi_0^2(x)}+\frac{1-d}{(1-\pi_0(x))^2}~~~\text{and}~~~\gamma_0^3(d,x)=\frac{d}{\pi_0^3(x)}-\frac{1-d}{(1-\pi_0(x))^3},
	\end{equation*}
 again because of $d(1-d)=0$. Evaluating at $u=0$, we have $\Psi(	\eta_u^m)=\Psi(\eta^m)=m_\eta$ and consequently, 
	\begin{align}\label{TaylorExp}
		g^\prime(0)=-\frac{t}{\sqrt{n}}\gamma_0\rho^{m_0}+\frac{t}{\sqrt{n}}\gamma_0(m_\eta-m_0),	
	\end{align}
and 
\begin{equation}\label{2ndTaylor}
		g^{(2)}(0)=-\frac{t^2}{n}\gamma_0^2 m_\eta(1-m_\eta) .
\end{equation}
For the remainder term, we have $\Vert g^{(3)}\Vert_{\infty}\lesssim n^{-3/2}$, given the uniform boundedness of $\Psi^{(2)}(\cdot)$.
\end{proof}

\begin{lemma}\label{lemma:negligible} Let Assumptions \ref{Ass:unconfounded} and \ref{Assump:Rate} be satisfied. Then, we have
\begin{align}\label{lemma:negligible:bound1}
\sqrt n\mathbb P_n[(\widehat\gamma-\gamma_0)\rho^{m_0}]=o_{P_0}(1).
\end{align}
\end{lemma}
\begin{proof}
Since $\widehat\gamma$ is based on an auxiliary sample, it is sufficient to consider deterministic functions $\gamma_n$ with the same rates of convergence as $\widehat\gamma$. We also write the corresponding propensity score as $\pi_n$, which is associated with $\gamma_n$.
Using the notation $\rho^{m_0}(Z_i)=Y_i-m_0(D_i,X_i)$, we evaluate for the conditional expectation that
\begin{align*}
&\mathbb E_0\left[\Big(\frac{1}{\sqrt n}\sum_{i=1}^n(\gamma_n-\gamma_0)(D_i,X_i)\rho^{m_0}(Z_i)\Big)^2\mid (D_1,X_1), \ldots, (D_n,X_n)\right]\\
& = \frac{1}{n}\sum_{i\neq i'}(\gamma_n-\gamma_0)(D_i,X_i)(\gamma_n-\gamma_0)(D_{i'},X_{i'}) \mathbb E_0\left[\rho^{m_0}(Z_i) \rho^{m_0}(Z_{i'})\mid (D_i,X_i), (D_{i'},X_{i'})\right]\\
& = \frac{1}{n}\sum_{i=1}^n(\gamma_n-\gamma_0)^2(D_i,X_i) Var_0(Y_i|X_i).
\end{align*}
We have $Var_0(Y_i|X_i)\leq 1$ since $Y_i\in\{0,1\}$ and thus we obtain for the unconditional squared expectation that
\begin{align*}
\mathbb E_0\left[\Big(\frac{1}{\sqrt n}\sum_{i=1}^n(\gamma_n-\gamma_0)(D_i,X_i)\rho^{m_0}(Z_i)\Big)^2\right]
\lesssim  \Vert \pi_n-\pi_0\Vert_{2, F_0}^2=o(1)
\end{align*}
by Assumption \ref{Assump:Rate}, which implies the desired result. 
\end{proof}
Each Gaussian process is associated with an intrinsic Hilbert space defined by its covariance kernel, see \cite{ghosal2017fundamentals}. This space is critical in analyzing the rate of contraction for its induced posterior. Consider a Hilbert space $\mathbb{H}$ with inner product $\langle\cdot,\cdot\rangle_{\mathbb{H}}$ and associated norm $\Vert\cdot\Vert_{\mathbb{H}}$. $\mathbb{H}$ is an reproducing kernel Hilbert space (RKHS) if there exists a symmetric, positive definite function $k:\mathcal{X}\times \mathcal{X}\mapsto \mathbb{R}$, called a kernel, that satisfies two properties: (i) $k(\cdot,\bm{x})\in\mathbb{H}$ for all $\bm{x}\in \mathcal{X}$ and; (ii) $f(\bm{x})=\langle f,k(\cdot,\bm{x})\rangle_{\mathbb{H}}$ for all $\bm{x}\in \mathcal{X}$ and $f\in \mathbb{H}$. It is well-known that every kernel defines a RKHS and every RKHS admits a unique reproducing kernel.

	Let $\mathbb{H}^{a_n}_1$ be the unit ball of the RKHS for the rescaled squared exponential process and let $\mathbb{B}_1^{s_m,p}$ be the unit ball of the H\"older class $\mathcal{C}^{s_m}([0,1]^p)$ in terms of the supremum norm $\Vert\cdot\Vert_{\infty}$. Denote $\Phi(\cdot)$ as the c.d.f. of a standard normal random variable with $\Phi^{-1}(\cdot)$ as its inverse.
	We introduce a class of functions $\mathcal{B}^m_n$ which is shown to contain the Gaussian process $W$ which sufficiently large probability, and is given by
\begin{equation}\label{SieveSet}
	\mathcal{B}^m_n:= \varepsilon_n \mathbb{B}_1^{s_m,p}+M_n\mathbb{H}_1^{a_n},
\end{equation}
where $a_n=n^{1/(2s_m+p)}(\log n)^{-(1+p)/(2s_m+p)}$, $\varepsilon_n=n^{-s_m/(2s_m+p)}(\log n)^{s_m(p+1)/(2s_m+p)}$, and $M_n=-2\Phi^{-1}(e^{-Cn\varepsilon_n^2})$. 
For notational simplicity, we suppress the dependence of the rescaled Gaussian process on the rescaling parameter $a_n$ in the following proofs. 

\begin{lemma}\label{lemma:RatesGP}
Under the conditions of Proposition \ref{prop:exponential}, the posterior distributions of the conditional mean functions contract at rate $\varepsilon_n$, i.e.,
\begin{equation*}
\Pi\left(\eta :\Vert m_{\eta}(d,\cdot)-m_0(d,\cdot)\Vert_{2, F_0}\geq M\varepsilon_n\mid Z^{(n)}\right)\to_{P_0}0
\end{equation*}
for $d\in\{0,1\}$ and every sufficiently large $M$, as $n\to\infty$.
\end{lemma}
\begin{proof}
	 By the assumed stochastic independence between the pair $Z^{(n)}$ and $\widehat{\gamma}$, we can proceed by studying the ordinary posterior distribution relative to the prior with $\widehat{\gamma}$ set equal to a deterministic function $\gamma_n$ and $(w,\lambda)$ following their prior. In other words, it is sufficient to consider the prior on $m$ given by $m(d,x) = \Psi\left(W^m(d,x) + \lambda\, \gamma_n(d,x)\right)$ where $W^m(d,\cdot)$ is the rescaled squared exponential process independent of $\lambda\sim N(0,\sigma_{n}^2)$ and $\gamma_n$ a sequence of functions $\|\gamma_n\|_\infty=O(1)$. It suffices to examine two conditional means $m_{\eta}(1,\cdot)$ and $m_{\eta}(0,\cdot)$ separately. We focus on the treatment arm with $d=1$, and leave $d$ off the notations in $W^m$ or $\eta^m$ as understood.
	 
	 We verify the following generic results in Theorem 2.1 of \cite{ghosal2000rates} to obtain the proper concentration rate for the posterior for the rescaled squared exponential process:
	  \begin{align}
	 &\Pi((w,\lambda): K\vee  V(p_{\eta_0^m},p_{w+\lambda \gamma_n})\leq \varepsilon_n^2)\geq c_1 \exp(-c_2n\varepsilon_n^2),\label{lemma:RatesGP:eq1}\\
	 &	\Pi(\mathcal{P}_n^c)\leq \exp(-c_3n\varepsilon_n^2),\label{lemma:RatesGP:eq2}\\
	 &\log  N(\varepsilon_n,\mathcal{P}_n,\Vert\cdot\Vert_{L^2(\nu)})\leq c_4 n\varepsilon_n^2,\label{lemma:RatesGP:eq3}
	 \end{align}
	 for positive constant terms $c_1,\cdots,c_4$ and for the set:
	 \begin{equation*}
	 	\mathcal{P}_n=\left\{p_{w+\lambda\gamma_n}:w\in \mathcal{B}^m_n,|\lambda|\leq M\sigma_n\sqrt{n}\varepsilon_n \right\}.
	 \end{equation*}

	Proof of \eqref{lemma:RatesGP:eq1}. The inequality \eqref{bound:K:V} in Lemma \ref{lemma:distance} yields
	\begin{equation*}
	\left\{(w,\lambda):\Vert w-\eta_0^m\Vert_{\infty}\leq c\varepsilon_n,|\lambda|\leq c\varepsilon_n \right\}\subset \left\{(w,\lambda): K\vee V(p_{\eta_0^m},p_{w+\lambda \gamma_n})\leq \varepsilon_n^2\right\}.
	\end{equation*}
	Given that we have independent priors of $W^m$ and $\lambda$, the prior probability of the set on the left of the above display can be lower bounded by $\Pi(\Vert W^m-\eta_0^m\Vert_{\infty}\leq c\varepsilon_n)\Pi(|\lambda|\leq c\varepsilon_n)$.
	By Proposition 11.19 of \cite{ghosal2017fundamentals} regarding the small exponent function $\phi^{a_n}_0$ and together with the upper bound \eqref{bound:concentration}, we infer
	\begin{align*}
	\Pi(\Vert W^m-\eta_0^m\Vert_{\infty}\leq c\varepsilon_n)\geq \exp\left(-\phi^{a_n}_0(\varepsilon_n/2)\right)\geq \exp\left(-cn\varepsilon^2_n\right),
	\end{align*}
for some positive constant $c$.
	The second term is lower bounded by $C\varepsilon_n/\sigma_n$, which is of order $O(\varepsilon_n)$ for $\sigma_{n}=O(1)$. Therefore, we have ensured that the prior assigns enough mass around a Kullback-Leibler neighborhood of the truth.
	
Proof of \eqref{lemma:RatesGP:eq2}. Referring to the sieve space for the Gaussian process, we apply Borell's inequality from Proposition 11.17 of \cite{ghosal2017fundamentals}:
		\begin{equation*}
			\Pi(W^m\not\in \mathcal{B}^m_n)\leq 1-\Phi(\iota_n+M_n),
		\end{equation*}
		where $\Phi(\cdot)$ is the c.d.f. of a standard normal random variable and the sequence $\iota_n$ is given by $\Phi(\iota_n)=\Pi(W\in \varepsilon_n \mathbb{B}^{s_m,p}_1)=e^{-\phi_0^{a_n}(\varepsilon_n)}$. Since our choice of $\varepsilon_n$ leads to $\phi_0^{a_n}(\varepsilon_n)\leq n\varepsilon_n^2$, we have $\iota_n\geq -M_n/2$ if $M_n=-2\Phi^{-1}(e^{-Cn\varepsilon_n^2})$ for some $C>1$. In this case, $\Pi(\mathcal{B}^{mc}_n)\leq 1-\Phi(M_n/2)\leq e^{-Cn\varepsilon_n^2}$. Next, we apply the univariate Gaussian tail inequality for $\lambda$:
	\begin{equation*}
		\Pi\left(|\lambda|\geq u_n\sigma_{n}\sqrt{n} \right)\leq 2e^{-u_n^2n\sigma_{n}^2/2},
	\end{equation*}
	which is bounded above by $e^{-Cn\varepsilon_n^2}$ for $u_n\to 0$ sufficiently slowly, given our assumption $\varepsilon_n=o(\sigma_{n})$. Hence, by the union bound, we have $\Pi(\mathcal{P}_n^c)\lesssim e^{-Cn\varepsilon_n^2}$.
	
	Proof of \eqref{lemma:RatesGP:eq3}. To bound the entropy number of the functional class $\mathcal{P}_n$, consider the inequality
	\begin{equation*}
	\Vert p_{w+\lambda\gamma}-p_{\bar{w}+\bar{\lambda}\gamma_n}\Vert_{L^2(\nu)}\lesssim \Vert w-\bar{w}\Vert_{2, F_0}+|\lambda-\bar{\lambda}|\Vert \gamma_n\Vert_{\infty},
	\end{equation*} 
	where the dominating measure $\nu$ is (\ref{DominateMeasure}). Thus, we have
	\begin{equation}
	N(\varepsilon_n,\mathcal{P}_n,\Vert\cdot\Vert_{L^2(\nu)})\leq  N(\varepsilon_n/2,\mathcal{B}^m_n,\Vert\cdot\Vert_{\infty}) \times N(C\varepsilon_n,[0,2M\sigma_n\sqrt{n}\varepsilon_n],|\cdot|)\lesssim n\varepsilon_n^2.
	\end{equation}
	Note that the logarithm of the second term grows at the rate of $O(\log n)$, and it is the first term that dominates. Because $\Psi$ is monotone and Lipschitz, a set of $\epsilon$-brackets in $L^2(F_0)$ for $\mathcal{B}_n^m$ translates into a set of $\varepsilon$-brackets in $L^2(\nu)$ for $\mathcal{P}_n$. Thus, Lemma \ref{lemma:Entropy} gives us $\log N(3\varepsilon_n,\mathcal{B}^m_n,\Vert\cdot\Vert)\lesssim n\varepsilon_n^2$.  
	
	 By Lemma 15 of \cite{ray2020causal}, this delivers the posterior contraction rate for $m_{\eta}(1,\cdot)$ in terms of the $L^2(F_0\pi_0)$-norm, which is equivalent to the $L^2(F_0)$-norm weighted by the  propensity score $\pi_0$. Analogous arguments lead to the desired result for the conditional mean  $m_{\eta}(0,\cdot)$ for the control group.
\end{proof}

Let $M_{ni}=e_i/\sum_{i=1}^ne_i$, where $e_i$'s are independently and identically drawn from the exponential distribution $\textup{Exp}(1)$. We also denote $X^{(n)}=(X_i)_{i=1}^n$. We adopt the following notations: $\mathbb{F}^*_n\bar{m}_{\eta}= \sum_{i=1}^n M_{ni}\bar{m}_{\eta}(X_i)$, $\mathbb{F}_n\bar{m}_{\eta}=n^{-1}\sum_{i=1}^n \bar{m}_{\eta}(X_i)$ and $F_0\bar{m}_{\eta}= \int \bar{m}_{\eta}(x)dF_0(x)$. Let $X^{(n)}=(X_i)_{i=1}^n$.
\begin{lemma}\label{lemma:DP}
Let the functional class $\{\bar{m}_{\eta}:\eta\in\mathcal{H}_n \}$ be a $P_0$-Glivenko-Cantelli class. Then for every $t$ in a sufficiently small neighborhood of 0 we have
	\begin{equation*}
		\sup _{\bar{m}_{\eta}:\eta\in\mathcal{H}_n }\left|\mathbb{E}\left[e^{t \sqrt{n}\left((\mathbb{F}^*_n-\mathbb{F}_n) \bar{m}_{\eta}\right)} \mid X^{(n)}\right]-e^{t^2 F_0\left(\bar{m}_{\eta}-F_0 \bar{m}_{\eta}\right)^2 / 2}\right| \rightarrow_{P_0} 0 .
\end{equation*}
\end{lemma}
\begin{proof}
	We verify the conditions from Lemma 1 in \cite{ray2020causal}. First, the Bayesian bootstrap law $\mathbb{F}^*_n$ is the same as the posterior law for $F$, when its prior is a Dirichlet process with its base measure taken to be zero. Second, the assumed $P_0$-Glivenko-Cantelli class entails 
	\begin{equation*}
		\sup_{\eta\in\mathcal{H}_n}\left|(\mathbb{F}_n-F_0)\bar{m}_{\eta}\right|=o_{P_0}(1).
	\end{equation*}
Last, the required moment condition on the envelope function for the class involving $\bar{m}_{\eta}$ is automatically satisfied because of $\Vert \bar{m}_{\eta}\Vert_{\infty}\leq 1$.
\end{proof}

The following lemma is in the same spirit of Lemma 9 in \cite{ray2020causal} with one important difference. That is, we do not restrict the range of the function $\varphi$ to be $[0,1]$. As we apply this lemma by taking $\varphi=\gamma_n-\gamma_0$, it can take on negative values. We apply the more general contraction principle from Theorem 4.12 of \cite{ledoux1991banach} instead of Proposition A.1.10 of \cite{van1996empirical}. This allows us to relax the positive range restriction in \cite{ray2020causal}. 
\begin{lemma}\label{lemma:product}
	Consider a set $\mathcal{H}$ of measurable functions $h:\mathcal{Z}\mapsto \mathbb{R}$ and a bounded measurable function $\varphi$. We have
	\begin{equation*}
	\mathbb{E}_0\sup_{h\in\mathcal{H}}|\mathbb{G}_n(\varphi h)|\leq 4\Vert\varphi\Vert_{\infty}\mathbb{E}_0\sup_{h\in\mathcal{H}}|\mathbb{G}_n(h)|+\sqrt{P_0\varphi^2}\sup_{h\in\mathcal{H}}|P_0h|.
	\end{equation*}
\end{lemma}
\begin{proof}
We start with $\mathbb{G}_n(\varphi h)=\mathbb{G}_n(\varphi (h-P_0h))+\mathbb{G}_n(\varphi)P_0h$.
The expectation of $\mathbb{G}_n(\varphi)P_0h$ is bounded by the second term on the right hand side of the inequality in the stated lemma as follows:
\begin{equation}\label{elementary}
	\mathbb{E}_0\sup_{h\in\mathcal{H}}\left|\mathbb{G}_n(\varphi) P_0h\right|\leq \sup_{h\in\mathcal{H}}|P_0h|\mathbb{E}_0\left|\frac{1}{\sqrt{n}}\sum_{i=1}^{n}(\varphi(Z_i)-\mathbb{E}_0[\varphi(Z)]) \right|\leq \sqrt{P_0\varphi^2}\sup_{h\in\mathcal{H}}|P_0h|,
\end{equation}
where the last inequality follows from the elementary bound $\mathbb{E}_0|\xi|\leq \sqrt{\mathbb{E}_0[\xi]^2}$ for any random variable $\xi$ and the fact that $\mathbb{E}_0\left[\frac{1}{\sqrt{n}}\sum_{i=1}^{n}(\varphi(Z_i)-P_0\varphi(Z_i))\right]^2\leq P_0\varphi^2$.

We now consider $\mathbb{G}_n(\varphi (h-P_0h))$ where we may assume that the function $h$ satisfies $P_0h=0$. Let $\epsilon_1,\ldots,\epsilon_n$ be i.i.d. Rademacher random variables independent of observations $Z^{(n)}$. By the symmetrization inequality in Lemma 2.3.6 of \cite{van1996empirical},
\begin{equation}\label{symmetry}
	\mathbb{E}_0\sup_{h\in\mathcal{H}}\left|\sum_{i=1}^n(\varphi(Z_i)h(Z_i)-P_0[\varphi h]) \right|\leq 2\Vert \varphi\Vert_{\infty} \mathbb{E}_0\sup_{h\in\mathcal{H}}\left|\sum_{i=1}^n\epsilon_i\frac{\varphi(Z_i)}{\Vert\varphi\Vert_{\infty}}h(Z_i)\right|.
\end{equation} 
Because $-1\leq \varphi(Z_{i})/\|\varphi\|_{\infty}\leq 1$ for all $i=1,\dots, n$, the map $h\mapsto \frac{\varphi}{\Vert\varphi\Vert_{\infty}}\times h$ forms a contraction mapping. Hence, we apply the contraction principle as in Theorem 4.12 on page 112 of \cite{ledoux1991banach}:
\begin{equation}\label{symmetry:2}
	\mathbb{E}_0\sup_{h\in\mathcal H}\left|\sum_{i=1}^n\epsilon_i\frac{\varphi(Z_i)}{\Vert\varphi\Vert_{\infty}}h(Z_i)\right|\leq \mathbb{E}_0\sup_{h\in\mathcal{H}}\left|\sum_{i=1}^n\epsilon_ih(Z_i)\right|.
\end{equation}
Another application by the symmetrization inequality from Lemma 2.3.6 of \cite{van1996empirical} that decouples the Rademacher variables leads to 
\begin{equation*}
\mathbb{E}_0\sup_{h\in\mathcal{H}}\left|\sum_{i=1}^n\epsilon_ih(Z_i)\right|\leq 2\mathbb{E}_0\sup_{h\in\mathcal{H}}\left|\sum_{i=1}^nh(Z_i)\right|.
\end{equation*}
The application of the above inequality requires the process $\{h:h\in\mathcal{H}\}$ under consideration to be centered as $P_0h=0$. Combined together with Inequalities \eqref{elementary}, \eqref{symmetry}, and \eqref{symmetry:2}, we have the desired result as stated in the lemma.
\end{proof}

The next lemma upper bounds the $L^2$ distance and Kullback-Leibler divergence of the probability density functions by the $L^2$ distance of the reparametrized function $\eta^m$, cf. Lemma 2.8 of \cite{ghosal2017fundamentals} or Lemma 15 of \cite{ray2020causal}. 
We introduce some simplifying notations by writing
\begin{equation*}
	m^1(\cdot)=m(1,\cdot)~~~\text{and}~~~ m^0(\cdot)=m(0,\cdot).
\end{equation*}
\begin{lemma}\label{lemma:distance}
For any measurable functions $v^m,w^m:[0,1]^p\mapsto \mathbb{R}$, we have 
\begin{align}
\Vert p_{v^m}-p_{w^m}\Vert_{L^2(\nu)}&\leq\Vert \Psi(v^{m^1})-\Psi(w^{m^1})\Vert_{L^2(F_0\pi_0)}\vee\Vert \Psi(v^{m^0})-\Psi(w^{m^0})\Vert_{L^2(F_0(1-\pi_0))} \notag\\
&\leq \Vert v^{m^1}-w^{m^1}\Vert_{2, F_0}\vee \Vert v^{m^0}-w^{m^0}\Vert_{2, F_0}.
\end{align}
In addition, it holds that
\begin{align}\label{bound:K:V}
	K(p_{v^m},p_{w^m})\vee V(p_{v^m},p_{w^m})\leq \Vert v^{m^1}-w^{m^1}\Vert^2_{2, F_0}\vee \Vert v^{m^0}-w^{m^0}\Vert^2_{2, F_0}.
\end{align}
\end{lemma}

The small ball exponent function for the associated Gaussian process prior is
\begin{align*}
	\phi_0(\varepsilon):= -\log \Pi(\Vert W\Vert_{\infty}<\varepsilon);
\end{align*}
see equation (11.10) in \cite{ghosal2017fundamentals}. In the above display, $\Vert\cdot\Vert_{\infty}$ is the uniform norm of $\mathcal{C}([0,1]^p)$, the Banach space in which the Gaussian process sits. $\mathbb{H}$ is the reproducing kernel Hilbert space (RKHS) of the process with its RKHS norm $\Vert\cdot\Vert_{\mathbb{H}}$. To abuse the notation a bit, we denote the small ball exponent of the rescaled process $W(at)$ by $\phi^a_0(\varepsilon)$. Lemma 11.55 in \cite{ghosal2017fundamentals} gives this bound for the (rescaled) squared exponential process:
	\begin{equation*} \label{lemma:SBexponent}
	\phi^a_0(\varepsilon)\lesssim a^p(\log(a/\varepsilon))^{1+p}.
	\end{equation*}

\begin{lemma}\label{lemma:Entropy}
Assume that $\varepsilon_n=n^{-s_m/(2s_m+p)}(\log n)^{s_m(1+p)/(2s_m+p)}$ and $M_n=-2\Phi^{-1}(e^{-Cn\varepsilon_n^2})$ for a positive constant $C>1$. Also, let $a_n\asymp n^{1/(2s_m+p)}(\log n)^{-(1+p)/(2s_m+p)}$. Then, for the sieve space $\mathcal{B}^m_n= \varepsilon_n \mathbb{B}_1^{s_m,p}+M_n\mathbb{H}_1^{a_n}$, we have
\begin{equation}
	\log N(3\varepsilon_n,\mathcal{B}^m_n,\Vert\cdot\Vert_{\infty})\lesssim n\varepsilon_n^2.
\end{equation}

\end{lemma}
\begin{proof}
The argument is similar as in Lemma 11.20 of \cite{ghosal2017fundamentals}. We follow the generic argument as in the proof of Theorem 11.20 of \cite{ghosal2017fundamentals} to bound the complexity number $	\log N(3\varepsilon_n,\mathcal{B}^m_n,\Vert\cdot\Vert_{\infty})$ given the bound for the small ball exponent in \eqref{lemma:SBexponent}.
We provide the proof for completeness. For some integer $N\geq 1$, let $h_1,\ldots,h_{N}\in M_n\mathbb{H}^{a_n}_1$ be $2\varepsilon_n$-separated functions in terms of the Banach space norm. Then, the $\varepsilon_n$-balls $h_1+\varepsilon_n\mathbb{B}^{s_m,p}_1,\ldots,h_N+\varepsilon_n\mathbb{B}^{s_m,p}_1$ are disjoint. Therefore, we have
\begin{align*}
	1\geq \sum_{j=1}^N \Pi(W\in h_j+\varepsilon_n\mathbb{B}^{s_m,p}_1)
	\geq\sum_{j=1}^N e^{-\Vert h_j\Vert_{\mathbb{H}}^2/2}\Pi(W\in \varepsilon_n\mathbb{B}^{s_m,p}_1)\geq N e^{-M_n^2/2}e^{-\phi_{0}^{a_n}(\varepsilon_n)},
\end{align*}
where the second inequality follows from Lemma 11.18 of \cite{ghosal2017fundamentals} and the last inequality makes use of the fact that $h_1,\ldots,h_{N}\in M_n\mathbb{H}_1$, as well as the definition of the small ball exponent function. 

For a maximal $2\varepsilon_n$-separated set $h_1,\ldots,h_N$, the balls around $h_1,\ldots,h_N$ of radius $2\varepsilon_n$ cover the set $M_n\mathbb{H}^{a_n}_1$. Thus, we have $\log N(2\varepsilon_n,M_n\mathbb{H}^{a_n}_1,\Vert\cdot\Vert_{\infty})\leq\log N\leq M_n^2/2+\phi^{a_n}_0(\varepsilon_n)$. Referring to the inequality (iii) of Lemma K.6 of \cite{ghosal2017fundamentals} for the quantile function of a standard normal distribution, we have $M_n^2\lesssim n\varepsilon_n^2$ by the choice of $M_n$ stated in the lemma. It is straightforward yet tedious to verify that
\begin{align}\label{bound:concentration}
	\phi^{a_n}_0(\varepsilon_n)\lesssim n\varepsilon_n^2,
\end{align}
for the specified $a_n$ and $\varepsilon_n$. Since any point of $\mathcal{B}^m_n$ is within $\varepsilon_n$ of an element of $M_n\mathbb{H}^{a_n}_1$, this also serves as a bound on $\log N(3\varepsilon_n,\mathcal{B}^m_n,\Vert\cdot\Vert_{\infty})$.
\end{proof}

A key step in showing the validity of the debiasing step is the following:
	\begin{align*}
	\mathbb{P}_n[\widehat{\bar{m}} + \widehat{\gamma}\rho^{\widehat m} - \bar{m}_0]=\mathbb{P}_n[\gamma_0\rho^{m_0}]+o_{P_0}(n^{-1/2}),
\end{align*}
which is equivalent to the following lemma.
\begin{lemma}\label{lemma:FreqDR}
Under Assumption \ref{Assump:Rate} for the pilot estimators, the following result holds:
	\begin{align*}
	\mathbb{P}_n[\widehat{\gamma}\rho^{\widehat{m}}+\widehat{\bar{m}}]=\mathbb{P}_n[\gamma_0\rho^{m_0}+\bar{m}_0]+o_{P_0}(n^{-1/2}).
\end{align*}
	\end{lemma}
\begin{proof}
	We start with the following identity:
		\begin{align*}
		\mathbb{P}_n[\widehat{\gamma}\rho^{\widehat{m}}+\widehat{\bar{m}}]=\mathbb{P}_n[\gamma_0\rho^{m_0}+\bar{m}_0]+R_{n1}+R_{n2}.
	\end{align*}
	where 
{\small 
\begin{align*}
	R_{n1}&=\frac{1}{n}\sum_{D_i}(Y_i-\widehat{m}(1,X_i))\left(\frac{1}{\widehat{\pi}(X_i)}-\frac{1}{\pi_0(X_i)}\right)+\frac{1}{n}\sum_{1-D_i}(Y_i-\widehat{m}(0,X_i))\left(\frac{1}{1-\widehat{\pi}(X_i)}-\frac{1}{1-\pi_0(X_i)}\right),\\
	R_{n2}&=\frac{1}{n}\sum_i(\widehat{m}(1,X_i)-m_0(1,X_i))\left(1-\frac{D_i}{{\pi_0}(X_i)}\right)+\frac{1}{n}\sum_i(\widehat{m}(0,X_i)-m_0(0,X_i))\frac{D_i-\pi_0(X_i)}{1-{\pi_0}(X_i)}.
\end{align*}}%
Referring to the first term $R_{n1}$, we have
{\small \begin{align*}
	R_{n1}	&=\frac{1}{n}\sum_{D_i}(m_0(1,X_i)-\widehat{m}(1,X_i))\left(\frac{1}{\widehat{\pi}(X_i)}-\frac{1}{\pi_0(X_i)}\right)
	+\frac{1}{n}\sum_{D_i}(Y_i-m_0(1,X_i))\left(\frac{1}{\widehat{\pi}(X_i)}-\frac{1}{\pi_0(X_i)}\right)\\
	&-\frac{1}{n}\sum_{1-D_i}(m_0(0,X_i)-\widehat{m}(0,X_i))\left(\frac{1}{1-\widehat{\pi}(X_i)}-\frac{1}{1-\pi_0(X_i)}\right)\\
	&-\frac{1}{n}\sum_{1-D_i}(Y_i-m_0(0,X_i))\left(\frac{1}{1-\widehat{\pi}(X_i)}-\frac{1}{1-\pi_0(X_i)}\right).
\end{align*}}%
The negligibility of the first and third terms in $R_{n1}$ follows from the Cauchy-Schwarz inequality and the rate conditions imposed in Assumption \ref{Assump:Rate}. The second and fourth terms can be combined together so that the negligibility can be shown as in Lemma \ref{lemma:negligible}.

	Consider $R_{n2}$. To bound its first summand, we condition on $(X_1, \ldots, X_n)$, as well as the pilot estimators $\widehat{m}$ and $\widehat{\pi}$, which are computed over the external sample. We use the fact that  $(D_i-\pi_0(X_i))$ has a conditional zero mean. Specifically, this leads to 
\begin{align*}
	&\mathbb E_0\left[\Big(\frac{1}{\sqrt n}\sum_{i=1}^n\frac{D_i-\pi_0(X_i)}{\widehat{\pi}(X_i)}(\widehat{m}(1,X_i)-m_0(1,X_i))\Big)^2\,\Big|\, X_1, \ldots, X_n,\widehat{m},\widehat{\pi}\right]\\
	& = \frac{1}{n}\sum_{i=1}^n\big(\widehat{m}(1,X_i)-m_0(1,X_i)\big)^2 \frac{Var_0(D_i|X_i)}{\widehat{\pi}^2(X_i)}
\end{align*}
using that $Var_0(D_i|X_i)=\pi_0(X_i)(1-\pi_0(X_i))$. By the overlapping condition as imposed in Assumption \ref{Ass:unconfounded}, i.e., $\bar\pi <\pi_0(X_i)$ for all $1\leq i\leq n$ and the uniform convergence of $\widehat{\pi}$ to $\pi_0$, we obtain
\begin{align*}
	\mathbb E_0\left[\Big(\frac{1}{\sqrt n}\sum_{i=1}^n\frac{D_i-\pi_0(X_i)}{\widehat{\pi}(X_i)}(\widehat{m}(1,X_i)-m_0(1,X_i))\Big)^2\,\Big|\,\widehat{m},\widehat{\pi}\right]
	\lesssim  \Vert \widehat{m}(1,\cdot)-m_0(1,\cdot)\Vert_{2, F_0}^2=o_{P_0}(1),
\end{align*}
where the last equation is due to the convergence rate for the pilot estimator $\widehat{m}$ in Assumption \ref{Assump:Donsker}. The negligibility of the second term in $R_{n2}$ is proved in a similar fashion.
	\end{proof}
The following lemma shows the stochastic equicontinuity when the true conditional mean function belongs to a H\"older space, which is $P_0$-Donsker, i.e., $s_m>p/2$. The main complication is that the sieve space related to the Gaussian process prior is not a fixed $P_0$-Donsker class, as it changes with sample size $n$ and the envelope function is also slowly diverging, cf. the comments in the third paragraph on Page 2007 of \cite{ray2020causal}. More specifically, for the rescaled squared exponential process priors, we rely on the metric entropy bounds in \cite{van2009adaptive}. With this important modification, the proof is along similar lines with the proof of Lemma 7 of \cite{ray2020causal} for the Riemann-Lioville process; also, see Lemma 5 of \cite{ray2020causal}.

We consider
\begin{equation}\label{HnSet}
		\mathcal{H}_n^m:= \left\{w_d+\lambda\gamma_n:(w_d,\lambda)\in \mathcal{W}_n\right\},
\end{equation}
where 
{\small \begin{align*}
		\mathcal{W}_n:= \left\{(w_d,\lambda):w_d\in \mathcal{B}^m_n,|\lambda|\leq M\sigma_n\sqrt{n}\varepsilon_n \right\}\cap \left\{(w_d,\lambda):\Vert \Psi(w_d(\cdot)+\lambda\gamma_n)-m_0(d,\cdot)\Vert_{2, F_0}\leq \varepsilon_n \right\},
\end{align*}}%
where the sieve space $\mathcal{B}_n^m$ in the first restriction for the Gaussian process $W_d$ is defined in the equation (\ref{SieveSet}) with $d\in\{0,1\}$, and $\varepsilon_n=(n/\log n)^{-s_m/(2s_m+p)}$.
\begin{lemma}\label{lemma:DClass}
Recall that the sieve space related to the Gaussian process is $\mathcal{B}^m_n= \varepsilon_n \mathbb{B}_1^{s_m,p}+M_n\mathbb{H}_1^{a_n}$. For $s_m>p/2$, we have
$	\E_0	\sup_{\eta\in\mathcal{H}^m_n}\mathbb{G}_n\left[m_\eta-m_0\right]=o(1)$.
\end{lemma}
\begin{proof}
	Because the link function $\Psi(\cdot)$ is monotone and Lipschitz continuous, separate sets of brackets for the two constituents of the set $\varepsilon_n \mathbb{B}_1^{s_m,p}+M_n\mathbb{H}_1^{a_n}$, as well as the bracket for $\{\lambda:|\lambda|\leq M \sigma_{n}\sqrt{n}\varepsilon_n\}$ can be combined into brackets for the sum space.
	\begin{small}
	\begin{equation*}
		\log N_{[]}(\varepsilon,\mathcal{H}_n^m,L^2(P_0))\leq \log N(\varepsilon,\varepsilon_n \mathbb{B}_1^{s_m,p},\Vert\cdot\Vert_{\infty})+\log N(\varepsilon,M_n\mathbb{H}_1^{a_n},\Vert\cdot\Vert_{\infty})+\log N(c\varepsilon,[0,2M\sigma_n\sqrt{n}\varepsilon_n],|\cdot|).
	\end{equation*}
\end{small}
The last term is of strictly smaller order than the second one. The bound for the first component attached to the H\"older space can be found in Proposition C.5 of \cite{ghosal2017fundamentals}:
\begin{equation*}
\log N(\varepsilon,\varepsilon_n \mathbb{B}_1^{s_m,p},\Vert\cdot\Vert_{\infty})\lesssim \left(\frac{\varepsilon_n}{\varepsilon}\right)^{s_m/p},
\end{equation*}
which is bounded if we take $\varepsilon=\varepsilon_n$. The entropy bound for the second component is given in Lemma \ref{lemma:Entropy}, which states that $\log N(\varepsilon,M_n\mathbb{H}_1^{a_n},\Vert\cdot\Vert_{\infty})\lesssim n\varepsilon_n^2\lesssim \varepsilon_n^{-2\upsilon}$,
with $\upsilon = p/(2s_m)$ modulo some $\log n$ term on the right hand of the bound. In this case, the empirical process bound of \citep[p.2644]{han2021set} yields
\begin{equation*}
\E_0\sup_{\eta\in\mathcal{H}^m_n}\left|  \mathbb{G}_n[m_\eta-m_0] \right|\lesssim L_n n^{(\upsilon-1)/(2\upsilon)}=O(L_nn^{1/2-s_m/p})=o(1),
\end{equation*}
where $L_n$ represents a term that diverges at certain polynomial order of $\log n$.
\end{proof}

\section{Proofs of Section \ref{sec:extension}}\label{appendix:proof_exponent}

\begin{proof}[Proof of Lemma \ref{lemma:lfd_exp}]
 For the submodel $t\rightarrow\eta_{t}$ defined in \ref{lemma:lfd_exp}, we evaluate
	\begin{align*}
	\log p_{\eta_t}(z) &= d\log\Psi(\eta^{\pi} + t\mathfrak{p})(x)+ (1-d)\log (1-\Psi(\eta^{\pi} + t\mathfrak{p}))(x) \\
	&+ \log c(y) +ay(\eta^m+t\mathfrak{m})(d,x) -A(q^{-1}(\eta^m+t\mathfrak{m}))(d,x)\\
	& + t\mathfrak{f}(x) - \log\mathbb{E}[e^{t\mathfrak{f}(X)}] + \log f(x).
	\end{align*}	
Taking derivative with respect to $t$ and evaluating at $t=0$ gives the score operator:
\begin{equation}\label{score_exp}
B_{\eta}(\mathfrak{p},\mathfrak{m},\mathfrak{f})(Z)= B_{\eta}^{\pi}\mathfrak{p}(Z) + B_{\eta}^{m}\mathfrak{m}(Z) + B_{\eta}^{f}\mathfrak{f}(Z),
\end{equation}
where $B_{\eta}^{\pi}\mathfrak{p}(Z) =   (D-\pi_\eta(X))\mathfrak{p}(X)$, $B_{\eta}^{f}\mathfrak{f}(Z)= \mathfrak{f}(X)$, and 
\begin{align*}
B_{\eta}^{m}\mathfrak{m}(Z) =& \left[aY-\frac{A^\prime(m_{\eta}(D, X))}{q^\prime(m_{\eta}(D, X))}\right]\mathfrak{m}(D, X),\notag\\
=& a\left(Y-m_{\eta}(D, X)\right)\mathfrak{m}(D,X).
\end{align*}
In the last equation, we made use of the relation (explicitly given here for continuous outcomes):
\begin{eqnarray*}
A^\prime(m_{\eta}(d,x)) 
& =& q^\prime(m_{\eta}(d,x))\int ayc(y)\exp\left[q(m_{\eta}(d,x)) ay-A(m_{\eta}(d,x))\right]\mathrm{d}y\\
&=&  q^\prime(m_{\eta}(d,x))\mathbb{E}_{\eta}\left[aY|D=d,X=x\right],
\end{eqnarray*}
which follows from the exponential family assumption. 

In this case, there is a one-to-one correspondence between the conditional density function and the conditional mean function of the outcome given covariates. One can easily verify the differentiability of the ATE parameter in the sense of \cite{van1998asymptotic} and show that the efficient influence function remains the same as in \cite{hahn1998role} and \cite{ray2020causal}. Given the particular form of the efficient influence function $\widetilde{\tau}_{\eta}$ in \eqref{eif_ate},
the function $\xi_{\eta}= (\xi_{\eta}^{\pi},\xi_{\eta}^{m},\xi_{\eta}^{f})$ defined in \eqref{lfd} satisfies
 $B_{\eta}\xi_{\eta}=\widetilde{\tau}_{\eta}$,
and hence, $\xi_{\eta}$ defines the least favorable direction.
\end{proof}

We emphasize that the least favorable direction calculation in the multinomial outcome case is not a trivial extension of \cite{hahn1998role} or \cite{ray2020causal}, because there are $J$ nonparametric components involved in the conditional probabilities of the multinomial outcomes given covariates, and we need to consider the perturbation of all $J$ components together.
 \begin{proof}[Proof of Lemma \ref{lemma:lfd:multinomial}]
Consider the log transformation of the joint density of $Z=(Y,D,X^\top)^\top$ given by
\begin{equation*}
\log p_\eta(z)  =d \log (\pi_\eta(x))+(1-d) \log (1-\pi_\eta(x))  +\sum_{j=0}^J \mathbbm1_{\left\{y=j\right\}} \log \left(m_{j,\eta}\left(d, x\right)\right)+\log f(x)
\end{equation*}
where $\mathbbm 1_{\{\cdot\}}$ denotes the indicator function.
Following the proof of Lemma 3.1, it is sufficient to consider the perturbations for $j=1,\ldots, J$:
\begin{align*}
\Psi_j(\eta^{m_1}+t\mathfrak{m}_1,\cdots,\eta^{m_J}+t\mathfrak{m}_J)(d,x)=\frac{\exp((\eta^{m_j}+t\mathfrak{m}_j)(d,x))}{1+\sum_{l=1}^J \exp((\eta^{m_l}+t\mathfrak{m}_l)(d,x))}
\end{align*}
or
\begin{align*}
\log\Psi_j(\eta^{m_1}+t\mathfrak{m}_1,\cdots,\eta^{m_J}+t\mathfrak{m}_J)(d,x) = (\eta^{m_j}+t\mathfrak{m}_j)(d,x) -\log\Big(1+\sum_{l=1}^J \exp((\eta^{m_l}+t\mathfrak{m}_l)(d,x))\Big).
\end{align*}
Taking derivatives
\begin{align*}
\frac{\partial\log \Psi_j(\eta^{m_1}+t\mathfrak{m}_1,\cdots,\eta^{m_J}+t\mathfrak{m}_J)(d,x)}{\partial t}\Bigg|_{t=0} &= \mathfrak{m}_j(d,x) -
\frac{\sum_{l=1}^J \exp(\eta^{m_l}(d,x)) \mathfrak{m}_l(d,x)  }{1 + \sum_{l=1}^J \exp(\eta^{m_l}(d,x))}\\
&= \mathfrak{m}_j(d,x) -
\sum_{l=1}^J m_{\eta,l}(d,x) \mathfrak{m}_l(d,x) 
\end{align*}
by the definition of $m_{\eta, l}$. Likewise, we also obtain
\begin{align*}
	\frac{\partial \log \Psi_0(\eta^{m_1}+t\mathfrak{m}_1,\cdots,\eta^{m_J}+t\mathfrak{m}_J)(d,x)}{\partial t}\Bigg|_{t=0} = -
	\sum_{l=1}^J m_{\eta,l}(d,x) \mathfrak{m}_l(d,x). 
\end{align*}
We need to verify the differentiability of the ATE parameter in the sense of \cite{van1998asymptotic}. Due to its technical feature, we leave this to the end of the proof. From there, we can see that the score operator of the vector of conditional means $(m_1,\ldots, m_J)$ is as follows:
\begin{align*}
&B_\eta^{m}(\mathfrak{m}_1,\ldots,\mathfrak{m}_J) (z)=\sum_{j=0}^J\mathbbm1_{\{y=j \}}	\frac{d\log \Psi_j(\eta^{m_1}+t\mathfrak{m}_1,\cdots,\eta^{m_J}+t\mathfrak{m}_J)(d,x)}{dt}\Bigg|_{t=0} \\
	&= \sum_{j=1}^J\mathbbm1_{\{y=j \}}\left(\mathfrak{m}_j(d,x) -
	\sum_{l=1}^J m_{\eta,l}(d,x) \mathfrak{m}_l(d,x) \right)+\mathbbm1_{\{y=0 \}}\left(- \sum_{l=1}^J m_{\eta,l}(d,x) \mathfrak{m}_l(d,x)\right).
\end{align*}
Given that $\mathbbm1_{\{y=0 \}}=1-\sum_{j=1}^J\mathbbm1_{\{y=j \}}$, the previous equation simplifies to
\begin{equation*}
B_\eta^{m}(\mathfrak{m}_1,\ldots,\mathfrak{m}_J) (z)
=\sum_{j=1}^J \left(\mathbbm1_{\left\{y=j\right\}}-m_{\eta,j}(d,x)\right) \mathfrak{m}_j(d,x).
\end{equation*}
Note that the conditional mean of $B_\eta^{m}(\mathfrak{m}_1,\ldots,\mathfrak{m}_J) (z)$ is zero for any $\mathfrak{m}_j(d,x)$, which aligns with the requirement of the score operator.

From our verification of differentiability, we infer that the influence function is of the generic form given in \cite{hahn1998role} and \cite{ray2020causal}. Also, it is contained in the closed linear span of the set of all score functions. Now, if we choose $\mathfrak{m}_j=j \gamma_\eta$, $1\leq j\leq J$, we obtain
\begin{align*}
B_\eta^{m}(\gamma_\eta, 2\gamma_\eta,\ldots, J\gamma_\eta) (z)
=\Bigg(\underbrace{\sum_{j=1}^J \mathbbm1_{\left\{y=j\right\}}j}_{=y}  -
\underbrace{\sum_{j=1}^J j\, m_{\eta,j}(d,x)}_{=m_{\eta}(d,x)} \Bigg)\gamma_\eta(d,x)=\left(y  -
m_{\eta}(d,x) \right)\gamma_\eta(d,x),
\end{align*}
which shows the results. 	
	
Now we check the pathwise differentiability of the ATE. To avoid the long display of various formulas, we consider the following decomposition
\begin{align*}
	\frac{\partial}{\partial t}\tau_{\eta_t}\Bigr\rvert_{t=0}=\frac{\partial }{\partial t}\int\mathbb{E}_{\eta_t}[Y|D=1,X=x]\mathrm{d}F_{\eta_t}(x)-\frac{\partial }{\partial t}\int\mathbb{E}_{\eta_t}[Y|D=0,X=x]\mathrm{d}F_{\eta_t}(x),
\end{align*}
and we focus on the first derivative involving the treatment group, as the other one can be handled analogously. We start with
\begin{align*}
	\frac{\partial }{\partial t}\mathbb{E}_{\eta_t}[\mathbb{E}_{\eta_t}[Y|D=1,X]]=\iint y\frac{\partial }{\partial t}p_{t}(y|1,x)f_t(x)\Bigr\rvert_{t=0}\mathrm{d}\nu(y)\mathrm{d}\mu(x),
\end{align*}
where $p_{t}(y|1,x)$ and $f_t(x)$ are the perturbed conditional density of outcome and marginal density of covariates, respectively. In addition, $\nu$ stands for the counting measure and $\mu$ is the Lebesgue measure. By the chain rule, we need to compute the following sum:
\begin{align}\label{TwoParts}
	\iint y\frac{\partial}{\partial t}p_{t}(y|1,x)\Bigr\rvert_{t=0}\mathrm{d}\nu(y)f_\eta(x)\mathrm{d}\mu(x)+
	\iint yp_\eta(y|1,x)\mathrm{d}\nu(y)\frac{\partial}{\partial t} f_t(x)\Bigr\rvert_{t=0}\mathrm{d}\mu(x).
\end{align}
Regarding the first part of the above sum, we follow the outline in Example 2 of \cite{levy2019tutorial} to compute 
\begin{align*}
	\frac{\partial}{\partial t}p_t(y|d,x)=\frac{\partial}{\partial t}\left[\prod_{j=0}^J m_{t,j}(d,  x)^{\mathbbm1\{y=j\}}\right]=\sum_{j=0}^J \mathbbm1_{\{y=j\}}\frac{\partial}{\partial t}m_{t,j}(d,  x)\prod_{k\neq j}m_{t,k}^{\mathbbm1\{y=k\}}(d,x).
\end{align*}
We thus evaluate for the derivatives of the conditional mean functions
\begin{align*}
	\frac{\partial}{\partial t}m_{t,j}(d,  x)\Bigr\rvert_{t=0}= m_{\eta, j}(d,  x)\left(\mathfrak{m}_j(d,x)-\sum_{l=1}^Jm_{\eta, l}(d,x)\mathfrak{m}_l(d,x)\right), ~~~\text{for}~~~j=1,\cdots,J,
\end{align*}
and 
\begin{align*}
	\frac{\partial}{\partial t}m_{t,0}(d,  x)\Bigr\rvert_{t=0}= m_{\eta, 0}(d,  x)\left(-\sum_{l=1}^Jm_{\eta,l}(d,x)\mathfrak{m}_l(d,x)\right).
\end{align*}
Thereafter, derivative of the conditional density can be written as
\begin{align*}
	\frac{\partial}{\partial t}p_t(y|d,x)\Bigr\rvert_{t=0}&=\left[\sum_{j=1}^J \mathbbm1_{\{y=j\}}\left(\mathfrak{m}_j(d,x)-\sum_{l=1}^Jm_{\eta,l}(d,x)\mathfrak{m}_l(d,x)\right)\right]\prod_{j=0}^J m_{\eta,j}(d,  x)^{\mathbbm1\{y=j\}}\\
	&=\left(B_\eta^{m}(\mathfrak{m}_1,\ldots,\mathfrak{m}_J) (z)-\mathbb{E}_{\eta}[B_\eta^{m}(\mathfrak{m}_1,\ldots,\mathfrak{m}_J) (Z)|D=d,X=x]\right) p_\eta(y|d,x),
\end{align*}
where the last equality follows from the fact that the conditional mean of the score given $(D,X)$ is zero. 
	To simplify the notation, we denote this conditional score function by
$	S_\eta(z)=B_\eta^{m}(\mathfrak{m}_1,\ldots,\mathfrak{m}_J) (z)$.
Referring to the first term in the summation \eqref{TwoParts}, we resort to the technique in Example 2 of \cite{levy2019tutorial} by converting the conditional argument from $d=1$ to  $d\in\{0,1\}$. Similar to the first two terms in the long display on Page 15 of \cite{levy2019tutorial}, we obtain 
	\begin{align*}
		\iint y\frac{\partial}{\partial t}p_{t}(y| 1,x)\Bigr\rvert_{t=0}\mathrm{d}\nu(y)f(x)\mathrm{d}\mu(x)
		=\mathbb{E}_\eta\left[\frac{D}{\pi_\eta(X)}(Y-m_\eta(D,X))S_\eta(Z)\right].
	\end{align*}
Referring to the second part of \eqref{TwoParts}, we immediately obtain
\begin{equation*}
	\iint yp_t(y|1,x)\mathrm{d}\nu(y)\frac{\partial}{\partial t}f_t(x)\Bigr\rvert_{t=0}\mathrm{d}\mu(x)=\mathbb{E}_\eta\left[\left(m_\eta(1,X)-\mathbb{E}_\eta[m_\eta(1,X)]\right)S_\eta(Z)\right]
\end{equation*}
Similarly for the control arm, we derive
	\begin{align*}
		&\frac{\partial}{\partial t}\int\mathbb{E}_{\eta_t}[Y|D=0,X=x]\mathrm{d}F_{\eta_t}(x)\Bigr\rvert_{t=0}\\
		=&\mathbb{E}_\eta\left[\left(m_\eta(0,X)-\mathbb{E}_\eta[m(0,X)]+\frac{1-D}{1-\pi_\eta(X)}(Y-m_\eta(D,X))\right)S_\eta(Z)\right].
	\end{align*}
The remaining part boils down to the existence of a vector-valued function $\tilde{\tau}_{P_\eta}$ such that
\begin{small}
	\begin{align*}
\frac{\partial}{\partial t}\tau(\eta_t)\Bigr\rvert_{t=0}&=	\mathbb{E}_{\eta}\left[ \tilde{\tau}_{\eta}(Z) B_\eta^{m}(\mathfrak{m}_1,\ldots,\mathfrak{m}_J) (Z)\right]\\
		&=\mathbb{E}_\eta\left[\left((\bar{m}_\eta(X)-\tau_\eta)+\left(\frac{D}{\pi_\eta(X)}-\frac{1-D}{1-\pi_\eta(X)}\right)(Y-m_\eta(D,X))\right)S_\eta(Z)\right].
	\end{align*}
\end{small}%
Consequently, we can take the solution as $\tilde{\tau}_{\eta}(z)=\bar{m}_\eta(x)-\tau_\eta+\gamma_\eta(d,x)(y-m_\eta(d,x))$, which concludes the proof.
\end{proof}

\section{Least Favorable Directions for Other Causal Parameters}\label{appendix:lfd}
In this part, we provide details on the least favorable directions for the first two examples in Section \ref{sec:other_param}. We properly address the binary outcome $Y$ and the reparameterization through the logistic type link function $\Psi(\cdot)$.
\subsection{Average Policy Effects}\label{appendix:lfd:ape}
The joint density of $Z_i=(Y_i,X_i)$ can be written as
\begin{equation}
	p_{m,f}(z)=m(x)^{y} (1-m(x))^{(1-y)}f(x).
\end{equation}
The observed data $Z_i$ can be described by $(m, f)$. It proves to be more convenient to consider the reparametrization of $(m, f)$ given by $\eta=(\eta^m,\eta^f)$, where
\begin{eqnarray}
	\eta^{m} = \Psi^{-1}(m),~~\eta^f =\log f. 
\end{eqnarray}
Consider the one-dimensional submodel $t\mapsto \eta_t$ defined by the path
\begin{eqnarray*}
	m_t(x) =  \Psi(\eta^m+t\mathfrak{m})(x),~~
	f_t(x)= f(x)e^{t\mathfrak{f}(x)}/\mathbb{E}[e^{t\mathfrak{f}(X)}],
\end{eqnarray*}
for the given direction $( \mathfrak{m},\mathfrak{f})$ with  $\mathbb{E}[\mathfrak{f}(X)]=0$. 
For this submodel, we further evaluate
\begin{align*}
	\log p_{\eta_t}(z) &=  y\log\Psi(\eta^m+t\mathfrak{m})(x) + (1-y)\log(1-\Psi(\eta^m+t\mathfrak{m}))(x)\\
	& + t\mathfrak{f}(x) - \log\mathbb{E}[e^{t\mathfrak{f}(X)}] + \log f(x).
\end{align*}	
Taking derivative with respect to $t$ and evaluating at $t=0$ gives the score operator:
\begin{equation}\label{score_pol}
	B_{\eta}(\mathfrak{m},\mathfrak{f})(Z)=  B_{\eta}^{m}\mathfrak{m}(Z) + B_{\eta}^{f}\mathfrak{f}(Z),
\end{equation}
where $B_{\eta}^{m}\mathfrak{m}(Z)=  (Y-m_\eta(X))\mathfrak{m}(X)$ and $B_{\eta}^{f}\mathfrak{f}(Z)= \mathfrak{f}(X)$. 

The efficient influence function for estimation of the policy effect parameter $\tau_\eta^{P}$ is given by
\begin{equation*}
	\widetilde{\tau}_{\eta}^{P}(z)=\gamma_\eta^{P}(x)(y-m_\eta(x))
\end{equation*}
where $\gamma_\eta^{P}(x)=\frac{g_1(x)-g_0(x)}{f(x)}$. Now the score operator $B_{\eta}$ given in \eqref{score_pol} applied to 
	$\xi_{\eta}^{P}(x)= 	\left(\gamma_\eta^{P}(x), 0\right)$,
yields $B_{\eta}\xi_{\eta}^{P}=\widetilde{\tau}_{\eta}^{P}$. Thus, $\xi_{\eta}^{P}$ defines the least favorable direction for this policy effect parameter. 
\subsection{Average Derivative}\label{appendix:lfd:ad}
The joint density of $Z_i=(Y_i,D_i,X_i)$  can be written as
\begin{equation}\label{x_den_conti}
	p_{m,f}(z)=m(d,x)^{y} (1-m(d,x))^{(1-y)}f(d,x).
\end{equation}
The observed data $Z_i$ can be described by  $(m, f)$. It proves to be more convenient to consider the reparametrization of $(m, f)$ given by $\eta=(\eta^m,\eta^f)$, where
\begin{eqnarray}\label{repar_conti}
	\eta^{m} = \Psi^{-1}(m),~~\eta^f =\log f. 
\end{eqnarray}
Consider the one-dimensional submodel $t\mapsto \eta_t$ defined by the path
\begin{eqnarray*}
	m_t(d,x) =  \Psi(\eta^m+t\mathfrak{m})(d,x),~~
	f_t(d,x)= f(d,x)e^{t\mathfrak{f}(d,x)}/\mathbb{E}[e^{t\mathfrak{f}(D,X)}],
\end{eqnarray*}
for the given direction $( \mathfrak{m},\mathfrak{f})$ with  $\mathbb{E}[\mathfrak{f}(D,X)]=0$. 
For this submodel defined in \eqref{repar_conti}, we further evaluate
\begin{align*}
	\log p_{\eta_t}(z) &=  y\log\Psi(\eta^m+t\mathfrak{m})(d,x) + (1-y)\log(1-\Psi(\eta^m+t\mathfrak{m}))(d,x)\\
	& + t\mathfrak{f}(d,x) - \log\mathbb{E}[e^{t\mathfrak{f}(D,X)}] + \log f(d,x).
\end{align*}	
Taking derivative with respect to $t$ and evaluating at $t=0$ gives the score operator:
\begin{equation}\label{score_ad}
	B_{\eta}(\mathfrak{m},\mathfrak{f})(Z)=  B_{\eta}^{m}\mathfrak{m}(Z) + B_{\eta}^{f}\mathfrak{f}(Z),
\end{equation}
where $B_{\eta}^{m}\mathfrak{m}(Z)=  (Y-m_\eta(D,X))\mathfrak{m}(D,X)$ and $B_{\eta}^{f}\mathfrak{f}(Z)= \mathfrak{f}(D,X)$. 
The efficient influence function for estimation of the AD parameter $\tau_\eta^{AD}=\mathbb{E}\left[\partial_dm_\eta(D,X)\right]$ is given by
\begin{equation*}
	\widetilde{\tau}_{\eta}^{AD}(z)=\partial_d m_\eta(d,x)-\mathbb{E}[\partial_d m_\eta(d,x)]+\gamma_\eta^{AD}(d,x)(y-m_\eta(d,x))
\end{equation*}
where 	$\gamma_\eta^{AD}(d,x)=\partial_d \pi_\eta(d,x)/\pi_\eta(d,x)$. Now the score operator $B_{\eta}$ given in \eqref{score_ad} applied to 
\begin{align*}
	\xi_{\eta}^{AD}(d,x)= 	\left(\gamma_\eta^{AD}(d,x), \partial_d m_\eta(d,x)-\mathbb{E}[\partial_d m_\eta(D,X)]\right),
\end{align*}
yields $B_{\eta}\xi_{\eta}^{AD}=\widetilde{\tau}_{\eta}^{AD}$. Thus, $\xi_{\eta}^{AD}$ defines the least favorable direction for the AD.

\section{Theory for One-parameter Exponential Family}\label{appendix:exponent}
We take $a=1$ in the exponential family for simplicity, that is, 
\begin{equation}\label{ConditionalPDF}
	f_{Y|D,X}(y\mid d,x) = c(y)\exp\left[q(m(d,x))y-A(m(d,x))\right],
\end{equation}
for some known functions $c(\cdot)$, $q(\cdot)$, and $A(\cdot)$.
We consider the reparametrization $	\eta^m(d,x)=q(m(d,x))
$ using the link function $q$
and we define the mapping $\Upsilon:= A\circ q^{-1}$, used below. Because the generalization of the binary outcome case to the above exponantial family involves some change of the likelihood function related to the conditional mean, we outline the necessary modifications. 

\begin{proposition}\label{prop:OnExponential}
Consider the one-parameter exponential family for the conditional distribution specified by \eqref{ConditionalPDF}. Let Assumption \ref{Ass:unconfounded} hold. Assume the function $\Upsilon$ is three time differentiable with $\Vert \Upsilon^{(\ell)}\Vert_{\infty}<\infty$ for $\ell=2,3$. The estimator $\widehat\gamma$ satisfies $\|\widehat\gamma\|_\infty=O_{P_0}(1)$ and $\|\widehat{\gamma}-\gamma_0\Vert_\infty= O_{P_0}\big((n/\log n)^{-s_\pi/(2s_\pi+p)}\big)$ for some $s_\pi>0$. Suppose $m_0(d,\cdot)\in \mathcal{C}^{s_m}([0,1]^{p})$ for $d\in\{0,1\}$ and some $s_m>0$ with $\sqrt{s_\pi \, s_m}>p/2$. Also, $\|\widehat{m}(d,\cdot)-m_0(d,\cdot)\Vert_{2, F_0}= O_{P_0}\big((n/\log n)^{-s_m/(2s_m+p)}\big)$. 
	Consider the propensity score-dependent prior on $m$ given by $m(d,x) = q^{-1}\left(W^m(d,x) + \lambda\,\widehat \gamma(d,x)\right)$, 
	where $W^m(d,x)$ is the rescaled squared exponential process for $d\in\{0,1\}$, with its rescaling parameter $a_n$ of the order in \eqref{RescaleRate},
	$\left(n/\log n\right)^{-s_m/(2s_m+p)}\lesssim u_n\sigma_n$ for some deterministic sequence $u_n\to 0$, and  $\sigma_n\lesssim 1$.
	Then, the posterior distribution satisfies Theorem \ref{thm:BvM}.
\end{proposition}
\begin{proof}
A close inspection shows that there are mainly three parts in which we need to adapt the argument due to the change of $p_{\eta^m}$ in the likelihood function. The first one is about the Kullback-Leibler (KL) divergence and related metrics to show the posterior contraction rate. The second one concerns the local asymptotic normality (LAN) expansion used in the conditional Laplace transform, where we show the connection of its second-order term to part of the variance of the influence function. Finally, we make use of the imposed smoothness assumption on $\Upsilon$ to show the negligibility of third order terms, which is needed for verifying the prior stability. We proceed in three steps. 

\textit{Step 1.}
First, in deriving the posterior contraction rate or determining the proper localized set $\mathcal{H}_n^m$, we need proper upper bounds for the Hellinger distance and Kullback-Leibler (KL) divergence between two probability density functions $(p_{\eta^m},p_{v^m})$ by the $L^2$ distance of the reparametrized functions $(\eta^m,v^m)$. Recall that $p_{\eta^m}(y,d,x)=c(y)\exp\left[q(m(d,x))ay-A(m(d,x))\right]$.
To abuse the notation a bit, we denote the corresponding probability densities by $p_{\eta^m}$ and $p_{v^m}$. 
From the proof of Lemma \ref{lemma:lfd_exp} we observe
\begin{align*}
	\mathbb E_{\eta}[Y|D=d,X=x]=\frac{(A'\circ q^{-1})(\eta^m(d,x))}{(q'\circ q^{-1})(\eta^m(d,x))}.
\end{align*}
Now the operator under consideration is $\Upsilon= A\circ q^{-1}$ and its derivative is given by $\Upsilon'= (A'\circ q^{-1})/(q'\circ q^{-1})$.
For the exponential family under consideration, the first and second order cumulants (conditional on covariates) are:
\begin{align*}
	\mathbb E_{\eta}[Y|D=d,X=x]=\Upsilon^{\prime}(\eta^m(d,x)),~~~Var_{\eta}(Y|D=d,X=x)=\Upsilon^{(2)}(\eta^m(d,x)).
\end{align*}
The conditional variance formula also shows the convexity of $\Upsilon(\cdot)$; see \cite{brown1986exp}. 

Considering the KL divergence $K(p_{\eta^m},p_{v^m})=\int \log\left( p_{\eta^m}(z)/p_{v^m}(z)\right)p_{\eta^m}(z)\mathrm{d}z$, we first compute
\begin{align*}
	\log \frac{p_{\eta^m}(z)}{p_{v^m}(z)}=(\eta^m(d,x)-v^m(d,x))y-[\Upsilon(\eta^m(d,x))-\Upsilon(v^m(d,x))].
\end{align*} 
Integrating over the conditional density for any given $(d,x)$ and utilizing the fact that the conditional mean is
$	m_{\eta}(d,x)=\Upsilon^\prime(\eta^m_{\eta}(d,x))$, we proceed for some intermediate value $\tilde{\eta}^m$:
\begin{align*}
	K(p_{\eta^m},p_{v^m})=&\int \left(\Upsilon^\prime(\eta^m(d,x))(\eta^m(d,x)-v^m(d,x))-[\Upsilon(\eta^m(d,x))-\Upsilon(v^m(d,x))]\right)\\
	&\qquad\qquad\qquad\qquad\times \pi^d(x)(1-\pi(x))^{1-d}\mathrm{d}\vartheta(d)\mathrm{d}F_{\eta}(x)\\ =&\int \Upsilon^{(2)}(\tilde{\eta}^m(d,x))(\eta^m(d,x)-v^m(d,x))^2\pi^d(x)(1-\pi(x))^{1-d}\mathrm{d}\vartheta(d)\mathrm{d}F_{\eta}(x)\\
	\lesssim& \Vert v^{m^1}-\eta^{m^1}\Vert^2_{2, F_\eta}\vee \Vert v^{m^0}-\eta^{m^0}\Vert^2_{2, F_\eta},
\end{align*}
where the last inequality follows from the condition $\Vert \Upsilon^{(2)}\Vert_{\infty}<\infty$. 
Recall that 
{\small \begin{equation}\label{Vbound}
	V(p_{\eta^m},p_{v^m})=\int\left[\log \frac{p_{\eta^m}(z)}{p_{v^m}(z)}- K(p_{\eta^m},p_{v^m})\right]^2p_{\eta^m}(z)\mathrm{d}z\leq \int\left[\log \frac{p_{\eta^m}(z)}{p_{v^m}(z)}\right]^2p_{\eta^m}(z)\mathrm{d}z.
\end{equation}}%
Therefore, we continue with the right hand side inequality of \eqref{Vbound} and calculate
\begin{small}
\begin{align*}
	&V(p_{\eta^m},p_{v^m})\\
	\leq&\int\left\{(\eta^m(d,x)-v^m(d,x))y-[\Upsilon(\eta^m(d,x))-\Upsilon(v^m(d,x))]\right\}^2 p_{\eta^m}(z)\mathrm{d}z\\
	=&\int(\eta^m(d,x)-v^m(d,x))^2[\Upsilon^{(2)}(\eta^m(d,x))+(\Upsilon^{\prime}(\eta^m(d,x)))^2]\pi^d(x)(1-\pi(x))^{1-d}\mathrm{d}\vartheta(d)\mathrm{d}F_{\eta}(x)\\
	-&2\int(\Upsilon(\eta^m(d,x))-\Upsilon(v^m(d,x)))(\eta^m(d,x)-v^m(d,x))\Upsilon^\prime(\eta^m(d,x))\pi^d(x)(1-\pi(x))^{1-d}\mathrm{d}\vartheta(d)\mathrm{d}F_{\eta}(x)\\
	+&\int(\Upsilon(\eta^m(d,x))-\Upsilon(v^m(d,x)))^2\pi^d(x)(1-\pi(x))^{1-d}\mathrm{d}\vartheta(d)\mathrm{d}F_{\eta}(x)\\
	=&\int \Upsilon^{(2)}(\eta^m(d,x))(\eta^m(d,x)-v^m(d,x))^2\pi^d(x)(1-\pi(x))^{1-d}\mathrm{d}\vartheta(d)\mathrm{d}F_{\eta}(x)\\
	+&\int\left\{(\eta^m(d,x)-v^m(d,x))\Upsilon^\prime(\eta^m(d,x))-[\Upsilon(\eta^m(d,x))-\Upsilon(v^m(d,x))]\right\}^2\pi^d(x)(1-\pi(x))^{1-d}\mathrm{d}\vartheta(d)\mathrm{d}F_{\eta}(x)\\
	\lesssim& \Vert v^{m^1}-\eta^{m^1}\Vert^2_{2, F_\eta}\vee \Vert v^{m^0}-\eta^{m^0}\Vert^2_{2, F_\eta},
\end{align*}
\end{small}
where in the first equality we have made use of the fact that 	$\mathbb{E}_{\eta}[Y^2|D=d,X=x]=\Upsilon^{(2)}(\eta^m(d,x))+(\Upsilon^{\prime}(\eta^m(d,x)))^2$. In sum, we have
\begin{align*}
	K(p_{\eta^m},p_{v^m})\vee V(p_{\eta^m},p_{v^m})\leq \Vert v^{m^1}-\eta^{m^1}\Vert^2_{2, F_\eta}\vee \Vert v^{m^0}-\eta^{m^0}\Vert^2_{2, F_\eta}.
\end{align*}
In addition, the squared Hellinger distance can be upper bounded by the KL divergence from Lemma B.1 in \cite{ghosal2017fundamentals}, so we have 
\begin{align}
	\Vert \sqrt{p_{v^m}}-\sqrt{p_{\eta^m}}\Vert_{L^2(\nu)}\leq \Vert v^{m^1}-\eta^{m^1}\Vert_{2, F_\eta}\vee \Vert v^{m^0}-\eta^{m^0}\Vert_{2, F_\eta}.
\end{align}
Because the posterior contraction holds in terms of the Hellinger distance under general conditions, the above upper bound allows us to translate this contraction to the corresponding $L_2$ distance for the $\eta^m$ function.

\textit{Step 2.}
We examine the changes to the LAN expansion in Lemma \ref{lemma:likelihood} as follows. 
For this purpose, we use the notation $g(u)=\log p_{\eta_u^m}$ for $u\in[0,1]$, as introduced at the beginning of Section \ref{appendix:key:lemmas}. Specifically, in the one-parameter exponential family case, we have 
\begin{equation*}
\log p_{\eta_u^m}(z)=y\eta_u^m(d,x)-\Upsilon(\eta_{u}^m(d,x))+\log c(y).
\end{equation*}
By the definition of $\Upsilon(\cdot)$, we know that it is a convex function, given that $\Upsilon^{(2)}(\eta^m(d,x))=Var_{\eta}(Y|D=d,X=x)$. 
Thereafter, we can obtain the first to third order derivatives of $g$ as
\begin{align*}
	g^\prime(0)&=\frac{t}{\sqrt{n}}\gamma_0\rho^{\Upsilon^\prime(\eta^m)}=\frac{t}{\sqrt{n}}\gamma_0\rho^{m_\eta},\\
	g^{(2)}(0)&=\frac{t^2}{n}\gamma^2_0\Upsilon^{(2)}(\eta^m), ~~~g^{(3)}(\tilde{u})=\frac{t^3}{n^{3/2}}\gamma^3_0\Upsilon^{(3)}(\eta_{\tilde{u}}^m),
\end{align*}
where $\tilde{u}$ is some intermediate value between $0$ and $1$. 

Following the lines of the proof of Lemma \ref{lemma:likelihood}, the second moment $P_0g^{(2)}(0)$ must be derived for the exponential family case.  Based on the previous calculations and the posterior convergence of $\eta^m$, it can be expressed as
\begin{align*}
	nP_0g^{(2)}(0)=&t^2\mathbb{E}_0[\gamma_0^2(D,X)\Upsilon^{(2)}(\eta_0^m(D,X))]+o_{P_0}(1)=t^2\mathbb{E}_0[\gamma_0^2(D,X)Var_0(Y|D,X)]+o_{P_0}(1)\\
	=&t^2\mathbb{E}_0[\gamma_0^2(D,X)(Y-m_0(D,X))]+o_{P_0}(1)=t^2P_0(B_0^m(\xi_0^m))^2+o_{P_0}(1),
\end{align*}
where the score operator $B_0^m=B_{\eta_0}^m$ is given in the proof of Lemma \ref{lemma:lfd_exp}. 

\textit{Step 3.}
Finally, we need to establish the following expansion in a key step to show the prior stability condition:
\begin{equation}
	\sup_{\eta^m\in \mathcal{H}_n^m}\left|\ell_n^m(\eta^m-t \gamma_n/\sqrt{n})-\ell_n^m(\eta^m-t \gamma_0/\sqrt{n})   \right|=o_{P_0}(1),
\end{equation}
where $\eta^m_{n,t}= \eta^m-t \gamma_n/\sqrt{n}$ and $\eta_t^m=\eta^m-t \gamma_0/\sqrt{n}$. Consider the following decomposition of the log-likelihood:
\begin{align*}
	\ell_n^m(\eta^m_{n,t})-\ell_n^m(\eta^m_{t})&=\ell_n^m(\eta^m_{n,t})-\ell_n^m(\eta^m)+\ell_n^m(\eta^m)-\ell_n^m(\eta^m_{t})\\
	&=n\mathbb{P}_n[\log p_{\eta^m_{n,t}}-\log p_{\eta^m}]+n\mathbb{P}_n[\log p_{\eta^m}-\log p_{\eta^m_{t}}].
\end{align*}
Then, we apply third-order Taylor expansions for the one-parameter exponential family separately to the two terms in the brackets of the above display:
\begin{align*}
	n\mathbb{P}_n[\log p_{\eta^m_{n,t}}-\log p_{\eta^m}]&=-t\sqrt{n}\mathbb{P}_n\left[\gamma_n\rho^{m_{\eta}}\right]-\frac{t^2}{2}\mathbb{P}_n\left[\gamma_n^2\Upsilon^{(2)}(\eta^m)\right]-\frac{t^3}{\sqrt n}\mathbb{P}_n\left[\gamma_n^3\Upsilon^{(3)}(\eta_{u^*}^m)\right],\\
	n\mathbb{P}_n[\log p_{\eta^m}-\log p_{\eta^m_{t}}]&=t\sqrt{n}\mathbb{P}_n\left[\gamma_0\rho^{m_{\eta}}\right]+\frac{t^2}{2}\mathbb{P}_n\left[\gamma_0^2\Upsilon^{(2)}(\eta^m)\right]+\frac{t^3}{\sqrt n}\mathbb{P}_n\left[\gamma_0^3\Upsilon^{(3)}(\eta_{u^{**}}^m)\right],
\end{align*}
for some intermediate points $u^{*}, u^{**}\in(0,1)$, cf. the equation (\ref{path}). The rest of the proof follows similar lines to our proof of Proposition \ref{prop:exponential}.
\end{proof}

\section{Posterior Computation in Algorithm \ref{algorithm} }\label{appendix:LPA}
We describe the Laplace approximation method used in Algorithm \ref{algorithm} (Posterior computation, Step (a)) for drawing the posterior of $\eta^m$ and thus $m_{\eta}$;
 see \citet[Chapters 3.3 to 3.5]{rassmusen2006gaussian} for more details on the Laplace approximation. Let $\boldsymbol W=[\boldsymbol D,\boldsymbol X]\in \mathbb{R}^{n\times (p+1)}$ be the matrix of $(D,X)$ in the data, $\boldsymbol W^* \in \mathbb{R}^{2n\times (p+1)}$ the evaluation points $(1,X)$ and $(0,X)$:
\begin{equation*}
\boldsymbol W^* = 
\begin{bmatrix}
\boldsymbol 1_n, &\boldsymbol X  \\
\boldsymbol 0_n, &\boldsymbol X \\
\end{bmatrix},
\end{equation*}
and $\boldsymbol \eta^*$ a $2n$-vector that gives the latent function $\eta^m$ evaluated at $\boldsymbol W^*$:
\begin{equation*}
\boldsymbol \eta^* = \left[\eta^m(1,X_1),\ldots,\eta^m(1,X_n),\eta^m(0,X_1),\ldots,\eta^m(0, X_n)\right]^{\top}.
\end{equation*}
Let $\boldsymbol \eta=\left[\eta^m(D_1,X_1),\dots,\eta^m(D_n, X_n)\right]^{\top}$ denote the $n$-vector of the latent function at  $\boldsymbol W$. For matrices $\boldsymbol W^*$ and $\boldsymbol W$, we define $K_c(\boldsymbol W^*,\boldsymbol W)$ as a $2n\times n$ matrix whose $(i,j)$-th element is $K_c(W_i^*,W_j)$, where $W_{i}^*$ is the $i$-th row of $\boldsymbol W^*$ and $W_j$ is the $j$-th row of $\boldsymbol W$. Analogously, $K_c(\boldsymbol W,\boldsymbol W)$ is an $n\times n$ matrix with the $(i,j)$-th element being $K_c(W_{i},W_j)$, and $K_c(\boldsymbol W^*,\boldsymbol W^*)$ is  a $2n\times 2n$ matrix with the $(i,j)$-th element being $K_c(W_{i}^*,W_{j}^*)$.

Given the mean zero GP prior with its covariance kernel $K_c$, the posterior of $\boldsymbol \eta^*$ is approximated by a Gaussian distribution with mean $\bar{\boldsymbol \eta}^*$ and covariance V$(\boldsymbol \eta^*)$ using the Laplace approximation. To be specific, let
\begin{eqnarray*}
\bar{\boldsymbol \eta}^*&=& K_c(\boldsymbol W^*,\boldsymbol W)K_c^{-1}(\boldsymbol W,\boldsymbol W)\,\widehat {\boldsymbol \eta},\\
\text{V}(\boldsymbol \eta^*) &=& K_c(\boldsymbol W^*,\boldsymbol W^*)-K_c(\boldsymbol W^*,\boldsymbol W)\left(K_c(\boldsymbol W,\boldsymbol W)+ \boldsymbol{\nabla}^{-1}\right)^{-1}K_c^\top(\boldsymbol W^*,\boldsymbol W),
\end{eqnarray*}
where $\widehat{\boldsymbol \eta}=\text{argmax}_{\boldsymbol\eta} p(\boldsymbol\eta|\boldsymbol W,\boldsymbol Y)$ maximizes the posterior $p(\boldsymbol\eta|\boldsymbol W,\boldsymbol Y)$ on the latent $\boldsymbol\eta$ and $\boldsymbol \nabla=-\frac{\partial^2\log p(\boldsymbol Y|\boldsymbol\eta)}{\partial\boldsymbol\eta\partial\boldsymbol\eta^\top}$ is a $n\times n$ diagonal matrix with the $i$-th diagonal entry being $-\frac{\partial^2\log p(\boldsymbol Y|\boldsymbol\eta)}{\partial\boldsymbol\eta_{i}^2}$.
We use the Matlab toolbox $\mathtt{GPML }$ for the implementation.\footnote{The $\mathtt{GPML}$ toolbox can be downloaded from http://gaussianprocess.org/gpml/code/matlab/doc/.}
In sum, we get the posterior draws of the vectors $\left[\eta^m(1,X_1),\ldots,\eta^m(1,X_n)\right]^{\top}$ and $\left[\eta^m(0,X_1),\ldots,\eta^m(0,X_n)\right]^{\top}$ from the above approximating Gaussian distribution with the mean $\bar{\boldsymbol \eta}^*$ and covariance V$(\boldsymbol \eta^*)$. 
We then obtain the posterior draws of the ATE by equation (\ref{debiased_bay_est}) through $m(d,X_i)=\Psi(\eta^m(d,X_i))$ for $d\in\{0,1\}$.

\section{Additional Simulation Results}\label{appendix:simu}
\renewcommand{\thetable}{A\arabic{table}}
    \setcounter{table}{0}
Appendix \ref{appendix:simu} presents additional simulation results for adjusted Bayesian inference methods. The design is the same as that in Section \ref{sec:simu}. Table \ref{tab:simu_C} evaluates the sensitivity of the finite sample performance with respect to the variance $\sigma_n$, which determines the influence strength of the prior correction term. We set $\sigma_n = c_{\sigma} \times \log  n/(\sqrt{n}\,\Gamma_n)$ with $c_{\sigma}\in\{1/5, 1/2, 1, 2, 5\}$.  Note that $c_{\sigma}=1$ corresponds to the simulation results reported in the main text. The performance of DR Bayes appears stable with respect to the choice of $c_{\sigma}$. The performance of PA Bayes, on the other hand, deteriorates when $\sigma_n$ takes relatively small or large values, such as the cases with $c_{\sigma}=1/5$, $t=0.10$ and $c_{\sigma}=5$, $t=0.01$. 

\begin{table}[H]
\centering
\caption{The effect of $c_{\sigma}$ on adjusted Bayesian inference methods: trimming based on the
estimated propensity score within $[t,1-t]$, $\bar n =$ the average sample size after trimming. CP = coverage probability, CIL = average length of the $95\%$ credible/confidence interval.}\label{tab:simu_C}
\vskip.15cm
{\footnotesize
 \begin{tabular}{clccccccccccccc}\toprule
\multicolumn{1}{c}{ $c_{\sigma}$}&\multicolumn{1}{c}{Methods}&\multicolumn{1}{c}{}& \multicolumn{1}{c}{Bias}&\multicolumn{1}{c}{CP}&\multicolumn{1}{c}{CIL}&\multicolumn{1}{c}{}&\multicolumn{1}{c}{Bias}& \multicolumn{1}{c}{CP}&\multicolumn{1}{c}{CIL}&\multicolumn{1}{c}{}&\multicolumn{1}{c}{Bias}& \multicolumn{1}{c}{CP}&\multicolumn{1}{c}{CIL}\\
\cline{1-2}\cline{4-6}\cline{8-10}\cline{12-14}
\multicolumn{1}{c}{ }&\multicolumn{1}{c}{}&\multicolumn{1}{c}{}&\multicolumn{3}{c}{$t = 0.10  (\bar n = 240)$}&\multicolumn{1}{c}{}& \multicolumn{3}{c}{$t =0.05 (\bar n =363)$}&\multicolumn{1}{c}{}&\multicolumn{3}{c}{$t = 0.01  (\bar n =664)$} \\
\cline{4-14}
$1/5$ & PA Bayes && -0.036 &    0.794  &    0.169 && -0.003 &    0.937  &    0.176 && 0.009 &    0.992  &    0.220\\
&DR Bayes  &&  -0.038 &    0.919 &    0.193 && -0.006 &    0.961 &    0.190 && 0.001 &    0.989 &    0.214   \\
\cline{1-2}\cline{4-14}
$1/2$ & PA Bayes &&  -0.024 &    0.962  &    0.215 &&0.016 &    0.968  &    0.222     &&0.032 &    0.968  &    0.289 \\
&DR Bayes  && -0.031 &    0.976 &    0.207  && 0.005 &    0.971 &    0.207 &&0.014 &    0.985 &    0.248   \\
\cline{1-2}\cline{4-14}
$1$ & PA Bayes &&  -0.008 &    0.981  &    0.260 && 0.033 &    0.949  &    0.254 &&0.047 &    0.897  &    0.308 \\
& DR Bayes  &&  -0.024 &    0.983 &    0.223 && 0.014 &    0.970 &    0.221   &&0.023 &    0.952 &    0.258  \\
\cline{1-2}\cline{4-14}
$2$ & PA Bayes &&   0.005 &    0.979  &    0.294 && 0.043 &    0.933  &    0.270 && 0.055 &    0.848  &    0.312\\
&DR Bayes  &&  -0.018 &    0.980 &    0.236&& 0.019 &    0.961 &    0.229&& 0.028 &    0.922 &    0.260 \\
\cline{1-2}\cline{4-14}
$5$ & PA Bayes &&  0.013 &    0.971  &    0.311 &&  0.047 &    0.928  &    0.276&& 0.058 &    0.836  &    0.313\\
&DR Bayes  &&  -0.015 &    0.980 &    0.242 && 0.022 &    0.958 &    0.232&& 0.030 &    0.907 &    0.261 \\
\bottomrule
\end{tabular}}
\end{table}

Table \ref{tab:simu_SS} reports the finite sample performance of PA and DR Bayes using sample splitting and compares it to the results in Table 1 that use the full sample twice. Sample splitting uses one half of the sample ($92$ treated and $1245$ control observations) to estimate the prior and posterior adjustments, and then draw the posterior of the conditional mean $m(d,x)$ using the other half of the sample ($93$ treated and $1245$ control observations).  The effective sample size $\bar n$ corresponds to the after-trimming size of the subsample used for drawing posteriors.  As Table \ref{tab:simu_SS} shows, DR Bayes using  sample splitting yields similar coverage probabilities to its counterpart in Table 1 that uses the full sample twice. The credible interval length increases as a result of halving the sample size.

\begin{table}[H]
\centering
\caption{Adjusted Bayesian inference methods using sample splitting: trimming based on the
estimated propensity score within $[t,1-t]$, $\bar n =$ the average sample size after trimming. CP = coverage probability, CIL = average length of the $95\%$ credible interval.
\qquad}\label{tab:simu_SS}
\vskip.15cm
{\footnotesize
 \begin{tabular}{lccccccccccccc}\toprule
\multicolumn{1}{c}{}&\multicolumn{1}{c}{}& \multicolumn{1}{c}{Bias}&\multicolumn{1}{c}{CP}&\multicolumn{1}{c}{CIL}&\multicolumn{1}{c}{}&\multicolumn{1}{c}{Bias}& \multicolumn{1}{c}{CP}&\multicolumn{1}{c}{CIL}&\multicolumn{1}{c}{}&\multicolumn{1}{c}{Bias}& \multicolumn{1}{c}{CP}&\multicolumn{1}{c}{CIL}\\
\cline{3-5}\cline{7-9}\cline{11-13}
\multicolumn{1}{l}{Sample splitting}&\multicolumn{1}{c}{}&\multicolumn{3}{c}{$t = 0.10  (\bar n = 124)$ }&\multicolumn{1}{c}{}& \multicolumn{3}{c}{$t =0.05 (\bar n =185)$}&\multicolumn{1}{c}{}&\multicolumn{3}{c}{$t = 0.01  (\bar n =339)$} \\
\cline{1-1}\cline{3-13}
 PA Bayes && -0.024 &    0.986  &    0.339 &&  0.009 &    0.977  &    0.342 && 0.017 &    0.950  &    0.410\\
 DR Bayes &&  -0.013 &    0.968  &    0.321 &&  0.017 &    0.962  &    0.317  &&0.016 &    0.934  &    0.385 \\
 \midrule
\multicolumn{1}{l}{Full sample}&\multicolumn{1}{c}{}&\multicolumn{3}{c}{$t = 0.10  (\bar n = 240)$ }&\multicolumn{1}{c}{}& \multicolumn{3}{c}{$t =0.05 (\bar n =363)$}&\multicolumn{1}{c}{}&\multicolumn{3}{c}{$t = 0.01  (\bar n =664)$} \\
\cline{1-1}\cline{3-13}
PA Bayes &&  -0.008 &    0.981  &    0.260 && 0.033 &    0.949  &    0.254 &&0.047 &    0.897  &    0.308 \\
DR Bayes  &&  -0.024 &    0.983 &    0.223 && 0.014 &    0.970 &    0.221   &&0.023 &    0.952 &    0.258\\
\bottomrule
\end{tabular}}
\end{table}

\bibliographystyle{abbrvnat}
\bibliography{Bayes_bib}

\end{document}